\documentclass[11pt]{article}

\usepackage[a4paper,margin=1in]{geometry}
\usepackage{amsmath,amssymb,amsfonts,amsthm,mathtools,bm}
\usepackage{enumitem}
\usepackage{graphicx}
\usepackage{adjustbox}
\graphicspath{{./}{figures/}}
\usepackage{booktabs}
\usepackage{longtable}
\usepackage{float}
\usepackage{algorithm}
\usepackage{algpseudocode}
\usepackage[natbibapa]{apacite}
\usepackage{hyperref}
\providecommand{\hypersetup}[1]{}
\hypersetup{
  colorlinks=true,
  linkcolor=blue,
  citecolor=blue,
  urlcolor=blue
}

\newtheorem{theorem}{Theorem}[section]
\newtheorem{proposition}{Proposition}[section]
\newtheorem{lemma}{Lemma}[section]
\newtheorem{corollary}{Corollary}[section]
\newtheorem{assumption}{Assumption}[section]
\newtheorem{remark}{Remark}[section]

\newcommand{\R}{\mathbb{R}}
\newcommand{\Sphere}{\mathbb{S}}
\newcommand{\Unif}{\mathrm{Unif}}
\newcommand{\Corr}{\mathrm{Corr}}
\newcommand{\Var}{\mathrm{Var}}
\newcommand{\Cov}{\mathrm{Cov}}
\newcommand{\E}{\mathbb{E}}
\newcommand{\Prob}{\mathbb{P}}
\newcommand{\tr}{\mathrm{tr}}
\newcommand{\diag}{\mathrm{diag}}
\newcommand{\perpn}{\mathrel{\perp\mspace{-10mu}\perp}}

\newcommand{\iid}{\stackrel{\mathrm{i.i.d.}}{\sim}}
\newcommand{\norm}[1]{\left\lVert #1\right\rVert}

\newcommand{\wideportraitfigure}[1]{%
  \includegraphics[width=\textwidth,height=0.70\textheight,keepaspectratio]{#1}%
}
\title{High-Dimensional Tests for Elliptical Models via Radial--Directional Dependence}
\author{Haoran Zhang and Long Feng\\
Nankai University}
\date{}

\begin{document}
\maketitle

\begin{abstract}
We develop high-dimensional goodness-of-fit tests for elliptical models by testing radial--directional independence after affine standardization.  The method forms coordinatewise correlations between the log-radius and directional components, using a sum statistic for dense departures, a max statistic for sparse departures, and a Cauchy combination for adaptation.  We derive oracle null limits, prove asymptotic independence of the sum and max components under both the null and a balanced local alternative, and establish validity of high-dimensional Hettmansperger--Randles plug-in standardization under explicit perturbation rates.  Simulations and data analyses show stable size control, dense--sparse power complementarity, and interpretable coordinate-level diagnostics.

\end{abstract}

\section{Introduction}
\label{sec:introduction}

Elliptical models occupy a central position in multivariate analysis.  They retain the affine geometry of the Gaussian family while allowing the radial distribution to be unspecified, thereby covering many heavy-tailed and semiparametric models used in robust statistics and shape estimation \citep{Maronna1976,Tyler1987,FangKotzNg1990}, portfolio theory and financial modeling \citep{OwenRabinovitch1983,GuptaVargaBodnar2013}, chemometric outlier detection \citep{RousseeuwDebruyneEngelenHubert2006}, and radar, array and covariance problems in signal processing \citep{OllilaTylerKoivunenPoor2012,SunBabuPalomar2016}.  In the formulation considered here, a random vector $\bm{X}\in\R^p$ is elliptically distributed if, for some location vector $\bm{\mu}$ and positive definite shape matrix $\mathbf{\Sigma}$, the standardized vector $\bm{Y}=\mathbf{\Sigma}^{-1/2}(\bm{X}-\bm{\mu})$ can be decomposed as
\[
        \bm{Y}=R\bm{U},
        \qquad
        \bm{U}\sim \Unif(\Sphere^{p-1}),
        \qquad
        R\perpn \bm{U}.
\]
Thus the null hypothesis is not a fully parametric Gaussian hypothesis; it is a symmetry and independence hypothesis about the radial and angular parts after unknown affine standardization.  This distinction is important because a valid test should not reject merely because the radial density is non-Gaussian or heavy-tailed.

Classical tests for spherical or elliptical symmetry were developed mainly for fixed dimension.  Early foundational work includes the ellipsoidal-symmetry test of \citet{Beran1979}, randomization and bootstrap ideas for nonparametric invariance hypotheses in \citet{Romano1989}, and rotationally invariant von Mises type tests for spherical symmetry in \citet{Baringhaus1991}.  Empirical-process and bootstrap approaches were further developed by \citet{KoltchinskiiSakhanenko2000}, with subsequent comparison and power studies in \citet{Sakhanenko2008}.  Other fixed-dimensional procedures use conditional Monte Carlo characterizations, moment or cumulant restrictions, chi-square partitions of directional cells, divergence criteria, empirical characteristic functions, or local asymptotic optimality arguments; representative examples include \citet{ZhuNeuhaus2000}, \citet{ManzottiPerezQuiroz2002}, \citet{Schott2002}, \citet{ZhuNeuhaus2003}, \citet{HufferPark2007}, \citet{BatsidisMartinPardoZografos2014}, \citet{AlbisettiBalabdaouiHolzmann2020} and \citet{BabicGelbgrasHallinLey2021}.  These works provide a rich low-dimensional theory and, in several cases, affine-invariant and distribution-free or bootstrap-calibrated procedures.  Their asymptotic justifications, however, typically keep $p$ fixed while $n$ grows, and many implementations rely on stable affine standardization, multidimensional partitions, or resampling schemes whose behavior changes substantially when $p$ is large.

The high-dimensional setting creates different statistical difficulties.  Sample covariance matrices may be singular when $p\ge n$, location and shape estimation errors can enter every coordinate of the standardized observations, and departures from ellipticity may be either dense, affecting many coordinates weakly, or sparse, affecting only a few directions strongly.  Although there is a broad literature on statistical inference under elliptical or heavy-tailed high-dimensional models \citep{feng2026high}, including robust shape and sphericity testing \citep{ZouPengFengWang2014}, high-dimensional Hettmansperger--Randles estimation \citep{YanFengZhang2025}, and robust sparse precision estimation \citep{LuFeng2025}, direct high-dimensional goodness-of-fit tests for the elliptical model are comparatively limited.  A notable recent contribution is \citet{WangLopes2026}, which develops a high-dimensional test for elliptical models based on low-dimensional random projections.  The present paper takes a different route: instead of comparing projected distributions, it directly tests a structural implication of ellipticity, namely the independence between radius and direction after standardization.

Our test is motivated by the observation that under the elliptical null, any scalar radial transform is independent of every directional coordinate.  We use the log-radius $L=\log R$ and form the vector of empirical coordinatewise correlations between $L$ and the components of $\bm{U}$.  Under the null, all these correlations are centered at zero.  Under alternatives, radial--directional dependence may be spread over many coordinates or concentrated in a small subset.  We therefore construct two complementary statistics: a sum statistic, which accumulates the squared coordinate correlations and is powerful for dense alternatives, and a max statistic, which detects the largest coordinatewise radial--directional association and is powerful for sparse alternatives.  A Cauchy combination p-value is then used to combine the sum and max components without needing to know the sparsity pattern in advance.

The main contributions of this paper are as follows.  First, we introduce a radial--directional moment-testing framework for high-dimensional elliptical goodness-of-fit problems with unspecified radial distribution, location and shape.  The resulting statistics are simple to compute and give coordinate-level diagnostic information through the largest radial--directional correlations.  Second, for the oracle standardized observations, we establish the null distribution theory in regimes where $p$ may diverge with $n$: the sum statistic is asymptotically normal, the max statistic has a Gumbel limit, and the two components are asymptotically independent.  Third, we prove that replacing the unknown location and shape by a high-dimensional Hettmansperger--Randles plug-in estimator preserves these first-order limits under explicit perturbation conditions.  This step separates the radial--directional testing argument from the robust standardization argument and makes clear what estimator accuracy is required.  Fourth, we provide a local-alternative result showing that the sum and max components remain asymptotically independent in a balanced detectable regime where both have nontrivial power.  Fifth, we provide a practical calibration and implementation scheme, together with simulations showing that the sum and max components have the expected dense-versus-sparse power complementarity and that their Cauchy combination is adaptive.  Finally, applications to near-infrared spectroscopy and mass-spectrometry data show how the proposed tests can identify localized radial--directional departures that are not always captured by existing high-dimensional elliptical tests.

The rest of the paper is organized as follows.  Section~\ref{sec:method} defines the oracle and plug-in radial--directional test statistics.  Section~\ref{sec:theory} gives the null theory for the sum, max and Cauchy-combination statistics and establishes plug-in validity.  Section~\ref{sec:simulation} reports simulation evidence under dense and sparse alternatives.  Section~\ref{sec:application} applies the method to spectroscopic and mass-spectrometry data.  Technical proofs and supplementary numerical results are collected in the appendices.

\paragraph{Notation.}
For a vector \(\bm{x}\), \(\|\bm{x}\|_2\), \(\|\bm{x}\|_\infty\) and \(\|\bm{x}\|_1\) denote its Euclidean, maximum and \(\ell_1\) norms.  For a matrix \(\mathbf{A}=(A_{jk})\), \(\|\mathbf{A}\|_{op}\), \(\|\mathbf{A}\|_F\), \(\|\mathbf{A}\|_{\max}\) and \(\|\mathbf{A}\|_{L_1}\) denote the spectral, Frobenius, entrywise maximum and maximum column-sum norms.  The trace, diagonal and smallest and largest eigenvalues of \(\mathbf{A}\) are denoted by \(\tr(\mathbf{A})\), \(\diag(\mathbf{A})\), \(\lambda_{\min}(\mathbf{A})\) and \(\lambda_{\max}(\mathbf{A})\).  The unit sphere in \(\R^p\) is \(\Sphere^{p-1}=\{\bm{u}\in\R^p:\|\bm{u}\|_2=1\}\).  For random variables, \(\E\), \(\Var\), \(\Cov\), \(\Corr\) and \(\Prob\) denote expectation, variance, covariance, correlation and probability.  For \(\alpha>0\), the Orlicz norm is \(\|W\|_{\psi_\alpha}=\inf\{c>0:\E\exp(|W|^\alpha/c^\alpha)\le2\}\).  The notation \(a_n\asymp b_n\) means that \(a_n/b_n\) is bounded away from zero and infinity.  We write \(O_p(a_n)\) and \(o_p(a_n)\) for stochastic order and negligible stochastic order, \(\xrightarrow{p}\) for convergence in probability, and \(\xrightarrow{d}\) for convergence in distribution.  Vectors are written in bold italic, such as \(\bm{X}\) and \(\bm{U}\), and matrices are written in bold roman, such as \(\mathbf{\Sigma}\) and \(\mathbf{\Omega}\).

\section{Radial--Directional Dependence Tests for Elliptical Models}
\label{sec:method}

\subsection{Elliptical null model}
\label{subsec:null-model}

Let $\bm{X}_1,\ldots,\bm{X}_n\in\R^p$ be independent observations from an unknown distribution.  The problem is to test whether this distribution belongs to an elliptical family with unspecified location and shape,
\[
H_0:\ \bm{X}\in\mathcal E_p
\qquad\text{against}\qquad
H_1:\ \bm{X}\notin\mathcal E_p .
\]
Here \(\mathcal E_p\) denotes the class of all $p$-dimensional elliptical distributions that admit the representation
\[
\begin{aligned}
\mathcal E_p
=
\bigl\{
&\mathcal L\bigl(\bm{\mu}+\mathbf{\Sigma}^{1/2}R\bm{U}\bigr):
\bm{\mu}\in\R^p,
\mathbf{\Sigma}\succ0,
\tr(\mathbf{\Sigma})=p,\\
& R>0,
\bm{U}\sim\Unif(\Sphere^{p-1}),
R\perpn\bm{U}
\bigr\}.
\end{aligned}
\]
Under $H_0$, there exist a location vector $\bm{\mu}\in\R^p$ and a positive definite shape matrix $\mathbf{\Sigma}$, normalized by $\tr(\mathbf{\Sigma})=p$, such that
\begin{equation}
\label{eq:elliptical-null}
\bm{X}_i=\bm{\mu}+\mathbf{\Sigma}^{1/2}\bm{Y}_i,
\qquad
\bm{Y}_i=R_i\bm{U}_i,
\qquad i=1,\ldots,n,
\end{equation}
where $R_i>0$ and
\begin{equation}
\label{eq:radial-directional-null}
\bm{U}_i\sim\Unif(\Sphere^{p-1}),
\qquad
R_i\perpn \bm{U}_i.
\end{equation}
Equivalently, after standardization by the location and shape parameters,
\begin{equation}
\label{eq:standardized-oracle}
\bm{Y}_i=\mathbf{\Sigma}^{-1/2}(\bm{X}_i-\bm{\mu}),
\qquad
R_i=\norm{\bm{Y}_i},
\qquad
\bm{U}_i=\bm{Y}_i/\norm{\bm{Y}_i}.
\end{equation}
Thus, radial--directional independence is a necessary structural feature of the elliptical null model; related spherical and elliptical symmetry tests also exploit the separation between radial and angular information \citep{AlbisettiBalabdaouiHolzmann2020}.  The proposed procedure tests this feature through high-dimensional correlations between a scalar radial transform and the directional coordinates.  Define
\[
L_i=\log R_i,
\qquad
m_L=\E(L_i),
\qquad
\sigma_L^2=\Var(L_i).
\]
For each coordinate $j=1,\ldots,p$, let
\[
\gamma_j=\Corr(L_i,U_{ij}).
\]
Under \eqref{eq:elliptical-null}--\eqref{eq:radial-directional-null}, $\gamma_j=0$ for every $j$.  The proposed test is based on the empirical evidence against this system of moment restrictions.  This formulation targets alternatives under which the radius and direction are dependent.  It does not require the departure from ellipticity to be expressed as a mean shift or as a change in the covariance structure.

\subsection{Oracle sum, max and Cauchy-combination statistics}
\label{subsec:oracle-statistics}

First suppose that $\bm{\mu}$ and $\mathbf{\Sigma}$ are known, so that $R_i$, $\bm{U}_i$ and $L_i$ are observed through \eqref{eq:standardized-oracle}.  Write
\[
\bar L=\frac1n\sum_{i=1}^nL_i,
\qquad
\bar U_j=\frac1n\sum_{i=1}^nU_{ij},
\]
and
\[
\hat\sigma_L^2=\frac1n\sum_{i=1}^n(L_i-\bar L)^2,
\qquad
\hat\sigma_{U,j}^2=\frac1n\sum_{i=1}^n(U_{ij}-\bar U_j)^2.
\]
The oracle coordinate correlation is
\[
\hat\gamma_j^{\rm or}
=
\frac{n^{-1}\sum_{i=1}^n(L_i-\bar L)(U_{ij}-\bar U_j)}{\hat\sigma_L\hat\sigma_{U,j}},
\qquad j=1,\ldots,p.
\]
Let
\[
\bm{g}_n^{\rm or}=(\hat\gamma_1^{\rm or},\ldots,\hat\gamma_p^{\rm or})^\top.
\]
The sum-type statistic is
\begin{equation}
\label{eq:sum-stat}
T_{\rm sum}=n\norm{\bm{g}_n^{\rm or}}_2^2=n\sum_{j=1}^p(\hat\gamma_j^{\rm or})^2,
\end{equation}
and the max-type statistic is
\begin{equation}
\label{eq:max-stat}
T_{\max}=n\norm{\bm{g}_n^{\rm or}}_\infty^2-2\log p+\log\log p.
\end{equation}
The sum component aggregates weak departures over many coordinates, whereas the max component is sensitive to a small number of large coordinate correlations.  This follows the max--sum principle used in high-dimensional testing, where sum-type and max-type statistics are designed for dense and sparse alternatives, respectively \citep{FengJiangLiuXiong2022,WangLiuFengMa2024}.

Let $\Phi$ denote the standard normal distribution function and define
\[
F_G(x)=\exp\{-\pi^{-1/2}e^{-x/2}\},
\qquad x\in\R.
\]
The asymptotic p-values associated with \eqref{eq:sum-stat} and \eqref{eq:max-stat} are
\begin{equation}
\label{eq:pvalue-sum-max}
P_{\rm sum}=1-\Phi\left(\frac{T_{\rm sum}-p}{\sqrt{2p}}\right),
\qquad
P_{\max}=1-F_G(T_{\max}).
\end{equation}
The adaptive statistic is obtained through the Cauchy combination
\begin{equation}
\label{eq:cauchy-stat}
T_{\rm cau}
=
\frac12\tan\{\pi(1/2-P_{\rm sum})\}
+
\frac12\tan\{\pi(1/2-P_{\max})\},
\end{equation}
with combined p-value
\begin{equation}
\label{eq:cauchy-pvalue}
P_{\rm cau}
=
\frac12-\frac1\pi\arctan(T_{\rm cau}).
\end{equation}
The test rejects the elliptical null hypothesis at level $\alpha$ when $P_{\rm cau}\le\alpha$.  The validity of this calibration is established in Section~\ref{sec:theory} through the asymptotic independence of the sum and max components.

\subsection{HR plug-in implementation}
\label{subsec:plugin-statistics}

In applications, $\bm{\mu}$ and $\mathbf{\Sigma}$ are unknown.  We estimate them by a high-dimensional HR procedure \citep{HettmanspergerRandles2002,YanFengZhang2025}.  It combines the spatial median, a spatial-sign graphical-lasso precision initializer, and an iterated HR location--shape update.  Throughout this subsection, for any nonzero vector $\bm{x}$, write
\[
\mathcal U(\bm{x})=\frac{\bm{x}}{\|\bm{x}\|_2}.
\]
The final scatter estimate is normalized by $\tr(\hat{\mathbf{\Sigma}})=p$, so it is on the same scale as the shape matrix in \eqref{eq:elliptical-null}.

The initial location is the spatial median
\begin{equation}
\label{eq:hr-spatial-median}
\hat{\bm{\mu}}^{(0)}
=
\arg\min_{\bm{m}\in\R^p}\sum_{i=1}^n\|\bm{X}_i-\bm{m}\|_2 .
\end{equation}
Using this location, form the spatial-sign covariance matrix
\[
\hat{\mathbf{S}}_0
=
\frac1n\sum_{i=1}^n
\mathcal U(\bm{X}_i-\hat{\bm{\mu}}^{(0)})
\mathcal U(\bm{X}_i-\hat{\bm{\mu}}^{(0)})^\top .
\]
The initial precision matrix is obtained from the spatial-sign graphical-lasso problem
\[
\hat{\mathbf{\Omega}}^{(0)}
=
\arg\min_{\mathbf{\Omega}\succ0}
\left[
\tr\{p\mathbf{\Omega}\hat{\mathbf{S}}_0\}
-
\log |\mathbf{\Omega}|
+
\lambda_n\|\mathbf{\Omega}\|_{1}
\right],
\]
where $\|\mathbf{\Omega}\|_1=\sum_{j,k}|\Omega_{jk}|$.  The initial shape estimate is
\[
\hat{\mathbf{\Sigma}}^{(0)}
=
\{\hat{\mathbf{\Omega}}^{(0)}\}^{-1},
\qquad
\hat{\mathbf{\Sigma}}^{(0)}
\leftarrow
p\hat{\mathbf{\Sigma}}^{(0)}/\tr\{\hat{\mathbf{\Sigma}}^{(0)}\}.
\]

For $t=0,1,2,\ldots$, the HR update is computed as follows.  First standardize the observations by the current location and shape:
\[
\hat{\bm{\varepsilon}}_i^{(t)}
=
\{\hat{\mathbf{\Sigma}}^{(t)}\}^{-1/2}
(\bm{X}_i-\hat{\bm{\mu}}^{(t)}),
\qquad
\hat{\bm{u}}_i^{(t)}
=
\mathcal U\{\hat{\bm{\varepsilon}}_i^{(t)}\}.
\]
The HR location equation $\sum_i\mathcal U\{\hat{\mathbf{\Sigma}}^{-1/2}(\bm{X}_i-\bm{m})\}=\bm{0}$ is solved by the fixed-point update
\[
\hat{\bm{\mu}}^{(t+1)}
=
\hat{\bm{\mu}}^{(t)}
+
\{\hat{\mathbf{\Sigma}}^{(t)}\}^{1/2}
\frac{n^{-1}\sum_{i=1}^n\hat{\bm{u}}_i^{(t)}}
     {n^{-1}\sum_{i=1}^n\|\hat{\bm{\varepsilon}}_i^{(t)}\|_2^{-1}}.
\]
Next compute the spatial-sign covariance matrix of the current residual directions,
\[
\hat{\mathbf{S}}_{\varepsilon}^{(t)}
=
\frac1n\sum_{i=1}^n
\hat{\bm{u}}_i^{(t)}\hat{\bm{u}}_i^{(t)\top}.
\]
For a matrix $\mathbf{A}=(A_{jk})_{1\le j,k\le p}$ and a bandwidth $h$, define
\[
\{B_h(\mathbf{A})\}_{jk}=A_{jk}\mathbf{1}(|j-k|\le h).
\]
The banded HR shape update is
\[
\tilde{\mathbf{\Sigma}}^{(t+1)}
=
 p\{\hat{\mathbf{\Sigma}}^{(t)}\}^{1/2}
B_h(\hat{\mathbf{S}}_{\varepsilon}^{(t)})
\{\hat{\mathbf{\Sigma}}^{(t)}\}^{1/2},
\qquad
\hat{\mathbf{\Sigma}}^{(t+1)}
=
\frac{p\mathcal P_+\{\tilde{\mathbf{\Sigma}}^{(t+1)}+\rho\mathbf{I}_p\}}
     {\tr\left[\mathcal P_+\{\tilde{\mathbf{\Sigma}}^{(t+1)}+\rho\mathbf{I}_p\}\right]},
\]
where $\mathcal P_+(\cdot)$ denotes projection onto the positive definite cone by replacing nonpositive eigenvalues by a small positive number, and $\rho>0$ is a ridge constant.  The iteration stops at the first $t$ for which
\begin{equation}
\label{eq:hr-stop-main}
\max\left\{
\|\hat{\bm{\mu}}^{(t+1)}-\hat{\bm{\mu}}^{(t)}\|_2,
\frac{\|\hat{\mathbf{\Sigma}}^{(t+1)}-\hat{\mathbf{\Sigma}}^{(t)}\|_F}
     {\max\{1,\|\hat{\mathbf{\Sigma}}^{(t)}\|_F\}}
\right\}
\le \varepsilon_{\rm HR}.
\end{equation}
In the numerical sections we use $\lambda_n=0.08$, $h=3$, $\rho=10^{-4}$, $\varepsilon_{\rm HR}=10^{-4}$ and at most $30$ HR iterations.

Algorithm~\ref{alg:hd-hr} summarizes the complete HR plug-in standardization used in the test.

\begin{algorithm}[H]
\caption{High-dimensional HR plug-in standardization}
\label{alg:hd-hr}
\begin{algorithmic}[1]
\Require Observations $\bm{X}_1,\ldots,\bm{X}_n\in\R^p$, graphical-lasso tuning $\lambda_n$, band width $h$, ridge $\rho$, tolerance $\varepsilon_{\rm HR}$ and maximum iteration number $K_{\max}$.
\Ensure HR estimates $\hat{\bm{\mu}}$, $\hat{\mathbf{\Sigma}}$ and standardized quantities $\{\hat{\bm{Y}}_i,\hat R_i,\hat{\bm{U}}_i,\hat L_i\}_{i=1}^n$.
\State Compute the spatial median
$\displaystyle
\hat{\bm{\mu}}^{(0)}=\arg\min_{\bm{m}\in\R^p}\sum_{i=1}^n\|\bm{X}_i-\bm{m}\|_2.
$
\State Form
$\displaystyle
\hat{\mathbf{S}}_0=n^{-1}\sum_{i=1}^n
\mathcal U(\bm{X}_i-\hat{\bm{\mu}}^{(0)})
\mathcal U(\bm{X}_i-\hat{\bm{\mu}}^{(0)})^\top.
$
\State Compute the sparse precision initializer.
\Statex \(\displaystyle
\hat{\mathbf{\Omega}}^{(0)}=\arg\min_{\mathbf{\Omega}\succ0}
\left[\tr\{p\mathbf{\Omega}\hat{\mathbf{S}}_0\}-\log|\mathbf{\Omega}|+
\lambda_n\|\mathbf{\Omega}\|_1\right].
\)
\State Set
$\displaystyle
\hat{\mathbf{\Sigma}}^{(0)}=\{\hat{\mathbf{\Omega}}^{(0)}\}^{-1},\qquad
\hat{\mathbf{\Sigma}}^{(0)}\leftarrow
p\hat{\mathbf{\Sigma}}^{(0)}/\tr\{\hat{\mathbf{\Sigma}}^{(0)}\}.
$
\For{$t=0,1,\ldots,K_{\max}-1$}
\State Compute residuals and directions
$\displaystyle
\hat{\bm{\varepsilon}}_i^{(t)}=
\{\hat{\mathbf{\Sigma}}^{(t)}\}^{-1/2}(\bm{X}_i-\hat{\bm{\mu}}^{(t)}),
\qquad
\hat{\bm{u}}_i^{(t)}=\mathcal U\{\hat{\bm{\varepsilon}}_i^{(t)}\}.
$
\State Update location
$\displaystyle
\hat{\bm{\mu}}^{(t+1)}=
\hat{\bm{\mu}}^{(t)}+
\{\hat{\mathbf{\Sigma}}^{(t)}\}^{1/2}
\frac{n^{-1}\sum_{i=1}^n\hat{\bm{u}}_i^{(t)}}
     {n^{-1}\sum_{i=1}^n\|\hat{\bm{\varepsilon}}_i^{(t)}\|_2^{-1}}.
$
\State Compute
$\displaystyle
\hat{\mathbf{S}}_{\varepsilon}^{(t)}=n^{-1}\sum_{i=1}^n
\hat{\bm{u}}_i^{(t)}\hat{\bm{u}}_i^{(t)\top}.
$
\State Update shape by banding, ridge correction and trace normalization
$\displaystyle
\tilde{\mathbf{\Sigma}}^{(t+1)}=
 p\{\hat{\mathbf{\Sigma}}^{(t)}\}^{1/2}
B_h(\hat{\mathbf{S}}_{\varepsilon}^{(t)})
\{\hat{\mathbf{\Sigma}}^{(t)}\}^{1/2},
$
$\displaystyle
\hat{\mathbf{\Sigma}}^{(t+1)}=
\frac{p\mathcal P_+\{\tilde{\mathbf{\Sigma}}^{(t+1)}+\rho\mathbf{I}_p\}}
     {\tr[\mathcal P_+\{\tilde{\mathbf{\Sigma}}^{(t+1)}+\rho\mathbf{I}_p\}]}.
$
\If{the stopping criterion in \eqref{eq:hr-stop-main} is satisfied}
\State \textbf{break}
\EndIf
\EndFor
\State Set $\hat{\bm{\mu}}=\hat{\bm{\mu}}^{(t+1)}$ and $\hat{\mathbf{\Sigma}}=\hat{\mathbf{\Sigma}}^{(t+1)}$.
\State Compute
$\displaystyle
\hat{\bm{Y}}_i=\hat{\mathbf{\Sigma}}^{-1/2}(\bm{X}_i-\hat{\bm{\mu}}),\quad
\hat R_i=\|\hat{\bm{Y}}_i\|_2,
\quad
\hat{\bm{U}}_i=\hat{\bm{Y}}_i/\|\hat{\bm{Y}}_i\|_2,
\quad
\hat L_i=\log\hat R_i.
$
\end{algorithmic}
\end{algorithm}

The fitted standardized observations are then
\begin{equation}
\label{eq:plugin-standardization}
\hat{\bm{Y}}_i=\hat{\mathbf{\Sigma}}^{-1/2}(\bm{X}_i-\hat{\bm{\mu}}),
\qquad
\hat R_i=\norm{\hat{\bm{Y}}_i},
\qquad
\hat{\bm{U}}_i=\hat{\bm{Y}}_i/\norm{\hat{\bm{Y}}_i},
\qquad
\hat L_i=\log\hat R_i.
\end{equation}
With
\[
\bar{\hat L}=\frac1n\sum_{i=1}^n\hat L_i,
\qquad
\bar{\hat U}_j=\frac1n\sum_{i=1}^n\hat U_{ij},
\]
and the empirical standard deviations $\hat\sigma_{\hat L}$ and $\hat\sigma_{\hat U,j}$, set
\begin{equation}
\label{eq:plugin-gammahat}
\hat g_{n,j}
=
\frac{n^{-1}\sum_{i=1}^n(\hat L_i-\bar{\hat L})(\hat U_{ij}-\bar{\hat U}_j)}{\hat\sigma_{\hat L}\hat\sigma_{\hat U,j}},
\qquad
\hat{\bm{g}}_n=(\hat g_{n,1},\ldots,\hat g_{n,p})^\top.
\end{equation}
The feasible statistics are
\[
\hat T_{\rm sum}=n\norm{\hat{\bm{g}}_n}_2^2,
\qquad
\hat T_{\max}=n\norm{\hat{\bm{g}}_n}_\infty^2-2\log p+\log\log p.
\]
The feasible p-values $\hat P_{\rm sum}$, $\hat P_{\max}$ and $\hat P_{\rm cau}$ are obtained from \eqref{eq:pvalue-sum-max}--\eqref{eq:cauchy-pvalue} with $T_{\rm sum}$ and $T_{\max}$ replaced by $\hat T_{\rm sum}$ and $\hat T_{\max}$.

\subsection{Finite-sample calibration}
\label{subsec:finite-sample-calibration}

The analytic feasible p-values in \eqref{eq:pvalue-sum-max} use the normal and Gumbel limits.  In finite samples, especially when the HR standardization is estimated in very high dimensions, an additional mean--variance calibration can be applied after HR standardization.  The calibration is carried out under the fitted elliptical null model and preserves the null separation between the radius and the direction.

Let $(\hat R_i,\hat{\bm{U}}_i)$ be the HR-standardized radii and directions from \eqref{eq:plugin-standardization}.  For $b=1,\ldots,B$, sample the radii with replacement and independently sample directions according to
\[
\hat R_1^{*(b)},\ldots,\hat R_n^{*(b)}\sim \widehat F_R,
\qquad
\bm{W}_1^{*(b)},\ldots,\bm{W}_n^{*(b)}\sim\Unif(\Sphere^{p-1}),
\]
where $\widehat F_R$ is the empirical distribution of $\hat R_1,\ldots,\hat R_n$.  The bootstrap standardized observations are
\[
\bm{Y}_i^{*(b)}=\hat R_i^{*(b)}\bm{W}_i^{*(b)},
\qquad i=1,\ldots,n.
\]

This radial--directional resampling mimics the representation $Y=RU$ with $R\perpn U$ and $U$ uniform on the sphere; it is closely related to the spherically symmetric bootstrap used in tests of spherical and elliptical symmetry \citep{AlbisettiBalabdaouiHolzmann2020}.  From each bootstrap sample compute $T_{{\rm sum}}^{*(b)}$ and $T_{\max}^{*(b)}$ by the same correlation formulas.  Let $\hat m_S,\hat s_S$ be the bootstrap mean and standard deviation of $T_{{\rm sum}}^{*(b)}$, and let $\hat m_M,\hat s_M$ be the corresponding quantities for $T_{\max}^{*(b)}$.

The bootstrap mean--variance corrected p-values are
\begin{equation}
\label{eq:boot-corrected-pvalues}
P_{\rm sum}^{\rm boot}
=
1-\Phi\left(\frac{\hat T_{\rm sum}-\hat m_S}{\hat s_S}\right),
\qquad
P_{\max}^{\rm boot}
=
1-F_G\left\{\mu_G+\frac{\sigma_G}{\hat s_M}(\hat T_{\max}-\hat m_M)\right\},
\end{equation}
where
\[
\mu_G=2\gamma_{\rm E}-\log\pi,
\qquad
\sigma_G^2=\frac{2\pi^2}{3},
\]
and $\gamma_{\rm E}$ is Euler's constant.  The corrected Cauchy p-value is obtained by replacing $P_{\rm sum}$ and $P_{\max}$ in \eqref{eq:cauchy-stat}--\eqref{eq:cauchy-pvalue} with $P_{\rm sum}^{\rm boot}$ and $P_{\max}^{\rm boot}$.  This calibration is used only to improve finite-sample centering and scaling; the theoretical reference distributions remain those established in Section~\ref{sec:theory}.

\section{Asymptotic Theory}
\label{sec:theory}

All limits in this section are taken along a sequence $n\to\infty$ and $p=p_n\to\infty$.  The following notation is used throughout.  Under the null model, the quantities $m_L=m_{L,p}$ and $\sigma_L^2=\sigma_{L,p}^2$ may depend on $p$.  Set
\[
Z_i=\frac{L_i-m_L}{\sigma_L},
\qquad
\rho_{L,p}=\sigma_L^{-1},
\qquad
\sigma_U^2=\Var(U_{ij})=p^{-1},
\]
\[
\xi_{ij}=\frac{(L_i-m_L)U_{ij}}{\sigma_L\sigma_U}=\sqrt p\,Z_iU_{ij}.
\]

\subsection{Regularity conditions}
\label{subsec:assumptions}

The regularity conditions are stated at two levels.  We first give the oracle radial--directional null model and then give the model-parameter conditions used for high-dimensional HR standardization.

\begin{assumption}
\label{ass:oracle}
The following conditions hold.
\begin{enumerate}[label=\textnormal{(\roman*)}]
\item $\bm{Y}_i=R_i\bm{U}_i$, $\bm{U}_i\sim\Unif(\Sphere^{p-1})$, and $R_i\perpn \bm{U}_i$.
\item The log-radius has a non-degenerate finite variance at each dimension,
\[
0<\sigma_L^2=\Var(\log R_i)<\infty,
\qquad
Z_i=\frac{\log R_i-m_L}{\sigma_L},
\qquad
m_L=\E(\log R_i),
\]
and there exist constants $\eta_Z>0$ and $C_Z<\infty$ such that
\[
\sup_{p\ge1}\E |Z_i|^{8+\eta_Z}\le C_Z.
\]
Moreover, for a deterministic sequence $b_{n,p}\ge 1$,
\[
\Pr\left(\max_{1\le i\le n}|Z_i|\le b_{n,p}\right)\to1,
\qquad
\frac{b_{n,p}^2\{(\log p)^5+\log p\log n\}}{n}\to0.
\]
\end{enumerate}
\end{assumption}

\begin{remark}
\label{rem:oracle-radial-condition}
Assumption~\ref{ass:oracle} is the oracle condition used for the null distribution theory.  Part \textnormal{(i)} is exactly the radial--directional representation of an elliptically standardized observation.  Part \textnormal{(ii)} is imposed on the standardized log-radius $Z_i$, not directly on $\log R_i-m_L$.  Hence $\sigma_L^2$ may vanish with $p$, as it does for concentrated high-dimensional radial laws.  For example, if $R_i^2\sim\chi_p^2$, then $\Var(\log R_i)=\psi_1(p/2)/4\asymp p^{-1}$.  The factor $\rho_{L,p}=\sigma_L^{-1}$ is therefore kept explicitly in all plug-in rates.  Appendix~\ref{app:radial-verification} verifies this condition and gives admissible deterministic choices of $b_{n,p}$ for common elliptical radial laws.
\end{remark}

Let
\[
\mathbf{\Omega}=\mathbf{\Sigma}^{-1},
\qquad
\bm{\varepsilon}_i=\mathbf{\Omega}^{1/2}(\bm{X}_i-\bm{\mu})=R_i\bm{U}_i,
\qquad
r_i=\|\bm{\varepsilon}_i\|_2=R_i,
\qquad
\zeta_k=\E(r_i^{-k}).
\]
For a matrix $\mathbf{A}=(A_{jk})$, write
\[
\|\mathbf{A}\|_{L_1}=\max_{1\le k\le p}\sum_{j=1}^p |A_{jk}|,
\qquad
\|\mathbf{A}\|_{\max}=\max_{1\le j,k\le p}|A_{jk}|.
\]

\begin{assumption}
\label{ass:hr-rate}
The high-dimensional HR estimator is computed by \eqref{eq:hr-spatial-median}--\eqref{eq:hr-stop-main}, and the following model-parameter conditions hold.
\begin{enumerate}[label=\textnormal{(\roman*)}]
\item There are constants $0<c_R<C_R<\infty$ and $K_R<\infty$ such that, for $k\in\{-1,1,2,3,4\}$,
\[
c_R\le \E\{(r_i/\sqrt p)^{-k}\}\le C_R,
\qquad
\norm{\zeta_1^{-1}r_i^{-1}}_{\psi_2}\le K_R.
\]
Consequently $\zeta_1\asymp p^{-1/2}$.
\item There are constants $\eta,h,d_0,D_0>0$ such that
\[
\eta\le \lambda_{\min}(\mathbf{\Sigma})
\le \lambda_{\max}(\mathbf{\Sigma})\le \eta^{-1},
\qquad
\tr(\mathbf{\Sigma})=p,
\qquad
\|\mathbf{\Sigma}\|_{L_1}\le h,
\]
where $\diag(\mathbf{\Sigma})=(d_1^2,\ldots,d_p^2)$ satisfies
\[
d_0\le \min_{1\le j\le p}d_j\le \max_{1\le j\le p}d_j\le D_0 .
\]
\item There are $T>0$, $0\le q<1$ and $s_0(p)>0$ such that
\[
\|\mathbf{\Omega}\|_{L_1}\le T,
\qquad
\max_{1\le j\le p}\sum_{k=1}^p |\Omega_{jk}|^q\le s_0(p).
\]
\item For
\[
\mathbf{S}=\E\{\mathcal U(\bm{X}_i-\bm{\mu})\mathcal U(\bm{X}_i-\bm{\mu})^\top\},
\]
there exists $\psi>0$ such that
\[
\limsup_{p\to\infty}\|\mathbf{S}\|_{op}\le 1-\psi.
\]
\item The graphical-lasso tuning parameter is
\[
\lambda_n
=
T\left\{
C_{\eta,T,1}\sqrt{\frac{\log p}{n}}
+C_{\eta,T,2}p^{-1/2}
\right\}
\]
for fixed positive constants $C_{\eta,T,1}$ and $C_{\eta,T,2}$, and
\[
\log p=o(n^{1/3}),
\qquad
\tau_n:=\lambda_n^{1-q}s_0(p),
\qquad
\tau_n\sqrt{\log p}\to0.
\]
\end{enumerate}
\end{assumption}

\begin{remark}
\label{rem:hr-rate-comment}
Assumption~\ref{ass:hr-rate} is a parameter-level condition for the HR estimator, not a direct assumption on the test statistics.  Conditions \textnormal{(i)}--\textnormal{(v)} collect the radial moment, shape regularity, precision-sparsity, spatial-sign and graphical-lasso tuning requirements used in the high-dimensional HR analysis of \citet{YanFengZhang2025}.  The perturbation expansion needed by the present tests is derived from Assumptions~\ref{ass:oracle}--\ref{ass:hr-rate} in Proposition~\ref{prop:yan-to-plugin}; the proof is given in Appendix~\ref{subsec:app-plugin-assumption}.
\end{remark}

Define
\[
\mathfrak c_n
=
 n^{-1/4}\sqrt{\log(np)}
 +n^{-(1-q)/2}(\log p)^{(1-q)/2}\sqrt{\log(np)}\,s_0(p)
 +p^{-(1-q)/2}\sqrt{\log(np)}\,s_0(p),
\]
\[
a_p=\E\left(\frac{L_i-m_L}{R_i}\right)+\frac1p\E\left(\frac{1-(L_i-m_L)}{R_i}\right),
\]
and
\begin{equation}
\label{eq:kappa-def}
\tilde\kappa_p=-(\sigma_L\sigma_U)^{-1}\zeta_1^{-1}a_p,
\qquad
\kappa_p=-(\sigma_L\sigma_U)^{-1}a_p.
\end{equation}
The testing rates are
\[
\mathfrak A_{S,n}
=
|\tilde\kappa_p|\left(n^{-1/2}+p^{-1/2}+\frac{\log p}{\sqrt n}\right)
+\frac{|\tilde\kappa_p|^2}{\sqrt p}
+|\tilde\kappa_p|\mathfrak c_n
+(1+\rho_{L,p})\left(\tau_n+\frac{\sqrt p}{n}\right)
+\frac{(1+\rho_{L,p})^2\tau_n^2}{\sqrt p},
\]
\[
\mathfrak A_{M,n}
=
\frac{|\tilde\kappa_p|\log p}{\sqrt p}
+\sqrt{\log p}\,\mathfrak c_n
+(1+\rho_{L,p})\left(\tau_n\log p+\frac{\log p}{\sqrt n}\right)
+(1+\rho_{L,p})^2\left(\tau_n^2\log p+\frac{\log p}{n}\right).
\]

For the plug-in expansion, let
\[
\bm{v}=\hat{\mathbf{\Sigma}}^{-1/2}(\hat{\bm{\mu}}-\bm{\mu}),
\qquad
\mathbf{B}=\hat{\mathbf{\Sigma}}^{-1/2}\mathbf{\Sigma}^{1/2}-\mathbf{I}_p,
\qquad
\bar{\bm{U}}=\frac1n\sum_{i=1}^n\bm{U}_i.
\]
Define
\[
\mathbf{\Delta}_n=\mathbf{M}_n-\E(\mathbf{M}_n),
\qquad
\mathbf{M}_n=\frac1n\sum_{i=1}^n\left[
\frac{L_i-m_L}{R_i}\mathbf{I}_p+
\frac{1-(L_i-m_L)}{R_i}\bm{U}_i\bm{U}_i^\top
\right],
\]
\[
\bm{r}_n^{(1)}=-(\sigma_L\sigma_U)^{-1}\mathbf{\Delta}_n\bm{v},
\]
\[
\bm{d}_n^{(\mathbf{\Sigma})}
=
\frac1{\sigma_L\sigma_U}\frac1n\sum_{i=1}^n
\left\{(L_i-m_L)(\mathbf{I}_p-\bm{U}_i\bm{U}_i^\top)\mathbf{B}\bm{U}_i
+\bm{U}_i(\bm{U}_i^\top \mathbf{B}\bm{U}_i)\right\},
\]
and
\[
\bm{b}_n^{\rm hard}=\bm{r}_n^{(1)}+\bm{d}_n^{(\mathbf{\Sigma})}.
\]

\begin{proposition}
\label{prop:yan-to-plugin}
Under Assumptions~\ref{ass:oracle}--\ref{ass:hr-rate},
\begin{equation}
\label{eq:hr-bahadur}
\bm{v}=\zeta_1^{-1}\bar{\bm{U}}+n^{-1/2}\bm{C}_n,
\qquad
\|\bm{C}_n\|_\infty=O_p(\mathfrak c_n),
\end{equation}
\[
\|\mathbf{B}\|_{op}=O_p(\tau_n),
\]
and
\[
\hat{\bm{g}}_n-\bm{g}_n^{\rm or}
=
\tilde\kappa_p\bar{\bm{U}}
+\kappa_p n^{-1/2}\bm{C}_n
+\bm{b}_n^{\rm hard}
+\bm{a}_n.
\]
The linearization remainder satisfies
\[
\frac n{\sqrt p}
\left\{ |(\bm{g}_n^{\rm or})^\top \bm{a}_n|
+|\bm{d}_n^\top \bm{a}_n|
+\|\bm{a}_n\|_2^2\right\}
=O_p(\mathfrak A_{S,n}),
\]
where
\[
\bm{d}_n
=\tilde\kappa_p\bar{\bm{U}}
+\kappa_p n^{-1/2}\bm{C}_n
+\bm{b}_n^{\rm hard},
\]
and
\[
\sqrt n\|\bm{a}_n\|_\infty=O_p(\mathfrak A_{M,n}/\sqrt{\log p}).
\]
Moreover,
\[
\frac n{\sqrt p}
\left\{|(\bm{g}_n^{\rm or})^\top \bm{b}_n^{\rm hard}|+
\|\bm{b}_n^{\rm hard}\|_2^2\right\}
=O_p\left((1+\rho_{L,p})\left\{\tau_n+\frac{\sqrt p}{n}\right\}+\frac{(1+\rho_{L,p})^2\tau_n^2}{\sqrt p}\right),
\]
\begin{equation}
\label{eq:derived-hard-max}
\sqrt n\|\bm{b}_n^{\rm hard}\|_\infty
=O_p\left((1+\rho_{L,p})\left\{\tau_n\sqrt{\log p}+\sqrt{\frac{\log p}{n}}\right\}\right).
\end{equation}
\end{proposition}

\subsection{Oracle null theory}
\label{subsec:oracle-theory}

\begin{theorem}
\label{thm:linearization}
Under Assumption~\ref{ass:oracle} and $\log p=o(\sqrt n)$.  Then $\gamma_j=0$ for all $j=1,\ldots,p$, and
\[
\hat\gamma_j^{\rm or}=\frac1n\sum_{i=1}^n\xi_{ij}+r_{n,j},
\qquad
\max_{1\le j\le p}|r_{n,j}|
=O_p\left(\frac{\log p+\sqrt{\log p}}n\right)
=o_p(n^{-1/2}).
\]
For each fixed $j$,
\[
\hat\gamma_j^{\rm or}=\frac1n\sum_{i=1}^n\xi_{ij}+O_p(n^{-1}).
\]
\end{theorem}

This result gives the first-order stochastic expansion of the coordinatewise radial--directional correlations.  Hence, after centering and self-normalization, the proposed problem reduces to the behavior of a high-dimensional vector of averages with nearly spherical coordinates.  The uniform remainder bound is used by both the quadratic statistic and the maximum statistic.

\begin{theorem}
\label{thm:sum-clt}
Under Assumption~\ref{ass:oracle}.  If
\[
p\to\infty,
\qquad
p=O(n^\kappa)
\quad\text{for some }\kappa\in(0,2),
\]
then
\[
\frac{T_{\rm sum}-p}{\sqrt{2p}}\xrightarrow{d} N(0,1).
\]
\end{theorem}

Theorem~\ref{thm:sum-clt} justifies the normal calibration of the sum statistic.  Since $T_{\rm sum}$ accumulates $p$ squared coordinate correlations, it is most sensitive when the radial--directional dependence is spread over many coordinates.  This is the dense regime in the sum--max terminology, as in high-dimensional tests that separate dense and sparse alternatives through sum and max components \citep{FengJiangLiuXiong2022}.

\begin{theorem}
\label{thm:max-gumbel}
Under Assumption~\ref{ass:oracle}.  If
\[
p\to\infty,
\qquad
\log p=o(n^{1/5}),
\]
then, for every $x\in\R$,
\[
\Prob(T_{\max}\le x)\to F_G(x)=\exp\{-\pi^{-1/2}e^{-x/2}\}.
\]
\end{theorem}

Theorem~\ref{thm:max-gumbel} gives the extreme-value calibration of the max statistic.  The centering $2\log p-\log\log p$ is the usual high-dimensional maximum correction, and the resulting test is designed to detect alternatives in which only a small number of directional coordinates are coupled with the radius.

\begin{theorem}
\label{thm:sum-max-independence}
Under Assumption~\ref{ass:oracle}.  If
\[
p\to\infty,
\qquad
p=O(n^\kappa)\quad\text{for some }\kappa\in(0,2),
\qquad
\log p=o(n^{1/5}),
\]
then
\[
\left(
\frac{T_{\rm sum}-p}{\sqrt{2p}},\ T_{\max}
\right)
\xrightarrow{d} (Z,G),
\]
where $Z\sim N(0,1)$, $G$ has distribution function $F_G$, and $Z$ and $G$ are independent.  Consequently, the oracle Cauchy p-value in \eqref{eq:cauchy-pvalue} satisfies
\[
P_{\rm cau}\xrightarrow{d} U(0,1).
\]
\end{theorem}

The preceding theorem is the key step for the adaptive procedure.  It shows that the large quadratic fluctuation and the extreme coordinate fluctuation separate asymptotically.  Therefore, the two p-values may be combined analytically, and the resulting Cauchy test inherits the dense sensitivity of the sum component and the sparse sensitivity of the max component.  This role parallels adaptive max--sum procedures for high-dimensional independence testing \citep{FengJiangLiuXiong2022,WangLiuFengMa2024}.

\subsection{HR plug-in null theory}
\label{subsec:plugin-theory}

\begin{theorem}
\label{thm:plugin-sum}
Under Assumptions~\ref{ass:oracle}--\ref{ass:hr-rate}.  If
\[
p\to\infty,
\qquad
p=O(n^\kappa)\quad\text{for some }\kappa\in(0,2),
\qquad
\log p=o(\sqrt n),
\qquad
\mathfrak A_{S,n}\to0,
\]
then
\[
\frac{\hat T_{\rm sum}-T_{\rm sum}}{\sqrt{2p}}=O_p(\mathfrak A_{S,n}).
\]
Consequently,
\[
\frac{\hat T_{\rm sum}-p}{\sqrt{2p}}\xrightarrow{d} N(0,1).
\]
\end{theorem}

Theorem~\ref{thm:plugin-sum} shows that estimating the unknown location and shape by the high-dimensional HR procedure does not affect the first-order null distribution of the sum statistic.  The rate $\mathfrak A_{S,n}$ displays the contributions of the HR location expansion, the spatial-sign graphical-lasso shape error, and the second-order Taylor remainder.

\begin{theorem}
\label{thm:plugin-max}
Under Assumptions~\ref{ass:oracle}--\ref{ass:hr-rate}.  If
\[
\log p=o(n^{1/5}),
\qquad
(\log p)^2=o(p),
\qquad
\mathfrak A_{M,n}\to0,
\]
then
\[
\hat T_{\max}-T_{\max}=O_p(\mathfrak A_{M,n}).
\]
Consequently, for every $x\in\R$,
\[
\Prob(\hat T_{\max}\le x)\to F_G(x)=\exp\{-\pi^{-1/2}e^{-x/2}\}.
\]
\end{theorem}

Theorem~\ref{thm:plugin-max} provides the corresponding stability result for the maximum statistic.  The rate $\mathfrak A_{M,n}$ is coordinatewise and therefore contains logarithmic factors from taking maxima over $p$ directions.

\begin{corollary}
\label{cor:plugin-cauchy}
Under the joint conditions of Theorems~\ref{thm:plugin-sum} and \ref{thm:plugin-max},
\[
\left(
\frac{\hat T_{\rm sum}-p}{\sqrt{2p}},\ \hat T_{\max}
\right)
\xrightarrow{d} (Z,G),
\]
where $Z\sim N(0,1)$, $G$ has distribution function $F_G$, and $Z\perpn G$.  Hence the feasible Cauchy p-value $\hat P_{\rm cau}$ satisfies
\[
\hat P_{\rm cau}\xrightarrow{d} U(0,1).
\]
\end{corollary}

Corollary~\ref{cor:plugin-cauchy} completes the null theory for the feasible test.  It permits the same analytic p-value formulas in Section~\ref{subsec:oracle-statistics} after replacing the oracle standardized observations by the HR-standardized observations.

\subsection{Power under radial--directional alternatives}
\label{subsec:alternatives}

We next study power under alternatives that preserve the uniform directional law but destroy radial--directional independence.  This distinction is important for the HR standardization: under a completely unrestricted non-elliptical alternative, the population HR location and shape need not equal the working pair $(\bm{\mu},\mathbf{\Sigma})$.  For the alternatives below, however, the direction remains uniform, so $(\bm{\mu},\mathbf{\Sigma})$ is still the population HR target even though the radius depends on the direction.

Let
\begin{equation}
\label{eq:alternative-baseline}
\bm{Y}_{0i}=R_{0i}\bm{U}_{0i},
\qquad
\bm{U}_{0i}\sim\Unif(\Sphere^{p-1}),
\qquad
R_{0i}\perpn \bm{U}_{0i},
\end{equation}
be an elliptical baseline vector after standardization.  For a nonempty active set $A\subset\{1,\ldots,p\}$, define
\[
s_A(\bm{u})=\frac1{\sqrt{|A|}}\sum_{j\in A}u_j,
\qquad \bm{u}\in\Sphere^{p-1}.
\]
For a signal strength parameter $\delta=\delta_n\ge0$, set
\begin{equation}
\label{eq:alternative-radius}
R_{1i}=R_{0i}\exp\{\delta_n s_A(\bm{U}_{0i})\},
\qquad
\bm{Y}_{1i}=R_{1i}\bm{U}_{0i}.
\end{equation}
The observed vector is
\begin{equation}
\label{eq:alternative-observed}
\bm{X}_i=\bm{\mu}+\mathbf{\Sigma}^{1/2}\bm{Y}_{1i}.
\end{equation}
When $\delta_n=0$, the baseline elliptical model is recovered.  When $\delta_n>0$, the conditional distribution of $R_{1i}$ depends on $\bm{U}_{0i}$, and the independence condition $R\perpn\bm{U}$ is violated.  The active set $A$ controls the sparsity pattern; the numerical study uses
\[
A_{\rm sp}=\{1\},
\qquad
A_{0.2}=\{1,\ldots,\lfloor0.2p\rfloor\},
\qquad
A_{\rm all}=\{1,\ldots,p\}.
\]

Write $s_n=|A|$, $L_{0i}=\log R_{0i}$, $L_{1i}=\log R_{1i}$, $m_{1,n}=\E L_{1i}$ and $\sigma_{1,n}^2=\Var(L_{1i})$.  Let
\[
\mathbf{A}_{1,n}
=
\E\{R_{1i}^{-1}(\mathbf{I}_p-\bm{U}_{0i}\bm{U}_{0i}^{\top})\},
\]
\[
\mathbf{M}_{1,n}
=
\E\left[ R_{1i}^{-1}
\left\{(L_{1i}-m_{1,n})\mathbf{I}_p+
[1-(L_{1i}-m_{1,n})]\bm{U}_{0i}\bm{U}_{0i}^{\top}\right\}\right],
\]
\[
\mathbf{H}_{1,n}=-(\sigma_{1,n}\sigma_U)^{-1}\mathbf{M}_{1,n}\mathbf{A}_{1,n}^{-1},
\qquad
\rho_{1,n}=\sigma_{1,n}^{-1},
\qquad
\sigma_U^2=p^{-1}.
\]
Define
\[
\mathfrak h_{2,n}=(\sigma_{1,n}\sigma_U)^{-1}\|\mathbf{M}_{1,n}\|_{op},
\qquad
\mathfrak h_{\infty,n}=(\sigma_{1,n}\sigma_U)^{-1}\|\mathbf{M}_{1,n}\|_{L_1},
\]
\[
\mathfrak H_{2,n}=\|\mathbf{H}_{1,n}\|_{op},
\qquad
\mathfrak H_{\infty,n}=\|\mathbf{H}_{1,n}\|_{L_1}.
\]
The HR perturbation rates under $H_1$ are
\[
\mathfrak R_{2,n}^{(1)}
=
\frac{\mathfrak H_{2,n}}{\sqrt n}
+\frac{\sqrt p\,\mathfrak h_{2,n}\mathfrak c_n}{\sqrt n}
+(1+\rho_{1,n})\left(\frac{\delta_n\tau_n}{\sqrt p}+\frac{\tau_n}{\sqrt n}\right)
+(1+\rho_{1,n})^2\left(\tau_n^2+\frac pn\right),
\]
\[
\mathfrak R_{\infty,n}^{(1)}
=
\mathfrak H_{\infty,n}\sqrt{\frac{\log p}{np}}
+\frac{\mathfrak h_{\infty,n}\mathfrak c_n}{\sqrt n}
+(1+\rho_{1,n})\left(\frac{\delta_n\tau_n}{\sqrt p}+\tau_n\sqrt{\frac{\log p}{n}}\right)
+(1+\rho_{1,n})^2\left(\tau_n^2+\frac{\log p}{n}\right).
\]

\begin{proposition}
\label{prop:h1-hr-behavior}
Under \eqref{eq:alternative-baseline}--\eqref{eq:alternative-observed}.  Suppose that the shape, precision-sparsity, spatial-sign and tuning conditions in Assumption~\ref{ass:hr-rate}\textnormal{(ii)}--\textnormal{(v)} hold, and that there are constants $c_1,C_1,K_1,C_\delta,\eta_1,C_{Z,1}>0$ such that, with
\[
Z_{1i}=\frac{L_{1i}-m_{1,n}}{\sigma_{1,n}},
\qquad
\rho_{1,n}=\sigma_{1,n}^{-1},
\qquad
0<\sigma_{1,n}^2<\infty,
\]
\[
\sup_p \E |Z_{1i}|^{8+\eta_1}\le C_{Z,1},
\qquad
\Pr\left(\max_{1\le i\le n}|Z_{1i}|\le b_{n,p}^{(1)}\right)\to1,
\]
\[
\frac{[b_{n,p}^{(1)}]^2\{(\log p)^5+\log p\log n\}}{n}\to0,
\qquad
\delta_n^2/p\le C_\delta,
\]
\[
c_1\le \E\{(R_{1i}/\sqrt p)^{-k}\}\le C_1,
\qquad
k\in\{-1,1,2,3,4\},
\qquad
\|\zeta_{1,n}^{-1}R_{1i}^{-1}\|_{\psi_2}\le K_1,
\]
where $\zeta_{1,n}=\E(R_{1i}^{-1})$.  Let $\hat{\bm{\mu}}$ and $\hat{\mathbf{\Sigma}}$ be computed by Algorithm~\ref{alg:hd-hr}.  Then
\[
\mathcal U\{\mathbf{\Sigma}^{-1/2}(\bm{X}_i-\bm{\mu})\}=\bm{U}_{0i},
\qquad
\E\bm{U}_{0i}=\bm{0},
\qquad
p\E(\bm{U}_{0i}\bm{U}_{0i}^{\top})=\mathbf{I}_p,
\]
so $(\bm{\mu},\mathbf{\Sigma})$ is the population HR location--shape target.  Moreover, with
\[
\bm{v}_1=\hat{\mathbf{\Sigma}}^{-1/2}(\hat{\bm{\mu}}-\bm{\mu}),
\qquad
\mathbf{B}_1=\hat{\mathbf{\Sigma}}^{-1/2}\mathbf{\Sigma}^{1/2}-\mathbf{I}_p,
\qquad
\bar{\bm{U}}_0=n^{-1}\sum_{i=1}^n\bm{U}_{0i},
\]
we have
\[
\bm{v}_1=\mathbf{A}_{1,n}^{-1}\bar{\bm{U}}_0+n^{-1/2}\bm{C}_{1,n},
\qquad
\|\bm{C}_{1,n}\|_\infty=O_p(\mathfrak c_n),
\qquad
\|\mathbf{B}_1\|_{op}=O_p(\tau_n).
\]
Let $\bm{g}_{n}^{(1),\rm or}$ be the oracle correlation vector computed from $(L_{1i},\bm{U}_{0i})$, and let $\hat{\bm{g}}_{n}^{(1)}$ be the HR plug-in vector computed from Algorithm~\ref{alg:hd-hr}.  Then
\[
\|\hat{\bm{g}}_{n}^{(1)}-\bm{g}_{n}^{(1),\rm or}\|_2
=O_p(\mathfrak R_{2,n}^{(1)}),
\qquad
\|\hat{\bm{g}}_{n}^{(1)}-\bm{g}_{n}^{(1),\rm or}\|_\infty
=O_p(\mathfrak R_{\infty,n}^{(1)}).
\]
\end{proposition}

\begin{theorem}
\label{thm:h1-consistency}
Under \eqref{eq:alternative-baseline}--\eqref{eq:alternative-observed}, Proposition~\ref{prop:h1-hr-behavior}, $p\to\infty$, $\log p=o(n^{1/5})$, and, with
\[
Z_{0i}=\frac{L_{0i}-\E L_{0i}}{\sigma_{0,n}},
\qquad
\sigma_{0,n}^2:=\Var(L_{0i}),
\qquad
0<\sigma_{0,n}^2<\infty,
\]
\[
\sup_p \E |Z_{0i}|^{8+\eta_0}\le C_{Z,0},
\qquad
\Pr\left(\max_{1\le i\le n}|Z_{0i}|\le b_{n,p}^{(0)}\right)\to1,
\]
\[
\frac{[b_{n,p}^{(0)}]^2\{(\log p)^5+\log p\log n\}}{n}\to0
\]
for some constants $\eta_0,C_{Z,0}>0$.  
Then
\[
\sigma_{1,n}^2=\sigma_{0,n}^2+\frac{\delta_n^2}{p},
\qquad
\gamma_j=
\Corr(L_{1i},U_{0ij})=
\frac{\delta_n\mathbf 1(j\in A)}{\sqrt{s_np}\,\sigma_{1,n}},
\]
\[
\|\bm{\gamma}\|_2^2=\frac{\delta_n^2}{p\sigma_{0,n}^2+\delta_n^2},
\qquad
\|\bm{\gamma}\|_\infty^2=
\frac{\delta_n^2}{s_n(p\sigma_{0,n}^2+\delta_n^2)}.
\]
Let
\[
\mathcal S_n=n\|\bm{\gamma}\|_2^2
=\frac{n\delta_n^2}{p\sigma_{0,n}^2+\delta_n^2},
\qquad
\mathcal M_n=n\|\bm{\gamma}\|_\infty^2
=\frac{n\delta_n^2}{s_n(p\sigma_{0,n}^2+\delta_n^2)}.
\]
For the oracle statistics, $P_{\rm sum}\xrightarrow{p}0$ if
\[
\frac{\mathcal S_n}{\sqrt p+p/\sqrt n}\to\infty,
\]
and $P_{\max}\xrightarrow{p}0$ if
\[
\frac{\mathcal M_n}{\log p}\to\infty.
\]
For the HR plug-in statistics, $\hat P_{\rm sum}\xrightarrow{p}0$ if
\[
\frac{\mathcal S_n}{\sqrt p+p/\sqrt n}\to\infty,
\qquad
\frac{n\{(\|\bm{\gamma}\|_2+\sqrt{p/n})\mathfrak R_{2,n}^{(1)}+[\mathfrak R_{2,n}^{(1)}]^2\}}{\mathcal S_n}\to0,
\]
and $\hat P_{\max}\xrightarrow{p}0$ if
\[
\frac{\mathcal M_n}{\log p}\to\infty,
\qquad
\frac{n\{(\|\bm{\gamma}\|_\infty+\sqrt{\log p/n})\mathfrak R_{\infty,n}^{(1)}+[\mathfrak R_{\infty,n}^{(1)}]^2\}}{\mathcal M_n}\to0.
\]
If at least one of the two plug-in p-values converges to zero and the other one is bounded away from one in probability, then $\hat P_{\rm cau}\xrightarrow{p}0$.
\end{theorem}

The preceding result covers consistency when the signal is strong.  The next theorem describes a balanced local regime in which the active set is neither fixed nor dense, and the signal is calibrated so that the sum and max components both have nontrivial limiting power.  This construction follows the max--sum adaptive testing principle used in high-dimensional factor testing by \citet{MaFengWangBao2024}, but here the signal enters through radial--directional correlations.

\begin{theorem}
\label{thm:h1-sum-max-independence}
Under \eqref{eq:alternative-baseline}--\eqref{eq:alternative-observed}, Proposition~\ref{prop:h1-hr-behavior}, and the moment and truncation conditions for $Z_{0i}$ stated in Theorem~\ref{thm:h1-consistency}.  Suppose that
\[
p\to\infty,
\qquad
\log p=o(n^{1/5}),
\qquad
p^{3/2}/n\to0.
\]
Let $A=A_n$, $s_n=|A_n|$, and assume that, for some $\ell\in(0,\infty)$,
\[
s_n\to\infty,
\qquad
s_n/p\to0,
\qquad
\frac{s_n\log p}{\sqrt p}\to\ell.
\]
For $m\ge2$ and $t\in\mathbb R$, define
\[
u_m(t)=\{2\log m-\log\log m+t\}^{1/2}.
\]
Assume that, for a fixed $\eta\in\mathbb R$,
\[
\mu_n
:=
\frac{\sqrt n\,\delta_n}{\sqrt{s_np}\,\sigma_{1,n}}
=u_p(0)-u_{s_n}(\eta),
\]
\[
\sigma_{1,n}^2=\sigma_{0,n}^2+\frac{\delta_n^2}{p}.
\]
Equivalently,
\[
\delta_n^2
=
\frac{s_np\sigma_{0,n}^2\{u_p(0)-u_{s_n}(\eta)\}^2}
{n-s_n\{u_p(0)-u_{s_n}(\eta)\}^2},
\]
where the denominator is positive for all sufficiently large $n$.  Then
\[
\left(
\frac{T_{\rm sum}-p}{\sqrt{2p}},\ T_{\max}
\right)
\xrightarrow{d}(Z+\theta_S,G_\eta),
\qquad
Z\perpn G_\eta,
\]
where $Z\sim N(0,1)$,
\[
\theta_S=\frac{(3-2\sqrt2)\ell}{\sqrt2},
\]
and $G_\eta$ has distribution function
\[
F_\eta(y)
=
\Prob(G_\eta\le y)
=
\exp\{-\lambda_0(y)-\lambda_1(y;\eta)\},
\]
with
\[
\lambda_0(y)=\pi^{-1/2}e^{-y/2},
\qquad
\lambda_1(y;\eta)=\frac1{2\sqrt\pi}\exp\left(-\frac\eta2-\frac{y}{2\sqrt2}\right).
\]
If additionally
\[
\frac{n\{(\|\bm{\gamma}\|_2+\sqrt{p/n})\mathfrak R_{2,n}^{(1)}+[\mathfrak R_{2,n}^{(1)}]^2\}}{\sqrt p}\to0,
\]
\[
n\{(\|\bm{\gamma}\|_\infty+\sqrt{\log p/n})\mathfrak R_{\infty,n}^{(1)}+[\mathfrak R_{\infty,n}^{(1)}]^2\}\to0,
\]
then
\[
\left(
\frac{\hat T_{\rm sum}-p}{\sqrt{2p}},\ \hat T_{\max}
\right)
\xrightarrow{d}(Z+\theta_S,G_\eta),
\qquad
Z\perpn G_\eta.
\]
Consequently, for any fixed $\alpha\in(0,1)$,
\[
\Prob(P_{\rm sum}\le\alpha)\to
1-\Phi\{z_{1-\alpha}-\theta_S\}\in(\alpha,1),
\]
\[
\Prob(P_{\max}\le\alpha)\to
1-F_\eta\{F_G^{-1}(1-\alpha)\}\in(\alpha,1),
\]
and the same two limits hold for $\hat P_{\rm sum}$ and $\hat P_{\max}$ under the two plug-in remainder conditions above.
\end{theorem}

\section{Simulation Studies}
\label{sec:simulation}

This section evaluates the finite-sample performance of the feasible HR plug-in tests at nominal level $\alpha=0.05$.  Throughout the simulations, $n=200$, the reported entries are rejection percentages based on 500 Monte Carlo replications, and the rejection rules are $\hat P_{\rm sum}\le0.05$, $\hat P_{\max}\le0.05$ and $\hat P_{\rm cau}\le0.05$.  Tables use the abbreviations S, M and C for the HR-Sum, HR-Max and HR-Cauchy tests.  The competing high-dimensional elliptical-model test of \citet{WangLopes2026} is denoted by W or WL according to table width.  The simulation design mirrors the theory in Section~\ref{subsec:alternatives}: the sum component targets dense radial--directional departures, the max component targets sparse departures, and the Cauchy p-value is used when the active-set size is unknown \citep{FengJiangLiuXiong2022,WangLiuFengMa2024}.

\subsection{Empirical size}
\label{subsec:simulation-size}

The size experiments are generated from the elliptical null model in \eqref{eq:elliptical-null}.  We set $\bm{\mu}=\bm 0$ and generate
\[
\bm{X}_i=\mathbf{\Sigma}^{1/2}\bm{Y}_{0i},
\qquad
\bm{Y}_{0i}=R_{0i}\bm{U}_{0i},
\qquad
\bm{U}_{0i}\sim\Unif(\Sphere^{p-1}),
\qquad
R_{0i}\perpn\bm{U}_{0i}.
\]
Five radial laws are considered.  For the Gaussian baseline, $\bm{Y}_{0i}=\bm{Z}_i$ with $\bm{Z}_i\sim N_p(\bm 0,\mathbf{I}_p)$, so $R_{0i}=\|\bm{Z}_i\|_2$.  For the $t_{10}$ baseline,
\[
\bm{Y}_{0i}=\frac{\bm{Z}_i}{\sqrt{G_i/10}},
\qquad
G_i\sim\chi^2_{10},
\qquad
G_i\perpn\bm{Z}_i.
\]
For the mixture-normal baseline, $\bm{Y}_{0i}=S_i\bm{Z}_i$, where $\Prob(S_i=1)=0.9$ and $\Prob(S_i=3)=0.1$.  Two additional light-tailed elliptical laws are generated directly from $R_{0i}\bm{U}_{0i}$.  For the Kotz-type power-exponential law,
\[
R_{0i}^{\mathrm{raw}}=(2G_i)^{1/(2\beta)},
\qquad
G_i\sim\Gamma\left(\frac{p}{2\beta},1\right),
\qquad
\beta=2,
\]
and $R_{0i}^{\mathrm{raw}}$ is multiplied by a deterministic constant so that $\E R_{0i}^2=p$.  For the bounded-radial law,
\[
R_{0i}^2=2pB_i,
\qquad
B_i\sim {\rm Beta}(p/2,p/2),
\]
which gives $\E R_{0i}^2=p$.  These five baselines include both concentrated log-radius laws, such as Gaussian, Kotz-type and bounded-radial laws, and non-concentrated radial laws, such as $t_{10}$ and mixture-normal laws.  Hence the size study is aligned with Assumption~\ref{ass:oracle}\textnormal{(ii)}, where $\sigma_L^2=\Var(\log R_i)$ is allowed to depend on $p$.

The outer shape matrix $\mathbf{\Sigma}$ is chosen from three structures.  The identity design is $\mathbf{\Sigma}^{\mathrm{I}}=\mathbf{I}_p$.  The autoregressive design is $\mathbf{\Sigma}^{\mathrm{AR}}_{jk}=0.3^{|j-k|}$.  The sparse-precision design first constructs a sparse positive-definite matrix $\mathbf{\Omega}^{\mathrm{SP}}$ and then sets $\mathbf{\Sigma}^{\mathrm{SP}}=(\mathbf{\Omega}^{\mathrm{SP}})^{-1}$.  Each shape matrix is normalized by $p\mathbf{\Sigma}/\tr(\mathbf{\Sigma})$ so that $\tr(\mathbf{\Sigma})=p$, matching the normalization used in \eqref{eq:elliptical-null} and in the HR algorithm.

Tables~\ref{tab:size-main-low} and \ref{tab:size-main-high} report empirical sizes in percentages.  Table~\ref{tab:size-main-low} contains all five radial laws for $p=50,100,200$.  Table~\ref{tab:size-main-high} focuses on the largest dimension $p=400$ and reports the three light-tailed nulls.  The displayed entries use the analytic null calibration and do not apply the finite-sample bootstrap correction in Section~\ref{subsec:finite-sample-calibration}.  The complete $p=400$ size results for the two non-concentrated radial laws are given in Appendix~\ref{app:simulation-tables}.

\begin{table}[H]
\centering
\caption{Empirical sizes (\%) of HR-Sum, HR-Max, HR-Cauchy and WL at nominal level $\alpha=0.05$ with $n=200$ and $p\le 200$. The entries use the analytic null calibration without bootstrap correction.}
\label{tab:size-main-low}
\renewcommand{\arraystretch}{1.05}
\setlength{\tabcolsep}{3.2pt}
\resizebox{0.65\textwidth}{!}{%
\begin{tabular}{lccccccccc}
\toprule
 & \multicolumn{3}{c}{$p=50$} & \multicolumn{3}{c}{$p=100$} & \multicolumn{3}{c}{$p=200$} \\
\cmidrule(lr){2-4}\cmidrule(lr){5-7}\cmidrule(lr){8-10}
Test & $\mathbf{I}_p$ & AR & SP & $\mathbf{I}_p$ & AR & SP & $\mathbf{I}_p$ & AR & SP \\
\midrule
\multicolumn{10}{c}{Gaussian} \\
HR-Sum & 4.4 & 5.2 & 4.8 & 5.4 & 3.0 & 4.2 & 5.2 & 4.4 & 5.8 \\
HR-Max & 3.6 & 4.0 & 4.4 & 4.0 & 4.6 & 3.8 & 4.2 & 6.4 & 6.0 \\
HR-Cauchy & 4.8 & 5.8 & 5.0 & 4.6 & 4.8 & 5.4 & 5.6 & 5.8 & 5.8 \\
WL & 5.8 & 5.4 & 4.6 & 5.4 & 4.4 & 5.2 & 5.8 & 5.8 & 6.6 \\
\addlinespace[1pt]
\multicolumn{10}{c}{Kotz $\beta=2$} \\
HR-Sum & 5.8 & 5.8 & 5.8 & 5.4 & 5.8 & 5.2 & 6.2 & 5.2 & 5.0 \\
HR-Max & 4.4 & 4.6 & 4.0 & 4.4 & 4.8 & 3.6 & 4.6 & 5.0 & 3.4 \\
HR-Cauchy & 5.4 & 5.6 & 5.6 & 5.8 & 6.4 & 4.0 & 6.0 & 5.4 & 5.0 \\
WL & 7.0 & 3.6 & 5.6 & 6.0 & 8.0 & 9.2 & 5.0 & 4.2 & 4.6 \\
\addlinespace[1pt]
\multicolumn{10}{c}{Bounded radial} \\
HR-Sum & 5.8 & 3.8 & 6.0 & 5.0 & 5.8 & 4.6 & 5.4 & 5.0 & 5.0 \\
HR-Max & 4.4 & 5.0 & 5.6 & 4.4 & 5.6 & 4.8 & 5.4 & 4.8 & 4.8 \\
HR-Cauchy & 5.4 & 5.8 & 6.2 & 5.4 & 6.0 & 5.0 & 5.0 & 5.0 & 6.0 \\
WL & 6.2 & 5.8 & 6.0 & 7.0 & 4.4 & 5.4 & 4.6 & 6.2 & 7.6 \\
\addlinespace[1pt]
\multicolumn{10}{c}{Mixture normal} \\
HR-Sum & 5.2 & 4.2 & 7.0 & 6.4 & 6.8 & 6.6 & 7.0 & 6.6 & 7.4 \\
HR-Max & 5.4 & 4.8 & 4.4 & 3.8 & 3.0 & 4.8 & 5.0 & 6.2 & 4.0 \\
HR-Cauchy & 6.0 & 5.2 & 6.0 & 5.4 & 6.0 & 5.8 & 6.2 & 6.0 & 4.6 \\
WL & 5.6 & 4.4 & 5.2 & 4.8 & 6.8 & 4.8 & 5.6 & 4.2 & 3.8 \\
\addlinespace[1pt]
\multicolumn{10}{c}{$t_{10}$} \\
HR-Sum & 6.6 & 5.2 & 5.2 & 5.8 & 5.8 & 4.8 & 8.2 & 7.6 & 8.6 \\
HR-Max & 4.8 & 5.0 & 4.2 & 4.4 & 5.4 & 5.0 & 4.2 & 3.6 & 4.8 \\
HR-Cauchy & 5.2 & 5.6 & 4.6 & 6.0 & 5.6 & 5.4 & 5.8 & 6.2 & 6.4 \\
WL & 3.4 & 5.0 & 4.6 & 5.0 & 5.2 & 4.8 & 4.8 & 5.0 & 9.0 \\
\bottomrule
\end{tabular}}
\end{table}

\begin{table}[H]
\centering
\caption{Empirical sizes (\%) of HR-Sum, HR-Max, HR-Cauchy and WL at nominal level $\alpha=0.05$ with $n=200$ and $p=400$ for light-tailed elliptical nulls. The entries use the analytic null calibration without bootstrap correction.}
\label{tab:size-main-high}
\renewcommand{\arraystretch}{1.05}
\setlength{\tabcolsep}{6pt}
\resizebox{0.4\textwidth}{!}{%
\begin{tabular}{lccc}
\toprule
Test & $\mathbf{I}_p$ & AR & SP \\
\midrule
\multicolumn{4}{c}{Gaussian} \\
HR-Sum & 4.6 & 4.2 & 4.0 \\
HR-Max & 4.4 & 4.4 & 5.0 \\
HR-Cauchy & 4.8 & 4.8 & 4.0 \\
WL & 6.0 & 3.8 & 5.8 \\
\addlinespace[1pt]
\multicolumn{4}{c}{Kotz $\beta=2$} \\
HR-Sum & 4.4 & 5.6 & 4.8 \\
HR-Max & 5.2 & 3.2 & 4.8 \\
HR-Cauchy & 4.6 & 4.6 & 5.4 \\
WL & 4.6 & 7.0 & 4.2 \\
\addlinespace[1pt]
\multicolumn{4}{c}{Bounded radial} \\
HR-Sum & 4.0 & 3.6 & 4.4 \\
HR-Max & 4.6 & 4.6 & 3.2 \\
HR-Cauchy & 4.4 & 4.6 & 4.6 \\
WL & 6.4 & 4.2 & 5.4 \\
\bottomrule
\end{tabular}}
\end{table}

The size results support two aspects of the proposed procedure.  First, the HR plug-in standardization keeps the tests close to the nominal level across identity, autoregressive and sparse-precision shapes, even though the shape matrix is unknown and may be high dimensional.  Second, the calibration remains stable for radial laws with very different log-radius behavior.  In particular, the Gaussian, Kotz-type and bounded-radial rows have $\sigma_L^2\asymp p^{-1}$, whereas the mixture-normal and $t_{10}$ rows have non-vanishing radial variability.  This empirical pattern is consistent with the theory, which standardizes $\log R_i$ by its own $\sigma_L$ rather than requiring a fixed lower bound on $\Var(\log R_i)$.  A useful point is that the proposed tests do not require the radial law to be estimated: only the centered and scaled log-radius enters the statistic.  Thus the same analytic calibration works for concentrated radial laws, where the raw log-radius variance shrinks with $p$, and for heavy-tailed or scale-mixture laws, where the radial variability remains visible.  The rows for AR and SP shapes also indicate that the HR standardization removes the main affine effect before the radial--directional diagnostic is applied.

\subsection{Empirical power}
\label{subsec:simulation-power}

The power experiments use the radial--directional alternatives in \eqref{eq:alternative-baseline}--\eqref{eq:alternative-observed}.  Starting from the same elliptical baseline $\bm{Y}_{0i}=R_{0i}\bm{U}_{0i}$ as in the size study, we set
\[
R_{1i}=R_{0i}\exp\{\delta_n s_A(\bm{U}_{0i})\},
\qquad
\bm{X}_i=\mathbf{\Sigma}^{1/2}R_{1i}\bm{U}_{0i},
\]
where
\[
A_{\rm sp}=\{1\},
\qquad
A_{0.2}=\{1,\ldots,\lfloor0.2p\rfloor\},
\qquad
A_{\rm all}=\{1,\ldots,p\}.
\]
The three choices are the sparse, moderately dense and fully dense cases.  The table column $\delta_n$ is the signal strength in \eqref{eq:alternative-radius}.  We display the Gaussian and $t_{10}$ baselines, which contrast a concentrated light-tailed log-radius law with a heavy-tailed law.  Complete numerical tables and the remaining radial-law figures are reported in Appendices~\ref{app:simulation-tables} and \ref{app:power-figures}.

\begin{table}[H]
\centering
\caption{Empirical powers (\%) under \eqref{eq:alternative-radius} with $n=200$, averaged over the three shape structures. For each dimension, S, M, C and W denote tests based on $\hat P_{\rm sum}$, $\hat P_{\max}$, $\hat P_{\rm cau}$ and WL, respectively.}
\label{tab:power-main}
\scriptsize
\renewcommand{\arraystretch}{1.05}
\setlength{\tabcolsep}{1.9pt}
\resizebox{0.75\textwidth}{!}{%
\begin{tabular}{llrrrrrrrrrrrr}
\toprule
 & & \multicolumn{4}{c}{$p=100$} & \multicolumn{4}{c}{$p=200$} & \multicolumn{4}{c}{$p=400$} \\
\cmidrule(lr){3-6}\cmidrule(lr){7-10}\cmidrule(lr){11-14}
Active set & $\delta_n$ & S & M & C & W & S & M & C & W & S & M & C & W \\
\midrule
\multicolumn{14}{c}{Gaussian} \\
$A_{\rm sp}$ & 1 & 100 & 100 & 100 & 13 & 100 & 100 & 100 & 12 & 97 & 100 & 100 & 10 \\
$A_{\rm sp}$ & 2 & 100 & 100 & 100 & 55 & 100 & 100 & 100 & 42 & 100 & 100 & 100 & 31 \\
$A_{\rm sp}$ & 3 & 100 & 100 & 100 & 80 & 100 & 100 & 100 & 70 & 100 & 100 & 100 & 58 \\
\addlinespace[1pt]
$A_{0.2}$ & 1 & 100 & 92 & 100 & 14 & 100 & 43 & 100 & 10 & 86 & 14 & 80 & 12 \\
$A_{0.2}$ & 2 & 100 & 100 & 100 & 55 & 100 & 64 & 100 & 40 & 97 & 19 & 95 & 29 \\
$A_{0.2}$ & 3 & 100 & 100 & 100 & 80 & 100 & 68 & 100 & 71 & 99 & 20 & 97 & 59 \\
\addlinespace[1pt]
$A_{\rm all}$ & 1 & 100 & 31 & 100 & 12 & 100 & 13 & 100 & 12 & 86 & 6 & 80 & 10 \\
$A_{\rm all}$ & 2 & 100 & 44 & 100 & 52 & 100 & 15 & 100 & 41 & 97 & 7 & 94 & 30 \\
$A_{\rm all}$ & 3 & 100 & 48 & 100 & 80 & 100 & 16 & 100 & 72 & 99 & 6 & 96 & 60 \\
\addlinespace[1pt]
\multicolumn{14}{c}{$t_{10}$} \\
$A_{\rm sp}$ & 1 & 60 & 97 & 97 & 6 & 25 & 59 & 59 & 5 & 15 & 17 & 21 & 3 \\
$A_{\rm sp}$ & 2 & 100 & 100 & 100 & 12 & 82 & 100 & 100 & 8 & 34 & 94 & 94 & 4 \\
$A_{\rm sp}$ & 3 & 100 & 100 & 100 & 30 & 99 & 100 & 100 & 13 & 62 & 100 & 100 & 5 \\
\addlinespace[1pt]
$A_{0.2}$ & 1 & 58 & 16 & 51 & 6 & 25 & 6 & 20 & 5 & 13 & 5 & 11 & 2 \\
$A_{0.2}$ & 2 & 100 & 64 & 100 & 15 & 77 & 17 & 71 & 7 & 29 & 6 & 23 & 3 \\
$A_{0.2}$ & 3 & 100 & 90 & 100 & 30 & 98 & 29 & 97 & 15 & 57 & 8 & 48 & 4 \\
$A_{0.2}$ & 4 & -- & -- & -- & -- & -- & -- & -- & -- & 78 & 11 & 71 & 9 \\
$A_{0.2}$ & 5 & -- & -- & -- & -- & -- & -- & -- & -- & 91 & 13 & 86 & 14 \\
\addlinespace[1pt]
$A_{\rm all}$ & 1 & 58 & 13 & 51 & 6 & 24 & 6 & 20 & 5 & 15 & 5 & 12 & 3 \\
$A_{\rm all}$ & 2 & 100 & 28 & 100 & 12 & 76 & 11 & 68 & 7 & 29 & 6 & 24 & 3 \\
$A_{\rm all}$ & 3 & 100 & 39 & 100 & 28 & 99 & 13 & 97 & 13 & 60 & 8 & 51 & 5 \\
$A_{\rm all}$ & 4 & -- & -- & -- & -- & -- & -- & -- & -- & 79 & 8 & 71 & 9 \\
$A_{\rm all}$ & 5 & -- & -- & -- & -- & -- & -- & -- & -- & 90 & 8 & 85 & 14 \\
\bottomrule
\end{tabular}}
\end{table}

Table~\ref{tab:power-main} displays the expected separation between the test components.  Under $A_{\rm sp}$, the signal is concentrated in one coordinate and the max statistic is the most direct detector; under $A_{0.2}$ and $A_{\rm all}$, the signal is distributed over many coordinates and the sum statistic is more effective.  The Cauchy test tracks the stronger component across these regimes, which is the main practical advantage of combining S and M when the analyst does not know the active-set size.  The comparison with WL also illustrates the difference between the two testing principles: WL is a broad projection-based elliptical goodness-of-fit test, whereas the proposed tests are constructed to detect radial--directional dependence and therefore have substantially higher power for the alternatives in \eqref{eq:alternative-radius}.  For the heavy-tailed $t_{10}$ baseline, larger signal values are needed at $p=400$ because radial variability is stronger, but the powers of S and C still increase steadily for dense alternatives and the max component remains effective for sparse alternatives.  These patterns are the finite-sample counterpart of the signal measures $\mathcal S_n$ and $\mathcal M_n$: $\mathcal S_n$ increases when many weak correlations accumulate, while $\mathcal M_n$ increases when at least one coordinate carries a strong correlation.  The Cauchy combination is therefore not merely a numerical aggregation device; it implements the theoretical sum--max complementarity and provides a default procedure when no reliable prior information about sparsity is available.

\begin{figure}[htbp]
\centering
\wideportraitfigure{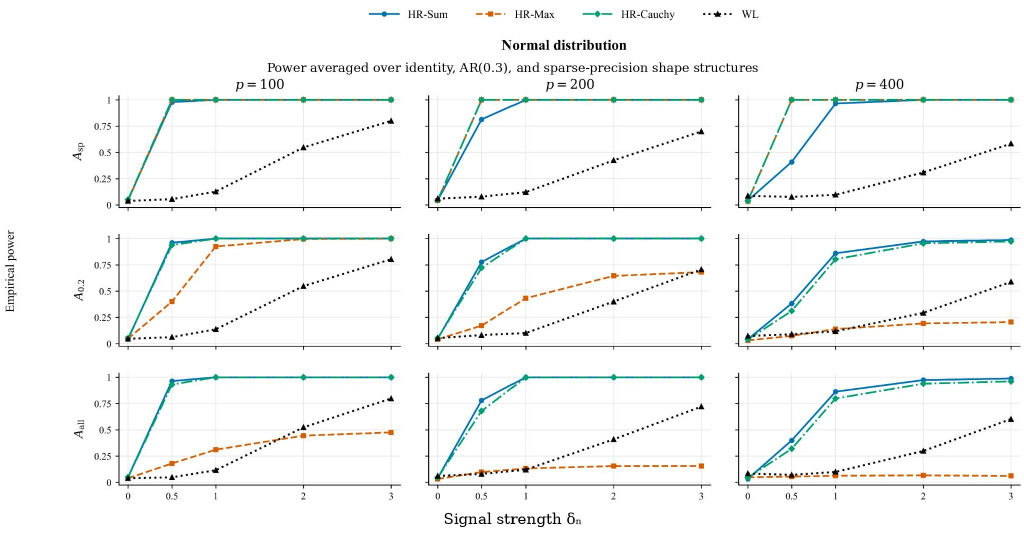}
\caption{Empirical power curves for the Gaussian baseline, averaged over $\mathbf{\Sigma}^{\mathrm{I}}$, $\mathbf{\Sigma}^{\mathrm{AR}}$ and $\mathbf{\Sigma}^{\mathrm{SP}}$. Rows correspond to $A_{\rm sp}$, $A_{0.2}$ and $A_{\rm all}$; columns correspond to $p=100,200,400$.}
\label{fig:power-main-normal}
\end{figure}

\begin{figure}[htbp]
\centering
\wideportraitfigure{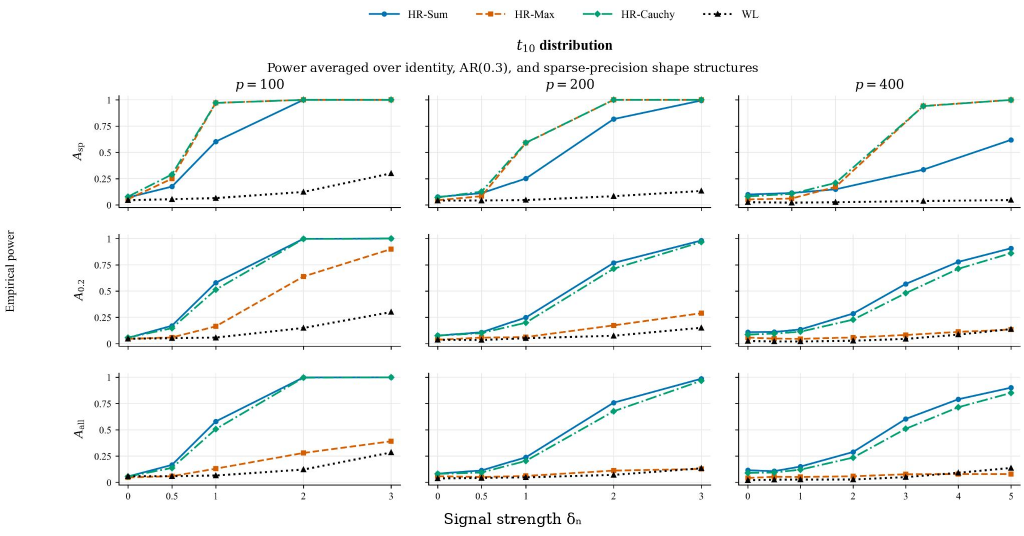}
\caption{Empirical power curves for the $t_{10}$ baseline, averaged over $\mathbf{\Sigma}^{\mathrm{I}}$, $\mathbf{\Sigma}^{\mathrm{AR}}$ and $\mathbf{\Sigma}^{\mathrm{SP}}$. Rows correspond to $A_{\rm sp}$, $A_{0.2}$ and $A_{\rm all}$; columns correspond to $p=100,200,400$. For $p=400$, the displayed signal-strength range extends to $\delta_n=5$.}
\label{fig:power-main-t10}
\end{figure}

Figures~\ref{fig:power-main-normal} and \ref{fig:power-main-t10} give the corresponding graphical summaries for the Gaussian and $t_{10}$ baselines.  The curves make the dense-versus-sparse complementarity visually clear: M rises fastest in the first row, S rises fastest in the lower rows, and C avoids committing to either regime.  The heavier radial tail in the $t_{10}$ baseline flattens the curves relative to the Gaussian baseline, which agrees with the signal-to-noise quantities $\mathcal S_n$ and $\mathcal M_n$ in Theorem~\ref{thm:h1-consistency}.  The figures also show that increasing $p$ makes the sparse and dense regimes look different in a practically meaningful way: sparse signals remain visible through the maximum coordinate, whereas dense signals become easier for S once many coordinates contribute.  This is precisely the diagnostic information that a single omnibus projection statistic does not provide.

\section{Real Data Applications}
\label{sec:application}

\subsection{Gasoline spectroscopy data}
\label{subsec:gasoline-data}

The first application uses the gasoline near-infrared spectroscopy data set included in the R package \texttt{pls} \citep{MevikWehrens2007,Kalivas1997}.  The analysis uses only the spectral matrix \texttt{gasoline\$NIR}; the octane response is not used.  The resulting sample size and dimension are
\[
n=60,
\qquad
p=401.
\]
The 401 wavelengths range from 900nm to 1700nm with spacing 2nm.

For each analysis window, the corresponding wavelength columns are extracted and standardized within that window.  If $X_{ij}$ denotes the raw intensity at wavelength $j$, the standardized value is
\[
\widetilde X_{ij}=\frac{X_{ij}-\bar X_j}{s_j},
\]
where $\bar X_j$ and $s_j$ are the sample mean and sample standard deviation computed using only the observations and wavelengths in the current window.  Window-specific standardization prevents scale differences across wavelength regions from dominating the test statistics.

The analysis is conducted at three resolutions.  The full-spectrum analysis uses all wavelengths from 900nm to 1700nm.  The broad-window analysis divides the spectrum into four windows: A1, 900--1098nm; A2, 1100--1298nm; A3, 1300--1498nm; and A4, 1500--1700nm.  The refined-window analysis focuses on selected local windows in the A1/A2 region.  The reported HR p-values use the bootstrap mean--variance correction in \eqref{eq:boot-corrected-pvalues}, with $B=200$ bootstrap repetitions.  The permutation check uses $B_\pi=499$ permutations, and WL is computed by Monte Carlo calibration.

Table~\ref{tab:gasoline-broad} reports the full-spectrum and broad-window p-values.  The full spectrum is strongly rejected by all methods.  At the window level, the proposed tests provide a componentwise interpretation.  In A1, the rejection is driven mainly by the sum component and WL does not reject, suggesting an aggregated radial--directional departure.  In A2, the max and Cauchy components are significant while the sum component is not, which is consistent with a more localized departure.  A3 and A4 are detected by both the proposed method and WL, indicating stronger global deviations from ellipticity.

\begin{table}[htbp]
\centering
\caption{Gasoline spectroscopy data: full-spectrum and broad-window p-values.}
\label{tab:gasoline-broad}
\small
\setlength{\tabcolsep}{6pt}
\begin{tabular}{lccccc}
\toprule
Window & $p_{\rm used}$ & HR-Sum & HR-Max & HR-Cauchy & WL \\
\midrule
Full spectrum & 401 & $5.9\times10^{-4}$ & $5.7\times10^{-5}$ & $1.05\times10^{-4}$ & $<10^{-4}$ \\
A1: 900--1098nm & 100 & 0.003 & 0.074 & 0.006 & 0.825 \\
A2: 1100--1298nm & 100 & 0.123 & 0.002 & 0.003 & $<10^{-4}$ \\
A3: 1300--1498nm & 100 & $1.4\times10^{-7}$ & $6.8\times10^{-4}$ & $2.9\times10^{-7}$ & $<10^{-4}$ \\
A4: 1500--1700nm & 101 & 0.001 & $4.5\times10^{-4}$ & $6.3\times10^{-4}$ & $<10^{-4}$ \\
\bottomrule
\end{tabular}
\end{table}

\begin{table}[htbp]
\centering
\caption{Gasoline spectroscopy data: selected refined-window p-values.}
\label{tab:gasoline-fine-selected}
\small
\setlength{\tabcolsep}{6pt}
\begin{tabular}{lccccc}
\toprule
Window & Center & HR-Sum & HR-Max & HR-Cauchy & WL \\
\midrule
900--998nm & 949 & 0.011 & 0.046 & 0.018 & 0.031 \\
1050--1148nm & 1099 & 0.180 & 0.008 & 0.015 & 0.976 \\
1100--1198nm & 1149 & 0.078 & 0.006 & 0.011 & $<2\times10^{-4}$ \\
\bottomrule
\end{tabular}
\end{table}

Table~\ref{tab:gasoline-fine-selected} reports selected refined windows in the A1/A2 region.  The 900--998nm window is mainly supported by the sum component, whereas the 1050--1148nm and 1100--1198nm windows are mainly supported by the max component.  Thus, the proposed tests not only reject the elliptical model but also distinguish whether the departure appears as an aggregate effect over many directions or as a localized directional effect.

\subsection{Arcene data}
\label{subsec:additional-applications}

We also examine two additional data sets to illustrate the complementarity between the proposed radial--directional tests and the WL moment-based elliptical test.  The Arcene data are serum mass-spectrometry profiles from the NIPS 2003 Feature Selection Challenge \citep{GuyonEtAl2004,UCIArcene2008}.  We pool the training and validation samples, giving $n=200$ observations and $10{,}000$ original features.  Class labels are not used in the test.  Because the Arcene feature order is randomized, we do not form adjacent feature windows; instead, we retain the top $K$ variables by marginal variance, remove non-finite or constant columns if present, and then standardize the retained variables columnwise before applying all tests.

The cookie data contain near-infrared spectra of biscuit dough samples \citep{BrownFearnVannucci2001,OsborneFearnMillerDouglas1984}.  We use only the first 700 spectral variables, corresponding to 1100--2498nm, and do not use the four constituent percentages as testing variables.  For each reported tail window, the wavelength columns are extracted and standardized within that window.  The ``no outliers'' analysis removes the two documented outlying samples in the calibration and validation sets.

Table~\ref{tab:additional-applications} reports selected blocks.  For Arcene, the top-300 and top-400 feature sets are strongly detected by the proposed tests while WL is much less significant.  In the top-500 block, the sum component is not significant but the max and Cauchy components reject, which points to a sparse radial--directional departure.  For the cookie tail windows, the proposed tests reject strongly, whereas WL is not significant in the selected local blocks.  These examples support the interpretation that the two testing principles are complementary: WL is sensitive to global kurtosis-type moment discrepancies, while the proposed procedure is sensitive to dependence between the fitted radius and fitted direction.

\begin{table}[htbp]
\centering
\caption{Additional data applications: selected p-values illustrating complementarity between the proposed tests and WL.}
\label{tab:additional-applications}
\scriptsize
\setlength{\tabcolsep}{3.5pt}
\begin{tabular}{llccccc}
\toprule
Data & Block & $(n,p)$ & HR-Sum & HR-Max & HR-Cauchy & WL \\
\midrule
Arcene & Arcene top 300 & $(200,300)$ & 0.002 & 0.002 & 0.002 & 0.293 \\
Arcene & Arcene top 400 & $(200,400)$ & 0.002 & 0.002 & 0.002 & 0.051 \\
Arcene & Arcene top 500 & $(200,500)$ & 0.094 & 0.012 & 0.021 & 0.076 \\
Cookie & Cookie 2350--2498nm & $(72,75)$ & 0.002 & 0.002 & 0.002 & 0.061 \\
Cookie & Cookie 2360--2498nm & $(72,70)$ & 0.002 & 0.002 & 0.002 & 0.224 \\
Cookie & Cookie 2360--2498nm (no outl.) & $(70,70)$ & 0.002 & 0.002 & 0.002 & 0.130 \\
\bottomrule
\end{tabular}
\end{table}

The distinction between the sum and max components is useful in these applications.  A significant sum component indicates that many directional coordinates jointly contribute to the radius--direction dependence, whereas a significant max component indicates that the departure is concentrated in a small number of directions.  The Arcene top-500 block and the gasoline A2 window are max-driven, while the gasoline A1 window is sum-driven.  This provides an interpretable sparse-versus-dense diagnostic in addition to the global reject-or-not decision.

\section{Conclusion}
\label{sec:conclusion}

This paper develops a high-dimensional goodness-of-fit framework for elliptical models by focusing on the radial--directional independence property.  The proposed statistics use coordinatewise correlations between the log-radius and the direction after affine standardization.  The sum component accumulates dense radial--directional departures, the max component detects sparse departures, and the Cauchy combination adapts between the two.  The theory establishes oracle normal and Gumbel limits, sum--max asymptotic independence under the null and under a balanced local alternative, and validity of high-dimensional HR plug-in standardization under explicit perturbation rates.  The assumptions on the log-radius allow both concentrated and non-concentrated radial laws, and the appendix verifies them for several common elliptical distributions.

Several extensions are natural.  First, the present paper uses coordinatewise radial--directional correlations after a robust affine standardization.  It would be useful to develop tests based on nonlinear directional features or data-adaptive bases, especially when departures are not aligned with the original coordinates.  Second, the HR plug-in theory is derived under sparsity and regularity conditions on the shape and precision matrices.  Extending the plug-in analysis to more general dependence structures, or to estimators designed for factor or low-rank-plus-sparse covariance models \citep{xu2025high}, would broaden the range of high-dimensional elliptical models covered by the method.

\appendix

\section{Proofs of the Main Results}

\subsection{Derivation of the HR plug-in expansion from Yan-type conditions}
\label{subsec:app-plugin-assumption}

\begin{proof}[Proof of Proposition~\ref{prop:yan-to-plugin}]
Let $\hat{\mathbf{\Omega}}=\hat{\mathbf{\Sigma}}^{-1}$ and
\[
\hat{\mathbf{I}}=\hat{\mathbf{\Omega}}^{1/2}\mathbf{\Sigma}^{1/2}
=\mathbf{I}_p+\mathbf{B}.
\]
The spatial-sign graphical-lasso event is
\begin{equation}
\label{eq:app-glasso-bound}
\begin{aligned}
\mathcal E_{\Omega,n}
=\Bigl\{&
\|\hat{\mathbf{\Omega}}-\mathbf{\Omega}\|_{\max}\le C\lambda_n,
\quad
\|\hat{\mathbf{\Omega}}-\mathbf{\Omega}\|_{op}\le C\lambda_n^{1-q}s_0(p),\\
&
 p^{-1}\|\hat{\mathbf{\Omega}}-\mathbf{\Omega}\|_F^2
\le C\lambda_n^{2-q}s_0(p)
\Bigr\},
\qquad
\Prob(\mathcal E_{\Omega,n})\ge 1-2p^{-2}.
\end{aligned}
\end{equation}
Since $\eta\le\lambda_{\min}(\mathbf{\Sigma})\le\lambda_{\max}(\mathbf{\Sigma})\le\eta^{-1}$,
\begin{equation}
\label{eq:app-B-bound}
\begin{aligned}
\|\mathbf{B}\|_{op}
&=\|\hat{\mathbf{\Omega}}^{1/2}\mathbf{\Sigma}^{1/2}-\mathbf{\Omega}^{1/2}\mathbf{\Sigma}^{1/2}\|_{op}                                                   \\
&\le \|\mathbf{\Sigma}^{1/2}\|_{op}\,
\|\hat{\mathbf{\Omega}}^{1/2}-\mathbf{\Omega}^{1/2}\|_{op}                                                        \\
&\le C\|\hat{\mathbf{\Omega}}-\mathbf{\Omega}\|_{op}
=O_p(\tau_n).
\end{aligned}
\end{equation}
For the HR location, the estimating equation is
\[
\bm{0}
=
\frac1n\sum_{i=1}^n
\mathcal U\{\hat{\mathbf{\Omega}}^{1/2}(\bm{X}_i-\hat{\bm{\mu}})\}
=
\frac1n\sum_{i=1}^n
\mathcal U\{R_i\bm{U}_i+R_i\mathbf{B}\bm{U}_i-\bm{v}\}.
\]
For $\bm{h}$ satisfying $\|\bm{h}\|_2\le r/2$ and $\bm{u}\in\Sphere^{p-1}$,
\begin{align*}
\mathcal U(r\bm{u}+\bm{h})
&=\frac{r\bm{u}+\bm{h}}
        {r\{1+2r^{-1}\bm{u}^\top\bm{h}+r^{-2}\|\bm{h}\|_2^2\}^{1/2}}                                      \\
&=\bm{u}+r^{-1}(\mathbf{I}_p-\bm{u}\bm{u}^\top)\bm{h}+\bm{q}(r,\bm{u},\bm{h}),                                  \\
\|\bm{q}(r,\bm{u},\bm{h})\|_\infty
&\le
C r^{-2}\|\bm{h}\|_2^2\|\bm{u}\|_\infty
+C r^{-2}\|\bm{h}\|_2\|(\mathbf{I}_p-\bm{u}\bm{u}^\top)\bm{h}\|_\infty .
\end{align*}
Substituting $\bm{h}_{i}=R_i\mathbf{B}\bm{U}_i-\bm{v}$ gives
\begin{equation}
\label{eq:app-location-linear}
\begin{aligned}
\bm{0}
&=
\bar{\bm{U}}
+
\left\{\frac1n\sum_{i=1}^n(\mathbf{I}_p-\bm{U}_i\bm{U}_i^\top)\mathbf{B}\bm{U}_i\right\}                                           \\
&\quad
-
\left\{\frac1n\sum_{i=1}^nR_i^{-1}(\mathbf{I}_p-\bm{U}_i\bm{U}_i^\top)\right\}\bm{v}
+
\frac1n\sum_{i=1}^n\bm{q}(R_i,\bm{U}_i,\bm{h}_i).
\end{aligned}
\end{equation}
Put
\[
\mathbf{Q}_n=\frac1n\sum_{i=1}^nR_i^{-1}(\mathbf{I}_p-\bm{U}_i\bm{U}_i^\top),
\qquad
\E(\mathbf{Q}_n)=\zeta_1(1-p^{-1})\mathbf{I}_p.
\]
The entries of \(\mathbf{Q}_n\) satisfy
\begin{align*}
\E(U_{ij}^2)&=p^{-1},
&
\E(U_{ij}^4)&=\frac{3}{p(p+2)},\\
\E(U_{ij}^2U_{ik}^2)&=\frac{1}{p(p+2)}\quad(j\ne k),
&
\Var\{R_i^{-1}(1-U_{ij}^2)\}&\le Cp^{-1},\\
\Var(R_i^{-1}U_{ij}U_{ik})&\le Cp^{-3}\quad(j\ne k).\end{align*}
For \(t>0\),
\begin{align*}
\Prob\left\{
\max_j |Q_{n,jj}-\zeta_1(1-p^{-1})|>t
\right\}&\le 2p\exp(-c npt^2),\\
\Prob\left\{
\max_{j\ne k}|Q_{n,jk}|>t
\right\}&\le 2p^2\exp(-c np^3t^2).
\end{align*}
With
\(t_1=Cp^{-1/2}\sqrt{\log p/n}\) and
\(t_2=Cp^{-3/2}\sqrt{\log p/n}\),
\begin{equation}
\label{eq:app-Qn-rates}
\begin{aligned}
\|\mathbf{Q}_n-\zeta_1(1-p^{-1})\mathbf{I}_p\|_{\max}
&=O_p\left(p^{-1/2}\sqrt{\frac{\log p}{n}}+p^{-3/2}\sqrt{\frac{\log p}{n}}\right),                                  \\
\|\mathbf{Q}_n^{-1}-\zeta_1^{-1}\mathbf{I}_p\|_{\max}
&=O_p\left(\zeta_1^{-1}p^{-1}+\zeta_1^{-2}p^{-1/2}\sqrt{\frac{\log p}{n}}\right).
\end{aligned}
\end{equation}
Combining \eqref{eq:app-glasso-bound}--\eqref{eq:app-Qn-rates} with \eqref{eq:app-location-linear},
\begin{equation}
\label{eq:app-v-expansion}
\begin{aligned}
\bm{v}
&=\zeta_1^{-1}\bar{\bm{U}}
+n^{-1/2}\bm{C}_n,                                                               \\
\bm{C}_n
&=n^{1/2}\{\bm{c}_{n,1}+\bm{c}_{n,2}+\bm{c}_{n,3}\},                         \\
\bm{c}_{n,1}
&=(\mathbf{Q}_n^{-1}-\zeta_1^{-1}\mathbf{I}_p)\bar{\bm{U}},                    \\
\bm{c}_{n,2}
&=\mathbf{Q}_n^{-1}\frac1n\sum_{i=1}^n
(\mathbf{I}_p-\bm{U}_i\bm{U}_i^\top)\mathbf{B}\bm{U}_i,                       \\
\bm{c}_{n,3}
&=\mathbf{Q}_n^{-1}\frac1n\sum_{i=1}^n\bm{q}(R_i,\bm{U}_i,\bm{h}_i),            \\
\|\bm{C}_n\|_\infty
&=O_p\{n^{-1/4}\sqrt{\log(np)}\}                                                \\
&\quad
+O_p\{n^{-(1-q)/2}(\log p)^{(1-q)/2}\sqrt{\log(np)}\,s_0(p)\}                  \\
&\quad
+O_p\{p^{-(1-q)/2}\sqrt{\log(np)}\,s_0(p)\}                                    \\
&=O_p(\mathfrak c_n).
\end{aligned}
\end{equation}
For the correlation vector, define
\[
\bm{\phi}(\bm{y})=(\log\|\bm{y}\|_2-m_L)\frac{\bm{y}}{\|\bm{y}\|_2},
\qquad
\bm{\phi}_i=(L_i-m_L)\bm{U}_i.
\]
For $\bm{y}=R\bm{u}$ and $\bm{h}=R\mathbf{B}\bm{u}-\bm{v}$,
\begin{align*}
D\bm{\phi}(R\bm{u})[\bm{h}]
&=
\left(\frac{\bm{u}^\top\bm{h}}{R}\right)\bm{u}
+(L-m_L)R^{-1}(\mathbf{I}_p-\bm{u}\bm{u}^\top)\bm{h}                                                                \\
&=
(L-m_L)(\mathbf{I}_p-\bm{u}\bm{u}^\top)\mathbf{B}\bm{u}
+\bm{u}(\bm{u}^\top\mathbf{B}\bm{u})                                                                                 \\
&\quad
-R^{-1}\left\{(L-m_L)\mathbf{I}_p+[1-(L-m_L)]\bm{u}\bm{u}^\top\right\}\bm{v}.
\end{align*}
The second derivative of $\bm{\phi}$ satisfies, whenever $\|\bm{h}\|_2\le R/2$,
\begin{equation}
\label{eq:app-phi-second}
\|\bm{\phi}(R\bm{u}+\bm{h})-\bm{\phi}(R\bm{u})-D\bm{\phi}(R\bm{u})[\bm{h}]\|_\infty
\le
C(1+|L-m_L|)R^{-2}\|\bm{h}\|_2^2\{\|\bm{u}\|_\infty+R^{-1}\|\bm{h}\|_\infty\}.
\end{equation}
Therefore
\begin{equation}
\label{eq:app-phi-average}
\begin{aligned}
\frac1{\sigma_L\sigma_U}\frac1n\sum_{i=1}^n\{\bm{\phi}(\hat{\bm{Y}}_i)-\bm{\phi}_i\}
&=
\bm{d}_n^{(\mathbf{\Sigma})}
-(\sigma_L\sigma_U)^{-1}\mathbf{M}_n\bm{v}
+\bm{r}_{n,2}                                                                  \\
&=
\bm{d}_n^{(\mathbf{\Sigma})}
-(\sigma_L\sigma_U)^{-1}\E(\mathbf{M}_n)\bm{v}
-(\sigma_L\sigma_U)^{-1}\mathbf{\Delta}_n\bm{v}
+\bm{r}_{n,2}.
\end{aligned}
\end{equation}
Because $R_i\perpn\bm{U}_i$ and $\E(\bm{U}_i\bm{U}_i^\top)=p^{-1}\mathbf{I}_p$,
\begin{align*}
\E(\mathbf{M}_n)
&=
\left[
\E\left(\frac{L_i-m_L}{R_i}\right)
+\frac1p\E\left(\frac{1-(L_i-m_L)}{R_i}\right)
\right]\mathbf{I}_p
=a_p\mathbf{I}_p.
\end{align*}
Equations \eqref{eq:kappa-def}, \eqref{eq:app-v-expansion} and \eqref{eq:app-phi-average} imply
\begin{equation}
\label{eq:app-main-expansion-no-centering}
\begin{aligned}
\frac1{\sigma_L\sigma_U}\frac1n\sum_{i=1}^n\{\bm{\phi}(\hat{\bm{Y}}_i)-\bm{\phi}_i\}
&=
\tilde\kappa_p\bar{\bm{U}}
+\kappa_p n^{-1/2}\bm{C}_n
+\bm{r}_n^{(1)}
+\bm{d}_n^{(\mathbf{\Sigma})}
+\bm{r}_{n,2}.
\end{aligned}
\end{equation}
The empirical centering and self-normalization in \eqref{eq:plugin-gammahat} give
\begin{equation}
\label{eq:app-centering-remainder}
\begin{aligned}
\hat{\bm{g}}_n-\bm{g}_n^{\rm or}
&=
\frac1{\sigma_L\sigma_U}\frac1n\sum_{i=1}^n\{\bm{\phi}(\hat{\bm{Y}}_i)-\bm{\phi}_i\}
+\bm{r}_{n,c}+\bm{r}_{n,s},                                                               \\
\bm{a}_n
&=\bm{r}_{n,2}+\bm{r}_{n,c}+\bm{r}_{n,s}.
\end{aligned}
\end{equation}
The sample moments obey
\begin{equation}
\label{eq:app-sample-moment-rates}
\begin{aligned}
\|\bar{\bm{U}}\|_2&=O_p(n^{-1/2}),
&
\|\bar{\bm{U}}\|_\infty&=O_p\left(\sqrt{\frac{\log p}{np}}\right),                                                        \\
\|\bm{g}_n^{\rm or}\|_2&=O_p\left(\sqrt{\frac pn}\right),
&
\|\bm{g}_n^{\rm or}\|_\infty&=O_p\left(\sqrt{\frac{\log p}{n}}\right),                                                       \\
\|\mathbf{\Delta}_n\bm{v}\|_2&=O_p\left(\frac{1}{\sqrt n}\left(\tau_n+\sqrt{\frac pn}\right)\right),
&
\|\mathbf{\Delta}_n\bm{v}\|_\infty&=O_p\left(\frac1{\sqrt{np}}\left(\tau_n\sqrt{\log p}+\sqrt{\frac{\log p}{n}}\right)\right), \\
\|\bm{d}_n^{(\mathbf{\Sigma})}\|_2&=O_p\left(\frac{\tau_n}{\sqrt n}\right),
&
\|\bm{d}_n^{(\mathbf{\Sigma})}\|_\infty&=O_p\left(\tau_n\sqrt{\frac{\log p}{n}}\right).
\end{aligned}
\end{equation}
Decompose the matrices in \eqref{eq:app-phi-average} into the parts multiplied by $L_i-m_L=\sigma_LZ_i$ and the parts not multiplied by $L_i-m_L$.  After division by $\sigma_L\sigma_U$, the former parts have the same order as in \eqref{eq:app-sample-moment-rates}, whereas the latter parts carry the additional factor $\rho_{L,p}=\sigma_L^{-1}$.  Thus \eqref{eq:app-sample-moment-rates} gives
\begin{equation}
\label{eq:app-hard-rates}
\begin{aligned}
\frac n{\sqrt p}
\left\{|(\bm{g}_n^{\rm or})^\top \bm{b}_n^{\rm hard}|+
\|\bm{b}_n^{\rm hard}\|_2^2\right\}
&=O_p\left((1+\rho_{L,p})\left\{\tau_n+\frac{\sqrt p}{n}\right\}+\frac{(1+\rho_{L,p})^2\tau_n^2}{\sqrt p}\right),                                                   \\
\sqrt n\|\bm{b}_n^{\rm hard}\|_\infty
&=O_p\left((1+\rho_{L,p})\left\{\tau_n\sqrt{\log p}+\sqrt{\frac{\log p}{n}}\right\}\right).
\end{aligned}
\end{equation}
The Taylor bound \eqref{eq:app-phi-second}, the centering identity
\[
\frac1n\sum_{i=1}^n(\hat L_i-\bar{\hat L})(\hat U_{ij}-\bar{\hat U}_j)
=
\frac1n\sum_{i=1}^n(\hat L_i-m_L)\hat U_{ij}
-(\bar{\hat L}-m_L)\bar{\hat U}_j,
\]
and the variance expansion
\[
\hat\sigma_{\hat L}^{-1}\hat\sigma_{\hat U,j}^{-1}
-(\sigma_L\sigma_U)^{-1}
=O_p\left(\sqrt{\frac{\log p}{n}}+(1+\rho_{L,p})\tau_n+\frac{|\tilde\kappa_p|}{\sqrt p}\right)(\sigma_L\sigma_U)^{-1}
\]
yield
\begin{equation}
\label{eq:app-a-rates}
\begin{aligned}
\frac n{\sqrt p}
\left\{ |(\bm{g}_n^{\rm or})^\top \bm{a}_n|+|\bm{d}_n^\top\bm{a}_n|+\|\bm{a}_n\|_2^2\right\}
&=O_p(\mathfrak A_{S,n}),                                                                 \\
\sqrt n\|\bm{a}_n\|_\infty
&=O_p(\mathfrak A_{M,n}/\sqrt{\log p}).
\end{aligned}
\end{equation}
Equations \eqref{eq:app-v-expansion}, \eqref{eq:app-B-bound}, \eqref{eq:app-main-expansion-no-centering}, \eqref{eq:app-centering-remainder}, \eqref{eq:app-hard-rates} and \eqref{eq:app-a-rates} are \eqref{eq:hr-bahadur}--\eqref{eq:derived-hard-max}.
\end{proof}

\subsection{Proof of Theorem~\ref{thm:linearization}}

\subsubsection{Auxiliary lemmas}

Let
\[
\begin{aligned}
\widetilde L_i&=L_i-m_L,\\
S_{n,j}&=\frac1n\sum_{i=1}^n \widetilde L_iU_{ij},\\
\hat C_j&=\frac1n\sum_{i=1}^n (L_i-\bar L)(U_{ij}-\bar U_j),\\
\sigma_U^2&=\Var(U_{1j}).
\end{aligned}
\]

\begin{lemma}
\label{lem:lin_null}
Under Assumption~\ref{ass:oracle},
\[
\E(U_{1j})=0,
\qquad
\E(U_{1j}^2)=\frac1p,
\qquad
\sigma_U^2=\frac1p,
\qquad
\gamma_j=0.
\]
\end{lemma}

\begin{proof}
Since $\bm{U}_i\sim\Unif(\Sphere^{p-1})$, spherical symmetry gives $\E(U_{1j})=0$ and $\E(U_{1j}^2)=1/p$. Thus $\sigma_U^2=1/p$. Moreover, by Assumption~\ref{ass:oracle}, $R_i\perpn \bm{U}_i$, so
\[
\E(\widetilde L_iU_{ij})=\E(\widetilde L_i)\E(U_{ij})=0.
\]
Hence $\Cov(L_i,U_{ij})=0$, that is, $\gamma_j=0$.
\end{proof}

\begin{lemma}
\label{lem:lin_basic_bounds}
Under Assumption~\ref{ass:oracle}. Then
\[
\bar L-m_L=\sigma_L O_p(n^{-1/2}),
\qquad
\max_{1\le j\le p}|\bar U_j|=O_p\!\left(\sqrt{\frac{\log p}{np}}\right),
\]
\[
\max_{1\le j\le p}|S_{n,j}|=\sigma_L O_p\!\left(\sqrt{\frac{\log p}{np}}\right),
\]
\[
\frac{\hat\sigma_L^2}{\sigma_L^2}-1=O_p(n^{-1/2}),
\]
\[
\max_{1\le j\le p}|\hat\sigma_{U,j}^2-\sigma_U^2|
=O_p\!\left(\frac{\sqrt{\log p}}{p\sqrt n}+\frac{\log p}{np}\right).
\]
For each fixed $j$, the sharper bounds
\[
\bar U_j=O_p((np)^{-1/2}),
\qquad
S_{n,j}=\sigma_L O_p((np)^{-1/2}),
\qquad
\hat\sigma_{U,j}^2-\sigma_U^2=O_p\!\left(\frac{1}{p\sqrt n}\right)
\]
hold as well.
\end{lemma}

\begin{proof}
Since $L_i-m_L=\sigma_L Z_i$,
\[
\bar L-m_L=\sigma_L\bar Z,
\qquad
\bar Z=\frac1n\sum_{i=1}^n Z_i,
\qquad
\E\bar Z^2=n^{-1}.
\]
Thus $\bar L-m_L=\sigma_LO_p(n^{-1/2})$.  For a uniform direction $\bm{U}\sim\Unif(\Sphere^{p-1})$ and every deterministic $\bm{a}\in\R^p$,
\[
\Pr(|\bm{a}^{\top}\bm{U}|>t)
\le 2\exp\left(-\frac{cpt^2}{\|\bm{a}\|_2^2}\right),
\qquad t>0,
\]
with a numerical constant $c>0$.  Applying this inequality to $\bm{a}=n^{-1}(1,\ldots,1)^{\top}$ after conditioning on all sample indices gives
\[
\Pr\left(\max_{1\le j\le p}|\bar U_j|>t\right)
\le 2p\exp(-cnpt^2).
\]
Taking $t=C\{\log p/(np)\}^{1/2}$ gives the displayed rate for $\max_j|\bar U_j|$.

Let $V_n=n^{-1}\sum_{i=1}^n Z_i^2$.  Assumption~\ref{ass:oracle}\textnormal{(ii)} gives
\[
\E(V_n-1)^2
=n^{-2}\sum_{i=1}^n\Var(Z_i^2)
\le n^{-1}\E Z_i^4=O(n^{-1}),
\qquad
V_n=1+O_p(n^{-1/2}).
\]
Conditional on $Z_1,\ldots,Z_n$, the same spherical concentration applied to the weighted average gives
\[
\Pr\left(
\max_{1\le j\le p}\left|\frac1n\sum_{i=1}^n Z_iU_{ij}\right|>t
\,\bigg|\,Z_1,\ldots,Z_n\right)
\le
2p\exp\left(-\frac{cnpt^2}{V_n}\right).
\]
On the event $V_n\le 2$, choosing $t=C\{\log p/(np)\}^{1/2}$ gives
\[
\max_{1\le j\le p}\left|\frac1n\sum_{i=1}^n Z_iU_{ij}\right|
=O_p\!\left(\sqrt{\frac{\log p}{np}}\right).
\]
Since $S_{n,j}=\sigma_L n^{-1}\sum_i Z_iU_{ij}$, this proves the displayed rate for $S_{n,j}$.

Furthermore,
\[
\frac{\hat\sigma_L^2}{\sigma_L^2}
=\frac1n\sum_{i=1}^n(Z_i-\bar Z)^2
=1+\frac1n\sum_{i=1}^n(Z_i^2-1)-\bar Z^2.
\]
The first centered average is $O_p(n^{-1/2})$ because $\E Z_i^4$ is uniformly bounded, and $\bar Z^2=O_p(n^{-1})$.  Hence
\[
\hat\sigma_L^2/\sigma_L^2-1=O_p(n^{-1/2}).
\]
For the directional variance, $\hat\sigma_{U,j}^2=n^{-1}\sum_i U_{ij}^2-\bar U_j^2$.  The beta law $U_{ij}^2\sim\mathrm{Beta}\{1/2,(p-1)/2\}$ implies
\[
\Pr\left(\left|\frac1n\sum_{i=1}^n(pU_{ij}^2-1)\right|>x\right)
\le 2\exp\{-c n\min(x^2,x)\},
\qquad x>0,
\]
uniformly in $j$.  Therefore
\[
\max_{1\le j\le p}\left|\frac1n\sum_{i=1}^n U_{ij}^2-\frac1p\right|
=O_p\!\left(\frac{\sqrt{\log p}}{p\sqrt n}\right),
\qquad
\max_j\bar U_j^2=O_p\!\left(\frac{\log p}{np}\right),
\]
which proves the stated uniform variance bound.  Omitting the union bound gives the fixed-coordinate rates.
\end{proof}

\begin{lemma}
\label{lem:lin_expansions}
Under Assumption~\ref{ass:oracle}. Then
\[
\max_{1\le j\le p}|\hat C_j-S_{n,j}|
\le |\bar L-m_L|\max_{1\le j\le p}|\bar U_j|
=\sigma_L O_p\!\left(\frac{\sqrt{\log p}}{n\sqrt p}\right),
\]
so that
\[
\max_{1\le j\le p}|\hat C_j|=\sigma_L O_p\!\left(\sqrt{\frac{\log p}{np}}\right).
\]
Moreover,
\[
\hat\sigma_L^{-1}=\sigma_L^{-1}\{1+O_p(n^{-1/2})\},
\]
and uniformly over $1\le j\le p$,
\[
\hat\sigma_{U,j}^{-1}=\sigma_U^{-1}\left[1+O_p\!\left\{\sqrt{\frac{\log p}{n}}+\frac{\log p}{n}\right\}\right],
\]
whence
\[
\max_{1\le j\le p}
\left|
\frac{1}{\hat\sigma_L\hat\sigma_{U,j}}-\frac{1}{\sigma_L\sigma_U}
\right|
=
\frac{1}{\sigma_L\sigma_U}
O_p\!\left(\sqrt{\frac{\log p}{n}}+\frac{\log p}{n}+\frac1{\sqrt n}\right)
=
\frac{\sqrt p}{\sigma_L}O_p\!\left(\sqrt{\frac{\log p}{n}}+\frac{\log p}{n}\right).
\]
For each fixed $j$,
\[
\hat C_j-S_{n,j}=\sigma_L O_p\bigl((n\sqrt p)^{-1}\bigr),
\qquad
\frac{1}{\hat\sigma_L\hat\sigma_{U,j}}-\frac{1}{\sigma_L\sigma_U}
=
\frac{1}{\sigma_L\sigma_U}O_p(n^{-1/2}).
\]
\end{lemma}

\begin{proof}
Using $\widetilde L_i=L_i-m_L$ and $\E(U_{ij})=0$,
\[
\hat C_j
=
\frac1n\sum_{i=1}^n \widetilde L_iU_{ij}-(\bar L-m_L)\bar U_j
=
S_{n,j}-(\bar L-m_L)\bar U_j.
\]
The two bounds for $\hat C_j-S_{n,j}$ and $\hat C_j$ follow from Lemma~\ref{lem:lin_basic_bounds}.  Also,
\[
\frac{\hat\sigma_L^2}{\sigma_L^2}=1+O_p(n^{-1/2}),
\qquad
\max_{1\le j\le p}\left|\frac{\hat\sigma_{U,j}^2}{\sigma_U^2}-1\right|
=O_p\!\left(\sqrt{\frac{\log p}{n}}+\frac{\log p}{n}\right)=o_p(1).
\]
For $|x|\le1/2$,
\[
(1+x)^{-1/2}=1-\frac{x}{2}+O(x^2),
\]
and hence the displayed inverse expansions follow with $x=\hat\sigma_L^2/\sigma_L^2-1$ and $x=\hat\sigma_{U,j}^2/\sigma_U^2-1$.  The fixed-coordinate bounds are obtained in the same way from the fixed-coordinate rates in Lemma~\ref{lem:lin_basic_bounds}.
\end{proof}

\subsubsection{Proof of Theorem~\ref{thm:linearization}}

\begin{proof}[Proof of Theorem~\ref{thm:linearization}]
Lemma~\ref{lem:lin_null} gives $\gamma_j=0$ for all $j$. To prove the expansion, write
\[
\hat\gamma_j
=
\hat C_j\cdot \frac{1}{\hat\sigma_L\hat\sigma_{U,j}}.
\]
Subtracting $(\sigma_L\sigma_U)^{-1}S_{n,j}$ from both sides gives
\[
\hat\gamma_j-\frac{1}{\sigma_L\sigma_U}S_{n,j}
=
\hat C_j\left(\frac{1}{\hat\sigma_L\hat\sigma_{U,j}}-\frac{1}{\sigma_L\sigma_U}\right)
+
\frac{1}{\sigma_L\sigma_U}(\hat C_j-S_{n,j}).
\]
By Lemma~\ref{lem:lin_expansions}, uniformly in $j$,
\begin{align*}
&\max_{1\le j\le p}
\left|
\hat C_j\left(\frac{1}{\hat\sigma_L\hat\sigma_{U,j}}-\frac{1}{\sigma_L\sigma_U}\right)
\right| \\
&\qquad=
\sigma_L O_p\!\left(\sqrt{\frac{\log p}{np}}\right)
\cdot
\frac{\sqrt p}{\sigma_L}O_p\!\left(\sqrt{\frac{\log p}{n}}+\frac{\log p}{n}\right) \\
&\qquad=
O_p\!\left(\frac{\log p}{n}+\frac{(\log p)^{3/2}}{n^{3/2}}\right)
=O_p\!\left(\frac{\log p}{n}\right),
\end{align*}
while
\[
\max_{1\le j\le p}
\left|
\frac{1}{\sigma_L\sigma_U}(\hat C_j-S_{n,j})
\right|
=
\frac{\sqrt p}{\sigma_L}\,\sigma_L O_p\!\left(\frac{\sqrt{\log p}}{n\sqrt p}\right)
=
O_p\!\left(\frac{\sqrt{\log p}}{n}\right).
\]
Hence
\[
\max_{1\le j\le p}
\left|
\hat\gamma_j-\frac{1}{\sigma_L\sigma_U}S_{n,j}
\right|
=
O_p\!\left(\frac{\log p+\sqrt{\log p}}{n}\right).
\]
In particular, since $p\ge 3$ eventually, this bound is $O_p((\log p)/n)$ in the high-dimensional cases considered below.
Since $\sigma_U=p^{-1/2}$,
\[
\frac{1}{\sigma_L\sigma_U}S_{n,j}
=
\frac{\sqrt p}{\sigma_L}\cdot\frac1n\sum_{i=1}^n \widetilde L_iU_{ij}
=
\frac1n\sum_{i=1}^n \xi_{ij}.
\]
Therefore
\[
\hat\gamma_j=\frac1n\sum_{i=1}^n \xi_{ij}+r_{n,j},
\qquad
\max_{1\le j\le p}|r_{n,j}|=O_p\!\left(\frac{\log p+\sqrt{\log p}}{n}\right).
\]
Because $\log p=o(\sqrt n)$, the displayed remainder is $o_p(n^{-1/2})$, proving the uniform expansion.

For the fixed-coordinate expansion, fix $j$. The same decomposition together with the fixed-coordinate estimates in Lemma~\ref{lem:lin_expansions} yields
\[
\hat C_j\left(\frac{1}{\hat\sigma_L\hat\sigma_{U,j}}-\frac{1}{\sigma_L\sigma_U}\right)=O_p(n^{-1}),
\qquad
\frac{1}{\sigma_L\sigma_U}(\hat C_j-S_{n,j})=O_p(n^{-1}).
\]
Hence
\[
\hat\gamma_j=\frac1n\sum_{i=1}^n \xi_{ij}+O_p(n^{-1}),
\]
which proves the fixed-coordinate statement.
\end{proof}

\subsection{Proof of Theorem~\ref{thm:sum-clt}}

Define
\[
\hat C_j=
\frac1n\sum_{i=1}^n (L_i-\bar L)(U_{ij}-\bar U_j),
\qquad
\bar Z=\frac1n\sum_{i=1}^n Z_i,
\qquad
\bar{\bm{U}}=\frac1n\sum_{i=1}^n \bm{U}_i,
\]
\[
\bm{W}_n=\frac1n\sum_{i=1}^n Z_i\bm{U}_i\in\R^p,
\qquad
\hat v_L=\frac{\hat\sigma_L^2}{\sigma_L^2}=\frac1n\sum_{i=1}^n (Z_i-\bar Z)^2.
\]

\subsubsection{Auxiliary lemmas}

\begin{lemma}
\label{lem:ideal_centering}
Under Assumption~\ref{ass:oracle} and $p=O(n^\kappa)$ for some $\kappa\in(0,2)$. Then
\[
np\Bigl(\|\bm{W}_n-\bar Z\,\bar{\bm{U}}\|^2-\|\bm{W}_n\|^2\Bigr)=o_p(\sqrt p).
\]
Consequently,
\[
\widetilde T_{\mathrm{sum}}
=
\widetilde T_{\mathrm{sum}}^{\circ}+o_p(\sqrt p),
\qquad
\widetilde T_{\mathrm{sum}}^{\circ}:=\frac{np}{\hat v_L}\|\bm{W}_n\|^2,
\]
where
\[
\widetilde T_{\mathrm{sum}}:=\frac{np}{\hat\sigma_L^2}\sum_{j=1}^p \hat C_j^2.
\]
\end{lemma}

\begin{proof}
Since $L_i-\bar L=\sigma_L(Z_i-\bar Z)$ and
\[
\frac1n\sum_{i=1}^n (Z_i-\bar Z)(\bm{U}_i-\bar{\bm{U}})=\bm{W}_n-\bar Z\,\bar{\bm{U}},
\]
we have
\[
\widetilde T_{\mathrm{sum}}=\frac{np}{\hat v_L}\|\bm{W}_n-\bar Z\,\bar{\bm{U}}\|^2.
\]
Now
\[
\|\bm{W}_n-\bar Z\,\bar{\bm{U}}\|^2-\|\bm{W}_n\|^2=-2\bar Z\,\bar{\bm{U}}^{\top}\bm{W}_n+\bar Z^2\|\bar{\bm{U}}\|^2.
\]
Because $\bar Z=O_p(n^{-1/2})$ and
\[
\E\|\bar{\bm{U}}\|^2=\frac1{n^2}\sum_{i=1}^n \E\|\bm{U}_i\|^2=\frac1n,
\]
we have $\|\bar{\bm{U}}\|^2=O_p(n^{-1})$, hence
\[
np\,\bar Z^2\|\bar{\bm{U}}\|^2=O_p(p/n)=o_p(\sqrt p)
\]
since $p=o(n^2)$. Next,
\[
\bar{\bm{U}}^{\top}\bm{W}_n=
\frac1{n^2}\sum_{i=1}^n\sum_{k=1}^n Z_i\bm{U}_k^{\top}\bm{U}_i
=
\frac{\bar Z}{n}+R_n,
\qquad
R_n:=\frac1{n^2}\sum_{i\ne k} Z_i\bm{U}_k^{\top}\bm{U}_i.
\]
A direct second-moment calculation gives $\E(R_n)=0$ and
\[
\E(R_n^2)
=
\frac1{n^4}\sum_{i=1}^n\sum_{k\ne i}\E\bigl[(\bm{U}_k^{\top}\bm{U}_i)^2\bigr]
=
\frac1{n^4}\cdot n(n-1)\cdot \frac1p
=
O\!\left(\frac1{n^2p}\right),
\]
because for independent uniform directions on $\Sphere^{p-1}$, $\E[(\bm{U}_k^{\top}\bm{U}_i)^2]=1/p$. Thus
\[
R_n=O_p\!\left(\frac1{n\sqrt p}\right),
\qquad
\bar{\bm{U}}^{\top}\bm{W}_n=O_p(n^{-3/2})+O_p\!\left(\frac1{n\sqrt p}\right).
\]
Therefore
\[
np\,|\bar Z\,\bar{\bm{U}}^{\top}\bm{W}_n|=O_p(p/n)+O_p(\sqrt{p/n})=o_p(\sqrt p).
\]
Combining the two bounds yields the first claim. Since $\hat v_L=1+o_p(1)$, the displayed approximation for $\widetilde T_{\mathrm{sum}}$ follows.
\end{proof}

\begin{lemma}
\label{lem:ideal_identity}
Under Assumption~\ref{ass:oracle} and $p=o(n^2)$. Define
\[
A_n=\frac1n\sum_{i=1}^n (Z_i^2-1),
\qquad
H_n=\frac{2p}{n}\sum_{1\le i<k\le n} Z_iZ_k\,\bm{U}_i^{\top}\bm{U}_k.
\]
Then
\[
np\|\bm{W}_n\|^2=p(1+A_n)+H_n,
\qquad
\hat v_L=1+A_n-\bar Z^2,
\]
and hence
\[
\widetilde T_{\mathrm{sum}}^{\circ}-p
=
\frac{p\bar Z^2+H_n}{1+A_n-\bar Z^2}.
\]
Moreover,
\[
\frac{p\bar Z^2}{\sqrt{2p}}=o_p(1),
\qquad
1+A_n-\bar Z^2\xrightarrow{p}1.
\]
\end{lemma}

\begin{proof}
By definition of $\bm{W}_n$,
\[
np\|\bm{W}_n\|^2
=
\frac{p}{n}\sum_{i=1}^n Z_i^2
+
\frac{2p}{n}\sum_{1\le i<k\le n} Z_iZ_k\,\bm{U}_i^{\top}\bm{U}_k
=
p(1+A_n)+H_n.
\]
Also,
\[
\hat v_L=\frac1n\sum_{i=1}^n (Z_i-\bar Z)^2=1+A_n-\bar Z^2.
\]
Substituting these two identities into $\widetilde T_{\mathrm{sum}}^{\circ}=np\|\bm{W}_n\|^2/\hat v_L$ gives the exact formula. Finally, $\bar Z=O_p(n^{-1/2})$, so
\[
\frac{p\bar Z^2}{\sqrt{2p}}=O_p\!\left(\frac{\sqrt p}{n}\right)=o_p(1)
\]
because $p=o(n^2)$, while $A_n=O_p(n^{-1/2})$ and $\bar Z^2=O_p(n^{-1})$ imply $1+A_n-\bar Z^2\to 1$ in probability.
\end{proof}

\begin{lemma}
\label{lem:Hn_martingale_clt}
Under Assumption~\ref{ass:oracle}, $p\to\infty$, and $p=O(n^\kappa)$ for some $\kappa\in(0,2)$. Then
\[
\frac{H_n}{\sqrt{2p}}\xrightarrow{d}N(0,1).
\]
\end{lemma}

\begin{proof}
Define
\[
\bm{\xi}_i=\sqrt p\,Z_i\bm{U}_i\in\R^p,
\qquad
\bm{R}_k=\sum_{i=1}^k \bm{\xi}_i,
\qquad
\mathcal F_k=\sigma(\bm{\xi}_1,\dots,\bm{\xi}_k).
\]
Then
\[
H_n=\frac{2}{n}\sum_{1\le i<k\le n}\bm{\xi}_i^{\top}\bm{\xi}_k
=\sum_{k=2}^n M_{n,k},
\qquad
M_{n,k}:=\frac{2}{n}\bm{R}_{k-1}^{\top}\bm{\xi}_k.
\]
Since $\E(\bm{\xi}_k)=0$ and $\bm{\xi}_k$ is independent of $\mathcal F_{k-1}$, $\{M_{n,k},\mathcal F_k\}$ is a martingale difference array. Moreover,
\[
\E(\bm{\xi}_i\bm{\xi}_i^{\top})=\mathbf{I}_p,
\qquad
\E\|\bm{\xi}_i\|^2=p,
\qquad
\E\|\bm{\xi}_i\|^4=O(p^2),
\qquad
\E[(\bm{\xi}_i^{\top}\bm{\xi}_k)^2]=p\quad(i\ne k).
\]
The predictable quadratic variation is
\[
V_n:=\sum_{k=2}^n \E(M_{n,k}^2\mid\mathcal F_{k-1})
=
\frac{4}{n^2}\sum_{k=2}^n \|\bm{R}_{k-1}\|^2.
\]
Using
\[
\|\bm{R}_{k-1}\|^2=
\sum_{i=1}^{k-1}\|\bm{\xi}_i\|^2+2\sum_{1\le i<j\le k-1}\bm{\xi}_i^{\top}\bm{\xi}_j
\]
and $\E(\bm{\xi}_i^{\top}\bm{\xi}_j)=0$ for $i\ne j$, we get
\[
\E\|\bm{R}_{k-1}\|^2=(k-1)p,
\qquad
\E(V_n)=\frac{4}{n^2}\sum_{k=2}^n (k-1)p=\frac{2(n-1)}{n}p.
\]
Write
\[
V_n=A_n^{\star}+B_n^{\star},
\qquad
A_n^{\star}=\frac{4}{n^2}\sum_{i=1}^{n-1}(n-i)\|\bm{\xi}_i\|^2,
\qquad
B_n^{\star}=\frac{8}{n^2}\sum_{1\le i<j\le n-1}(n-j)\bm{\xi}_i^{\top}\bm{\xi}_j.
\]
Since $\|\bm{\xi}_i\|^2$ are independent,
\[
\Var(A_n^{\star})
\le
\frac{16}{n^4}\sum_{i=1}^{n-1}(n-i)^2\E\|\bm{\xi}_i\|^4
=
O\!\left(\frac{p^2}{n}\right).
\]
Likewise, if $Y_{ij}=\bm{\xi}_i^{\top}\bm{\xi}_j$, then $\E(Y_{ij})=0$ and $\E(Y_{ij}Y_{rs})=0$ unless $(i,j)=(r,s)$, so
\[
\Var(B_n^{\star})=
\frac{64}{n^4}\sum_{1\le i<j\le n-1}(n-j)^2\E(Y_{ij}^2)
=
\frac{64p}{n^4}\sum_{j=2}^{n-1}(j-1)(n-j)^2
=
O(p).
\]
Therefore
\[
\Var(V_n)=O\!\left(\frac{p^2}{n}\right)+O(p)=o(p^2),
\qquad
\frac{V_n}{2p}\xrightarrow{p}1.
\]
To verify the Lindeberg condition, note that conditionally on $\mathcal F_{k-1}$, $\bm{R}_{k-1}$ is deterministic, and for any fixed $\bm{a}\in\R^p$,
\[
\E[(\bm{a}^{\top}\bm{\xi}_k)^4]
=
p^2\E(Z_k^4)\E[(\bm{a}^{\top}\bm{U}_k)^4]
\le
Cp^2\cdot \frac{\|\bm{a}\|^4}{p^2}
=
C\|\bm{a}\|^4,
\]
where we used the spherical fourth-moment identity $\E[(\bm{a}^{\top}\bm{U}_k)^4]=3\|\bm{a}\|^4/[p(p+2)]$. Hence
\[
\E(M_{n,k}^4)
\le \frac{C}{n^4}\E\|\bm{R}_{k-1}\|^4.
\]
Now
\[
\|\bm{R}_{k-1}\|^4
\le 2\left(\sum_{i=1}^{k-1}\|\bm{\xi}_i\|^2\right)^2
+8\left(\sum_{1\le i<j\le k-1}\bm{\xi}_i^{\top}\bm{\xi}_j\right)^2,
\]
so independence, centering, $\E\|\bm{\xi}_i\|^4=O(p^2)$, and $\E[(\bm{\xi}_i^{\top}\bm{\xi}_j)^2]=p$ imply
\[
\E\|\bm{R}_{k-1}\|^4=O\bigl((k-1)^2p^2\bigr).
\]
Consequently,
\[
\E(M_{n,k}^4)\le \frac{C}{n^4}(k-1)^2p^2,
\qquad
\sum_{k=2}^n \E(M_{n,k}^4)=O\!\left(\frac{p^2}{n}\right).
\]
Therefore
\[
\frac{1}{(2p)^2}\sum_{k=2}^n \E(M_{n,k}^4)\to 0,
\]
and for every $\varepsilon>0$,
\[
\frac1{2p}\sum_{k=2}^n
\E\left[M_{n,k}^2\mathbf 1\{|M_{n,k}|>\varepsilon\sqrt{2p}\}\right]
\le
\frac1{\varepsilon^2(2p)^2}\sum_{k=2}^n \E(M_{n,k}^4)\to 0.
\]
Thus the Lindeberg condition holds. Hall and Heyde's martingale CLT \cite[Corollary~3.1]{HallHeyde1980} now yields
\[
\frac{H_n}{\sqrt{2p}}=\frac{\sum_{k=2}^n M_{n,k}}{\sqrt{2p}}\xrightarrow{d}N(0,1).
\]
\end{proof}

\begin{proposition}
\label{prop:sum_clt_stage1}
Under Assumption~\ref{ass:oracle}. Suppose furthermore that
\[
p\to\infty,
\qquad
p=O(n^\kappa)
\quad\text{for some }\kappa\in(0,2).
\]
Define the idealized sum statistic
\[
\widetilde T_{\mathrm{sum}}=\frac{np}{\hat\sigma_L^2}\sum_{j=1}^p \hat C_j^2.
\]
Then
\[
\frac{\widetilde T_{\mathrm{sum}}-p}{\sqrt{2p}}\xrightarrow{d}N(0,1).
\]
\end{proposition}

\begin{proof}
By Lemma~\ref{lem:ideal_centering},
\[
\widetilde T_{\mathrm{sum}}=\widetilde T_{\mathrm{sum}}^{\circ}+o_p(\sqrt p).
\]
By Lemma~\ref{lem:ideal_identity},
\[
\widetilde T_{\mathrm{sum}}^{\circ}-p=\frac{p\bar Z^2+H_n}{1+A_n-\bar Z^2},
\]
with $(p\bar Z^2)/\sqrt{2p}=o_p(1)$ and $1+A_n-\bar Z^2\to 1$ in probability. Lemma~\ref{lem:Hn_martingale_clt} then gives
\[
\frac{\widetilde T_{\mathrm{sum}}^{\circ}-p}{\sqrt{2p}}\xrightarrow{d}N(0,1).
\]
Combining these two facts and applying Slutsky's theorem proves the proposition.
\end{proof}

\begin{lemma}
\label{lem:eta_quadratic_bound}
Let $\eta_{ij}=U_{ij}^2-p^{-1}$. Then, for any deterministic coefficients $a_1,\dots,a_p$,
\[
\E\left(\sum_{j=1}^p a_j\eta_{ij}\right)^2
=
\frac{2}{p(p+2)}\sum_{j=1}^p a_j^2
-
\frac{2}{p^2(p+2)}\left(\sum_{j=1}^p a_j\right)^2
\le
\frac{2}{p(p+2)}\sum_{j=1}^p a_j^2.
\]
\end{lemma}

\begin{proof}
For $\bm{U}_i\sim\Unif(\Sphere^{p-1})$, the standard sphere moments are
\[
\E(U_{ij}^2)=\frac1p,
\qquad
\E(U_{ij}^4)=\frac{3}{p(p+2)},
\qquad
\E(U_{ij}^2U_{im}^2)=\frac{1}{p(p+2)}\quad (j\ne m).
\]
Hence
\[
\E(\eta_{ij}^2)=\frac{2(p-1)}{p^2(p+2)},
\qquad
\E(\eta_{ij}\eta_{im})=-\frac{2}{p^2(p+2)}\quad (j\ne m).
\]
Expanding the square and substituting these identities gives the stated formula.
\end{proof}

\begin{lemma}
\label{lem:HG_bounds}
Define
\[
H_{ik}:=\sum_{j=1}^p U_{kj}^2\eta_{ij},
\qquad
G_{ik\ell}:=\sum_{j=1}^p U_{kj}U_{\ell j}\eta_{ij},
\]
for $i\ne k$ and pairwise distinct $i,k,\ell$. Then
\[
\E(H_{ik}^2)=O(p^{-3}),
\qquad
\E(G_{ik\ell}^2)=O(p^{-3}).
\]
Moreover, $H_{ik}$ is degenerate in both arguments $i$ and $k$, and $G_{ik\ell}$ is degenerate in each of its three arguments.
\end{lemma}

\begin{proof}
Since $\E(\eta_{ij})=0$, conditioning on $\bm{U}_k$ gives
\[
\E(H_{ik}\mid \bm{U}_k)=\sum_{j=1}^p U_{kj}^2\E(\eta_{ij})=0.
\]
Conditioning on $\bm{U}_i$ gives
\[
\E(H_{ik}\mid \bm{U}_i)=\frac1p\sum_{j=1}^p \eta_{ij}=\frac1p(\|\bm{U}_i\|^2-1)=0.
\]
Thus $H_{ik}$ is completely degenerate. The same argument, together with $\E(U_{kj})=0$, shows that $G_{ik\ell}$ is degenerate in each argument.

For the second moments, conditioning on $\bm{U}_k$ and applying Lemma~\ref{lem:eta_quadratic_bound} with $a_j=U_{kj}^2$ yields
\[
\E(H_{ik}^2\mid \bm{U}_k)
\le
\frac{2}{p(p+2)}\sum_{j=1}^p U_{kj}^4.
\]
Taking expectation and using $\E\sum_{j=1}^p U_{kj}^4=3/(p+2)$ gives $\E(H_{ik}^2)=O(p^{-3})$. Likewise, conditioning on $(\bm{U}_k,\bm{U}_\ell)$ and applying Lemma~\ref{lem:eta_quadratic_bound} with $a_j=U_{kj}U_{\ell j}$ gives
\[
\E(G_{ik\ell}^2\mid \bm{U}_k,\bm{U}_\ell)
\le
\frac{2}{p(p+2)}\sum_{j=1}^p U_{kj}^2U_{\ell j}^2.
\]
Taking expectation and using independence of $\bm{U}_k$ and $\bm{U}_\ell$ yields $\E(G_{ik\ell}^2)=O(p^{-3})$.
\end{proof}

\begin{lemma}
\label{lem:A1n_second_moment}
Under Assumption~\ref{ass:oracle}. Recall that
\[
\eta_{ij}=U_{ij}^2-\frac1p,
\qquad
W_j^{(-i)}=\frac1n\sum_{k\ne i} Z_kU_{kj},
\]
and define
\[
A_{1n}=-p^2\sigma_L^2\sum_{i=1}^n\sum_{j=1}^p \bigl(W_j^{(-i)}\bigr)^2\eta_{ij}.
\]
Then
\[
\E\left\{\left(\frac{A_{1n}}{\sigma_L^2}\right)^2\right\}=O\!\left(\frac{p}{n}\right).
\]
In particular,
\[
\frac{A_{1n}}{\sigma_L^2\sqrt p}\xrightarrow{p}0.
\]
\end{lemma}

\begin{proof}
For each fixed $i$ and $j$,
\[
\bigl(W_j^{(-i)}\bigr)^2
=
\frac1{n^2}
\left(
\sum_{k\ne i} Z_k^2U_{kj}^2
+
2\sum_{\substack{k<\ell\\ k,\ell\ne i}} Z_kZ_\ell U_{kj}U_{\ell j}
\right).
\]
Hence $A_{1n}=B_{2n}+B_{3n}$, where
\[
B_{2n}:=-\frac{p^2\sigma_L^2}{n^2}\sum_{i\ne k} Z_k^2H_{ik},
\qquad
B_{3n}:=-\frac{2p^2\sigma_L^2}{n^2}
\sum_{i=1}^n\sum_{\substack{k<\ell\\ k,\ell\ne i}} Z_kZ_\ell G_{ik\ell}.
\]
Since
\[
H_{ik}=\sum_{j=1}^p U_{kj}^2\eta_{ij}=\sum_{j=1}^p U_{ij}^2U_{kj}^2-\frac1p,
\]
we have $H_{ik}=H_{ki}$. Put $X_{ik}=Z_k^2H_{ik}$. The degeneracy of $H_{ik}$ implies that $\E(X_{ik}X_{i'k'})=0$ unless either $(i,k)=(i',k')$ or $(i,k)=(k',i')$. Therefore
\[
\E\left\{\left(\frac{B_{2n}}{\sigma_L^2}\right)^2\right\}
=
\frac{p^4}{n^4}
\left\{
\sum_{i\ne k}\E(X_{ik}^2)+\sum_{i\ne k}\E(X_{ik}X_{ki})
\right\}.
\]
By independence of the $Z$'s and $U$'s, finiteness of $\E(Z_k^4)$, and Lemma~\ref{lem:HG_bounds}, both $\E(X_{ik}^2)$ and $\E(X_{ik}X_{ki})$ are $O(p^{-3})$. Since there are $O(n^2)$ ordered pairs $(i,k)$ with $i\ne k$,
\[
\E\left\{\left(\frac{B_{2n}}{\sigma_L^2}\right)^2\right\}=O\!\left(\frac{p^4}{n^4}\cdot n^2\cdot p^{-3}\right)=O\!\left(\frac{p}{n^2}\right).
\]
Next, let
\[
\mathcal I_3:=\{(i,k,\ell):1\le i\le n,\ 1\le k<\ell\le n,\ k,\ell\ne i\},
\qquad
Y_{ik\ell}:=Z_kZ_\ell G_{ik\ell}.
\]
Then
\[
B_{3n}=-\frac{2p^2\sigma_L^2}{n^2}\sum_{(i,k,\ell)\in\mathcal I_3}Y_{ik\ell}.
\]
If $\{i,k,\ell\}\ne\{i',k',\ell'\}$, then among the six indices there is at least one appearing exactly once. Conditioning on all other variables and using the degeneracy of $G_{ik\ell}$ together with $\E(Z_r)=0$ shows that
\[
\E(Y_{ik\ell}Y_{i'k'\ell'})=0.
\]
Thus only pairs with the same underlying triple can contribute. The number of such ordered pairs is $O(n^3)$, and each contribution is $O(p^{-3})$ by Cauchy--Schwarz and Lemma~\ref{lem:HG_bounds}. Therefore
\[
\E\left\{\left(\frac{B_{3n}}{\sigma_L^2}\right)^2\right\}=O\!\left(\frac{p^4}{n^4}\cdot n^3\cdot p^{-3}\right)=O\!\left(\frac{p}{n}\right).
\]
Finally,
\[
\left|\E\left(\frac{B_{2n}B_{3n}}{\sigma_L^4}\right)\right|
\le
\left[\E\left\{\left(\frac{B_{2n}}{\sigma_L^2}\right)^2\right\}\right]^{1/2}
\left[\E\left\{\left(\frac{B_{3n}}{\sigma_L^2}\right)^2\right\}\right]^{1/2}
=
O\!\left(\frac{p}{n^{3/2}}\right).
\]
Since $A_{1n}=B_{2n}+B_{3n}$,
\[
\E\left\{\left(\frac{A_{1n}}{\sigma_L^2}\right)^2\right\}
=
\E\left\{\left(\frac{B_{2n}}{\sigma_L^2}\right)^2\right\}
+\E\left\{\left(\frac{B_{3n}}{\sigma_L^2}\right)^2\right\}
+2\E\left(\frac{B_{2n}B_{3n}}{\sigma_L^4}\right)
=
O\!\left(\frac{p}{n}\right).
\]
Markov's inequality then yields $A_{1n}/(\sigma_L^2\sqrt p)\to 0$ in probability.
\end{proof}

\begin{lemma}
\label{lem:weighted_remainder}
Under Assumption~\ref{ass:oracle}. Suppose furthermore that
\[
p\to\infty,
\qquad
p=O(n^\kappa)
\quad\text{for some }\kappa\in(0,2).
\]
Define
\[
R_{U,n}:=n\sum_{j=1}^p \hat C_j^2\bigl(\hat\sigma_{U,j}^{-2}-p\bigr).
\]
Then
\[
\frac{R_{U,n}}{\sigma_L^2}=o_p(\sqrt p).
\]
\end{lemma}

\begin{proof}
Write
\[
\Delta_j:=\hat\sigma_{U,j}^2-\frac1p.
\]
By Lemma~\ref{lem:lin_basic_bounds},
\[
\max_{1\le j\le p}|\Delta_j|=O_p\!\left(\frac{\sqrt{\log p}}{p\sqrt n}+\frac{\log p}{np}\right),
\qquad
\max_{1\le j\le p}|p\Delta_j|=O_p\!\left(\sqrt{\frac{\log p}{n}}+\frac{\log p}{n}\right)=o_p(1).
\]
Hence, on an event whose probability tends to one,
\[
\hat\sigma_{U,j}^{-2}=p-p^2\Delta_j+\rho_j,
\qquad
|\rho_j|\le Cp^3\Delta_j^2,
\]
so that
\[
\max_{1\le j\le p}|\rho_j|=O_p\!\left(\frac{p\log p}{n}+\frac{p(\log p)^2}{n^2}\right).
\]
Decompose
\[
R_{U,n}=R_{1n}+R_{2n},
\qquad
R_{1n}:=-p^2n\sum_{j=1}^p \hat C_j^2\Delta_j,
\qquad
R_{2n}:=n\sum_{j=1}^p \hat C_j^2\rho_j.
\]
Now
\[
\hat C_j=\sigma_L(W_j-\bar Z\bar U_j),
\qquad
\|\hat{\bm{C}}\|^2\le 2\sigma_L^2\|\bm{W}_n\|^2+2\sigma_L^2\bar Z^2\|\bar{\bm{U}}\|^2.
\]
Since $\|\bm{W}_n\|^2=O_p(n^{-1})$, $\bar Z=O_p(n^{-1/2})$, and $\|\bar{\bm{U}}\|^2=O_p(n^{-1})$, we have $\|\hat{\bm{C}}\|^2=\sigma_L^2O_p(n^{-1})$, hence
\[
\frac{n}{\sigma_L^2}\sum_{j=1}^p \hat C_j^2=O_p(1).
\]
Therefore
\[
\frac{|R_{2n}|}{\sigma_L^2}
\le
\left(\frac{n}{\sigma_L^2}\sum_{j=1}^p \hat C_j^2\right)\max_{1\le j\le p}|\rho_j|
=
O_p\!\left(\frac{p\log p}{n}+\frac{p(\log p)^2}{n^2}\right)
=
o_p(\sqrt p),
\]
because $(\sqrt p\,\log p)/n\to 0$ and $(\sqrt p\,(\log p)^2)/n^2\to 0$ under $p=O(n^\kappa)$ with $\kappa<2$.

Next, define
\[
R_{1n}^{\circ}:=-p^2n\sigma_L^2\sum_{j=1}^p W_j^2\Delta_j.
\]
Since
\[
\hat C_j^2-\sigma_L^2W_j^2=-2\sigma_L^2W_j\bar Z\bar U_j+\sigma_L^2\bar Z^2\bar U_j^2,
\]
Cauchy--Schwarz yields
\[
\frac{1}{\sigma_L^2}\sum_{j=1}^p |\hat C_j^2-\sigma_L^2W_j^2|
\le
2|\bar Z|\,\|\bm{W}_n\|\,\|\bar{\bm{U}}\|+\bar Z^2\|\bar{\bm{U}}\|^2
=
O_p(n^{-3/2}).
\]
Hence
\[
\begin{aligned}
\frac{|R_{1n}-R_{1n}^{\circ}|}{\sigma_L^2}
&\le
p^2n\left(\max_{1\le j\le p}|\Delta_j|\right)
\frac{1}{\sigma_L^2}\sum_{j=1}^p |\hat C_j^2-\sigma_L^2W_j^2| \\
&=
O_p\!\left(\frac{p\sqrt{\log p}}{n}\right)
+O_p\!\left(\frac{p\log p}{n^{3/2}}\right)
=
o_p(\sqrt p).
\end{aligned}
\]
Also,
\[
\Delta_j=\frac1n\sum_{i=1}^n \eta_{ij}-\bar U_j^2,
\qquad
\eta_{ij}=U_{ij}^2-\frac1p.
\]
Therefore
\[
R_{1n}^{\circ}=A_n+B_n,
\]
where
\[
A_n:=-p^2\sigma_L^2\sum_{i=1}^n\sum_{j=1}^p W_j^2\eta_{ij},
\qquad
B_n:=np^2\sigma_L^2\sum_{j=1}^p W_j^2\bar U_j^2.
\]
Using $\|\bm{W}_n\|^2=O_p(n^{-1})$ and $\max_j\bar U_j^2=O_p((\log p)/(np))$, we get
\[
\frac{|B_n|}{\sigma_L^2}
\le
np^2\|\bm{W}_n\|^2\max_{1\le j\le p}\bar U_j^2
=
O_p\!\left(\frac{p\log p}{n}\right)
=
o_p(\sqrt p).
\]
Now define
\[
W_j^{(-i)}:=\frac1n\sum_{k\ne i} Z_kU_{kj},
\qquad
W_j=W_j^{(-i)}+\frac1n Z_iU_{ij}.
\]
Expanding $W_j^2$ gives $A_n=A_{1n}+A_{2n}+A_{3n}$, where
\[
A_{1n}=-p^2\sigma_L^2\sum_{i=1}^n\sum_{j=1}^p \bigl(W_j^{(-i)}\bigr)^2\eta_{ij},
\]
\[
A_{2n}=-\frac{2p^2\sigma_L^2}{n}\sum_{i=1}^n Z_i\sum_{j=1}^p U_{ij}W_j^{(-i)}\eta_{ij},
\qquad
A_{3n}=-\frac{p^2\sigma_L^2}{n^2}\sum_{i=1}^n Z_i^2\sum_{j=1}^p U_{ij}^2\eta_{ij}.
\]
Lemma~\ref{lem:A1n_second_moment} yields $A_{1n}/\sigma_L^2=o_p(\sqrt p)$. For $A_{2n}$, write
\[
A_{2n}=-\frac{2p^2\sigma_L^2}{n^2}\sum_{i\ne k} Z_iZ_kB_{ik},
\qquad
B_{ik}:=\sum_{j=1}^p U_{ij}U_{kj}\eta_{ij}.
\]
Conditioning on $\bm{U}_k$, set $a_j=U_{kj}$ and $V_{ij}=U_{ij}\eta_{ij}=U_{ij}^3-p^{-1}U_{ij}$. By spherical symmetry, $\E(V_{ij})=0$ and $\E(V_{ij}V_{im})=0$ for $j\ne m$. Using
\[
\E(U_{ij}^2)=\frac1p,
\qquad
\E(U_{ij}^4)=\frac{3}{p(p+2)},
\qquad
\E(U_{ij}^6)=\frac{15}{p(p+2)(p+4)},
\]
one gets $\E(V_{ij}^2)=O(p^{-3})$, hence $\E(B_{ik}^2\mid \bm{U}_k)\le Cp^{-3}$ and so $\E(B_{ik}^2)=O(p^{-3})$. Expanding $\E\{(A_{2n}/\sigma_L^2)^2\}$ and using independence and centering of the $Z$'s shows that only index pairings with $\{r,s\}=\{i,k\}$ survive, whence
\[
\E\left\{\left(\frac{A_{2n}}{\sigma_L^2}\right)^2\right\}=O\!\left(\frac{p^4}{n^4}\cdot n^2\cdot p^{-3}\right)=O\!\left(\frac{p}{n^2}\right).
\]
Thus $A_{2n}/\sigma_L^2=o_p(\sqrt p)$. For $A_{3n}$, set
\[
T_i:=\sum_{j=1}^p U_{ij}^2\eta_{ij}=\sum_{j=1}^p U_{ij}^4-\frac1p.
\]
Since $\E\sum_{j=1}^p U_{ij}^4=3/(p+2)=O(p^{-1})$, we have $\E|T_i|=O(p^{-1})$. Therefore
\[
\E\left|\frac{A_{3n}}{\sigma_L^2}\right|
\le
\frac{p^2}{n^2}\sum_{i=1}^n \E(Z_i^2)\E|T_i|
=
O\!\left(\frac{p}{n}\right),
\]
so $A_{3n}/\sigma_L^2=o_p(\sqrt p)$ by Markov's inequality. Hence $A_n/\sigma_L^2=o_p(\sqrt p)$, and therefore
\[
\frac{R_{1n}^{\circ}}{\sigma_L^2}=o_p(\sqrt p),
\qquad
\frac{R_{1n}}{\sigma_L^2}=o_p(\sqrt p).
\]
Combining this with $R_{2n}/\sigma_L^2=o_p(\sqrt p)$ gives $R_{U,n}/\sigma_L^2=o_p(\sqrt p)$.
\end{proof}

\subsubsection{Proof of Theorem~\ref{thm:sum-clt}}

\begin{proof}[Proof of Theorem~\ref{thm:sum-clt}]
Recall the idealized statistic
\[
\widetilde T_{\mathrm{sum}}=\frac{np}{\hat\sigma_L^2}\sum_{j=1}^p \hat C_j^2.
\]
By definition,
\[
T_{\mathrm{sum}}-\widetilde T_{\mathrm{sum}}
=
\frac1{\hat\sigma_L^2}
\underbrace{n\sum_{j=1}^p \hat C_j^2\bigl(\hat\sigma_{U,j}^{-2}-p\bigr)}_{R_{U,n}}.
\]
Lemma~\ref{lem:weighted_remainder} shows that $R_{U,n}/\sigma_L^2=o_p(\sqrt p)$. Since $\hat\sigma_L^2/\sigma_L^2=1+O_p(n^{-1/2})$,
\[
T_{\mathrm{sum}}-\widetilde T_{\mathrm{sum}}=\frac{R_{U,n}/\sigma_L^2}{\hat\sigma_L^2/\sigma_L^2}=o_p(\sqrt p),
\qquad
\frac{T_{\mathrm{sum}}-\widetilde T_{\mathrm{sum}}}{\sqrt{2p}}\xrightarrow{p}0.
\]
On the other hand, Proposition~\ref{prop:sum_clt_stage1} gives
\[
\frac{\widetilde T_{\mathrm{sum}}-p}{\sqrt{2p}}\xrightarrow{d}N(0,1).
\]
Therefore
\[
\frac{T_{\mathrm{sum}}-p}{\sqrt{2p}}
=
\frac{\widetilde T_{\mathrm{sum}}-p}{\sqrt{2p}}
+
\frac{T_{\mathrm{sum}}-\widetilde T_{\mathrm{sum}}}{\sqrt{2p}}
\xrightarrow{d}N(0,1)
\]
by Slutsky's theorem.
\end{proof}

\subsection{Proof of Theorem~\ref{thm:max-gumbel}}

\subsubsection{Auxiliary lemmas}

\begin{lemma}
\label{lem:max_reduction}
Let
\[
\bm{S}_n:=\frac1{\sqrt n}\sum_{i=1}^n \bm{\xi}_i
=
\left(
\frac1{\sqrt n}\sum_{i=1}^n \xi_{i1},\dots,\frac1{\sqrt n}\sum_{i=1}^n \xi_{ip}
\right)^{\top}.
\]
Under the assumptions of Theorem~\ref{thm:max-gumbel}, let $M_n=\sqrt n\,\|\bm{g}_n^{\rm or}\|_\infty$. Then
\[
M_n=\|\bm{S}_n\|_\infty+\varepsilon_n,
\qquad
\varepsilon_n=o_p\bigl((\log p)^{-1/2}\bigr),
\]
and
\[
\|\bm{S}_n\|_\infty=O_p(\sqrt{\log p}).
\]
Consequently,
\begin{equation}
\label{eq:max_square_reduction_new}
T_{\max}=\|\bm{S}_n\|_\infty^2-2\log p+\log\log p+o_p(1).
\end{equation}
\end{lemma}

\begin{proof}
By Theorem~\ref{thm:linearization},
\[
\sqrt n\,\hat\gamma_j=\frac1{\sqrt n}\sum_{i=1}^n \xi_{ij}+\sqrt n\,r_{n,j},
\]
so
\[
M_n=\|\bm{S}_n\|_\infty+\varepsilon_n,
\qquad
|\varepsilon_n|\le \sqrt n\max_{1\le j\le p}|r_{n,j}|
=O_p\!\left(\frac{\log p+\sqrt{\log p}}{\sqrt n}\right)
=O_p\!\left(\frac{\log p}{\sqrt n}\right).
\]
Moreover,
\[
\frac{(\log p)/\sqrt n}{(\log p)^{-1/2}}
=
\left\{\frac{(\log p)^3}{n}\right\}^{1/2}
\longrightarrow 0,
\]
because $\log p=o(n^{1/5})$. Hence $\varepsilon_n=o_p((\log p)^{-1/2})$.

Next, conditional on $Z_1,\ldots,Z_n$, for every $j$,
\[
\frac1{\sqrt n}\sum_{i=1}^n\xi_{ij}
=
\frac{\sqrt p}{\sqrt n}\sum_{i=1}^n Z_iU_{ij}.
\]
The weighted spherical concentration inequality gives
\[
\Pr\left(
\left|\frac1{\sqrt n}\sum_{i=1}^n\xi_{ij}\right|>t
\,\bigg|\,Z_1,\ldots,Z_n\right)
\le
2\exp\left(-\frac{ct^2}{V_n}\right),
\qquad
V_n=n^{-1}\sum_{i=1}^n Z_i^2.
\]
Since $V_n=1+O_p(n^{-1/2})$, on the event $V_n\le2$,
\[
\Pr\left(\|\bm{S}_n\|_\infty>t\mid Z_1,\ldots,Z_n\right)
\le 2p\exp(-ct^2/2).
\]
Taking $t=C\sqrt{\log p}$ with $C$ large enough gives
\[
\Prob(\|\bm{S}_n\|_\infty>C\sqrt{\log p})\to 0,
\]
whence $\|\bm{S}_n\|_\infty=O_p(\sqrt{\log p})$.

Finally,
\[
M_n^2-\|\bm{S}_n\|_\infty^2=(M_n-\|\bm{S}_n\|_\infty)(M_n+\|\bm{S}_n\|_\infty).
\]
The first factor is $o_p((\log p)^{-1/2})$, while the second is $O_p(\sqrt{\log p})$. Therefore $M_n^2-\|\bm{S}_n\|_\infty^2=o_p(1)$, which proves \eqref{eq:max_square_reduction_new}.
\end{proof}

\begin{proposition}
\label{prop:koike_hyperrectangle}
Let
\[
\bm{S}_n=\frac1{\sqrt n}\sum_{i=1}^n \bm{X}_i\in\R^p,
\]
where $\bm{X}_1,\dots,\bm{X}_n$ are independent centered random vectors with
\[
\frac1n\sum_{i=1}^n\Cov(\bm{X}_i)=\mathbf{\Sigma},
\qquad
\min_{1\le j\le p}\mathbf{\Sigma}_{jj}\ge b>0.
\]
Assume that, for a deterministic $B_n\ge1$,
\[
\max_{1\le i\le n}\max_{1\le j\le p}\|X_{ij}\|_{\psi_1}\le B_n,
\qquad
\frac{B_n^2(\log p)^5}{n}\to0.
\]
Let $\bm{G}\sim N(0,\mathbf{\Sigma})$. Then
\[
\sup_{A\in\mathcal A^{\mathrm{re}}}
\bigl|\Prob(\bm{S}_n\in A)-\Prob(\bm{G}\in A)\bigr|\to 0,
\]
where $\mathcal A^{\mathrm{re}}$ denotes the class of all hyperrectangles in $\R^p$. In particular,
\[
\sup_{t\in\R}
\bigl|\Prob(\|\bm{S}_n\|_\infty\le t)-\Prob(\|\bm{G}\|_\infty\le t)\bigr|\to 0.
\]
\end{proposition}

\begin{proof}
Koike's hyperrectangle Gaussian approximation bound gives
\[
\sup_{A\in\mathcal A^{\mathrm{re}}}
\bigl|\Prob(\bm{S}_n\in A)-\Prob(\bm{G}\in A)\bigr|
\le
C\left\{\frac{B_n^2(\log p)^5}{n}\right\}^{c_0}
\]
for numerical constants $C,c_0>0$ depending only on $b$ \citep[Theorem~2.1]{Koike2021}.  The right hand side converges to zero.  The last assertion follows by taking $A=[-t,t]^p$.
\end{proof}

\begin{lemma}
\label{lem:gaussian_approx_max}
Under the assumptions of Theorem~\ref{thm:max-gumbel}, if
\[
\bm{G}=(G_1,\dots,G_p)^{\top}\sim N(0,\mathbf{I}_p),
\]
then
\begin{equation}
\label{eq:gaussian_approx_max_new}
\sup_{t\in\R}
\left|
\Prob(\|\bm{S}_n\|_\infty\le t)-\Prob(\|\bm{G}\|_\infty\le t)
\right|\longrightarrow 0.
\end{equation}
\end{lemma}

\begin{proof}
Let
\[
V_n=\frac1n\sum_{i=1}^nZ_i^2,
\qquad
\mathcal E_n=\left\{\frac12\le V_n\le2,\\ \max_{1\le i\le n}|Z_i|\le b_{n,p}\right\}.
\]
Assumption~\ref{ass:oracle}\textnormal{(ii)} gives $\Pr(\mathcal E_n)\to1$.  Conditional on $Z_1,\ldots,Z_n$, define
\[
\bm{X}_i=V_n^{-1/2}\sqrt p\,Z_i\bm{U}_i.
\]
Then
\[
\E(\bm{X}_i\mid Z_1,\ldots,Z_n)=\bm{0},
\qquad
\frac1n\sum_{i=1}^n\Cov(\bm{X}_i\mid Z_1,\ldots,Z_n)=\mathbf{I}_p.
\]
On $\mathcal E_n$,
\[
\max_{i,j}\|X_{ij}\|_{\psi_1\mid Z_1,\ldots,Z_n}
\le C\max_i |Z_i|V_n^{-1/2}
\le Cb_{n,p},
\]
because $\sqrt p\,U_{ij}$ has a bounded sub-Gaussian, hence bounded sub-exponential, norm uniformly in $p$ and $j$.  Proposition~\ref{prop:koike_hyperrectangle}, applied conditionally, gives
\[
\sup_{t\in\R}
\left|
\Pr\left(V_n^{-1/2}\|\bm{S}_n\|_\infty\le t\mid Z_1,\ldots,Z_n\right)
-
\Pr(\|\bm{G}\|_\infty\le t)
\right|
\le
C\left\{\frac{b_{n,p}^2(\log p)^5}{n}\right\}^{c_0}
\]
on $\mathcal E_n$.  Therefore
\[
\sup_t\left|
\Pr\left(V_n^{-1/2}\|\bm{S}_n\|_\infty\le t\right)
-
\Pr(\|\bm{G}\|_\infty\le t)
\right|\to0.
\]
It remains to remove the factor $V_n^{-1/2}$.  Since $\E(V_n-1)^2=O(n^{-1})$,
\[
|V_n^{1/2}-1|=O_p(n^{-1/2}).
\]
The Gaussian anti-concentration bound for the maximum of $p$ standard normals gives, for $0<\varepsilon<1/2$,
\[
\sup_{t\in\R}
\Pr\{t<\|\bm{G}\|_\infty\le t(1+\varepsilon)\}
\le C\varepsilon\sqrt{\log p}.
\]
With $\varepsilon_n=|V_n^{1/2}-1|$, $\varepsilon_n\sqrt{\log p}=O_p(\sqrt{\log p/n})=o_p(1)$, and hence
\[
\sup_t\left|
\Pr(\|\bm{S}_n\|_\infty\le t)
-
\Pr(V_n^{-1/2}\|\bm{S}_n\|_\infty\le t)
\right|\to0.
\]
Combining the last two displays proves \eqref{eq:gaussian_approx_max_new}.
\end{proof}

\begin{lemma}
\label{lem:gaussian_extreme_value}
Let $\bm{G}\sim N(0,\mathbf{I}_p)$ and define
\[
u_p(x):=\sqrt{2\log p-\log\log p+x},\qquad x\in\R.
\]
Then, for every fixed $x\in\R$,
\[
\Prob\bigl(\|\bm{G}\|_\infty^2-2\log p+\log\log p\le x\bigr)
\longrightarrow
\exp\!\bigl(-\pi^{-1/2}e^{-x/2}\bigr).
\]
\end{lemma}

\begin{proof}
Since $G_1,\dots,G_p$ are i.i.d. $N(0,1)$,
\[
\Prob(\|\bm{G}\|_\infty\le u)=\bigl(1-2\bar\Phi(u)\bigr)^p,
\qquad u>0,
\]
where $\bar\Phi(u)=1-\Phi(u)$. Taking $u=u_p(x)$ gives
\[
\Prob\bigl(\|\bm{G}\|_\infty^2-2\log p+\log\log p\le x\bigr)
=
\bigl(1-2\bar\Phi(u_p(x))\bigr)^p.
\]
By Mills' ratio,
\[
\bar\Phi(u)=\frac{\phi(u)}{u}\{1+o(1)\},
\qquad
\phi(u)=\frac{1}{\sqrt{2\pi}}e^{-u^2/2},
\qquad u\to\infty.
\]
Since $u_p(x)^2=2\log p-\log\log p+x$,
\[
e^{-u_p(x)^2/2}=p^{-1}(\log p)^{1/2}e^{-x/2},
\qquad
u_p(x)\sim \sqrt{2\log p}.
\]
Hence
\[
2p\bar\Phi(u_p(x))
\sim
2p\frac{\phi(u_p(x))}{u_p(x)}
=
2\cdot \frac{1}{\sqrt{2\pi}}\cdot \frac{(\log p)^{1/2}}{\sqrt{2\log p}}e^{-x/2}
=
\pi^{-1/2}e^{-x/2}.
\]
Therefore $a_p:=2\bar\Phi(u_p(x))\to 0$ and $pa_p\to \pi^{-1/2}e^{-x/2}$, so $(1-a_p)^p\to \exp(-\pi^{-1/2}e^{-x/2})$.
\end{proof}

\subsubsection{Proof of Theorem~\ref{thm:max-gumbel}}

\begin{proof}[Proof of Theorem~\ref{thm:max-gumbel}]
Lemma~\ref{lem:max_reduction} gives
\[
T_{\max}=\|\bm{S}_n\|_\infty^2-2\log p+\log\log p+o_p(1).
\]
Fix $x\in\R$ and set $u=u_p(x)$. By Lemma~\ref{lem:gaussian_approx_max},
\[
\left|
\Prob(\|\bm{S}_n\|_\infty\le u)-\Prob(\|\bm{G}\|_\infty\le u)
\right|\to 0.
\]
Equivalently,
\[
\Prob\bigl(\|\bm{S}_n\|_\infty^2-2\log p+\log\log p\le x\bigr)
-
\Prob\bigl(\|\bm{G}\|_\infty^2-2\log p+\log\log p\le x\bigr)
\to 0.
\]
Lemma~\ref{lem:gaussian_extreme_value} now yields
\[
\Prob\bigl(\|\bm{S}_n\|_\infty^2-2\log p+\log\log p\le x\bigr)
\longrightarrow
\exp\!\bigl(-\pi^{-1/2}e^{-x/2}\bigr).
\]
Combining this with the reduction above and applying Slutsky's theorem proves
\[
\Prob(T_{\max}\le x)
\longrightarrow
\exp\!\bigl(-\pi^{-1/2}e^{-x/2}\bigr),
\qquad x\in\R.
\]
\end{proof}
\subsection{Proof of Theorem~\ref{thm:sum-max-independence}}
\subsubsection{Auxiliary lemmas}

\begin{lemma}
\label{lem:sum_to_Hn}
Under Assumption~\ref{ass:oracle}, and suppose that
\[
p\to\infty,
\qquad
p=O(n^\kappa)
\quad\text{for some }\kappa\in(0,2).
\]
Then
\[
\frac{T_{\mathrm{sum}}-p}{\sqrt{2p}}
=
\frac{H_n}{\sqrt{2p}}+o_p(1).
\]
\end{lemma}

\begin{proof}
By Lemma~\ref{lem:weighted_remainder} and the identity
\[
T_{\mathrm{sum}}-\widetilde T_{\mathrm{sum}}
=\hat\sigma_L^{-2}R_{U,n},
\]
used in the proof of Theorem~\ref{thm:sum-clt}, we have
\[
T_{\mathrm{sum}}-\widetilde T_{\mathrm{sum}}=o_p(\sqrt p),
\]
and by Lemma~\ref{lem:ideal_centering},
\[
\widetilde T_{\mathrm{sum}}-\widetilde T_{\mathrm{sum}}^{\circ}=o_p(\sqrt p).
\]
Hence
\[
T_{\mathrm{sum}}-\widetilde T_{\mathrm{sum}}^{\circ}=o_p(\sqrt p),
\qquad
\frac{T_{\mathrm{sum}}-\widetilde T_{\mathrm{sum}}^{\circ}}{\sqrt{2p}}=o_p(1).
\]

Next, by Lemma~\ref{lem:ideal_identity},
\[
\widetilde T_{\mathrm{sum}}^{\circ}-p
=
\frac{p\bar Z^2+H_n}{1+A_n-\bar Z^2}.
\]
Therefore,
\[
\frac{\widetilde T_{\mathrm{sum}}^{\circ}-p}{\sqrt{2p}}
=
\frac{H_n}{\sqrt{2p}}\cdot \frac{1}{1+A_n-\bar Z^2}
+
\frac{p\bar Z^2}{\sqrt{2p}}\cdot \frac{1}{1+A_n-\bar Z^2}.
\]
By Lemma~\ref{lem:ideal_identity},
\[
\frac{p\bar Z^2}{\sqrt{2p}}=o_p(1),
\qquad
1+A_n-\bar Z^2\xrightarrow{p}1.
\]
Since \(H_n/\sqrt{2p}=O_p(1)\) by Lemma~\ref{lem:Hn_martingale_clt}, Slutsky's theorem gives
\[
\frac{\widetilde T_{\mathrm{sum}}^{\circ}-p}{\sqrt{2p}}
=
\frac{H_n}{\sqrt{2p}}+o_p(1).
\]
Combining the two reductions yields the claim.
\end{proof}

\begin{lemma}
\label{lem:joint_fidi_Hn_Sn}
Under Assumption~\ref{ass:oracle}, $p\to\infty$, and $p=O(n^\kappa)$ for some $\kappa\in(0,2)$. Fix an integer \(d\ge 1\), and let \(1\le j_1<\cdots<j_d\le p\) be distinct indices. Then
\[
\left(
\frac{H_n}{\sqrt{2p}},\ S_{n,j_1},\dots,S_{n,j_d}
\right)
\xrightarrow{d}
(N_0,N_1,\dots,N_d),
\]
where \(N_0\sim N(0,1)\), \((N_1,\dots,N_d)\sim N(0,\mathbf{I}_d)\), and \(N_0\) is independent of \((N_1,\dots,N_d)\).
\end{lemma}

\begin{proof}
Let
\[
\bm{S}_{n,J}:=(S_{n,j_1},\dots,S_{n,j_d})^\top,
\qquad
T_n(\bm{a}):=\frac{H_n}{\sqrt{2p}}+\bm{a}^\top \bm{S}_{n,J},\qquad \bm{a}\in\R^d.
\]
We prove that for every fixed \(\bm{a}\in\R^d\),
\[
T_n(\bm{a})\xrightarrow{d}N\bigl(0,1+\|\bm{a}\|_2^2\bigr).
\]
The Cram\'er--Wold device then yields the stated joint convergence.

Recall from the proof of Lemma~\ref{lem:Hn_martingale_clt} that
\[
H_n=\sum_{k=2}^n M_{n,k},
\qquad
M_{n,k}:=\frac{2}{n}\bm{R}_{k-1}^\top \bm{\xi}_k,
\qquad
\bm{R}_{k-1}:=\sum_{i=1}^{k-1}\bm{\xi}_i,
\qquad
\mathcal F_k:=\sigma(\bm{\xi}_1,\dots,\bm{\xi}_k).
\]
Since
\[
\bm{a}^\top \bm{S}_{n,J}=\sum_{k=1}^n \frac{1}{\sqrt n}\bm{a}^\top \bm{\xi}_{k,J},
\qquad
\bm{\xi}_{k,J}:=(\xi_{k j_1},\dots,\xi_{k j_d})^\top,
\]
we may write
\[
T_n(\bm{a})=\sum_{k=1}^n D_{n,k}(\bm{a}),
\]
where
\[
D_{n,k}(\bm{a}):=\frac{2}{n\sqrt{2p}}\bm{R}_{k-1}^\top \bm{\xi}_k+\frac{1}{\sqrt n}\bm{a}^\top \bm{\xi}_{k,J},
\qquad k=1,\dots,n,
\]
with the convention \(\bm{R}_0=\bm{0}\). Since \(\E(\bm{\xi}_k\mid\mathcal F_{k-1})=0\), \(\{D_{n,k}(\bm{a}),\mathcal F_k\}\) is a martingale difference array.

Let
\[
V_n(\bm{a}):=\sum_{k=1}^n \E\bigl(D_{n,k}(\bm{a})^2\mid\mathcal F_{k-1}\bigr)=V_{n,1}+V_{n,2}+V_{n,3},
\]
where
\[
V_{n,1}:=\sum_{k=1}^n \E\left\{\left(\frac{2}{n\sqrt{2p}}\bm{R}_{k-1}^\top\bm{\xi}_k\right)^2\Bigm|\mathcal F_{k-1}\right\},
\]
\[
V_{n,2}:=\sum_{k=1}^n \E\left\{\left(\frac{1}{\sqrt n}\bm{a}^\top\bm{\xi}_{k,J}\right)^2\Bigm|\mathcal F_{k-1}\right\},
\]
\[
V_{n,3}:=2\sum_{k=1}^n \E\left\{\frac{2}{n\sqrt{2p}}\bm{R}_{k-1}^\top\bm{\xi}_k\cdot \frac{1}{\sqrt n}\bm{a}^\top\bm{\xi}_{k,J}\Bigm|\mathcal F_{k-1}\right\}.
\]
The first term is exactly the predictable quadratic variation from Lemma~\ref{lem:Hn_martingale_clt}, so
\[
V_{n,1}\xrightarrow{p}1.
\]
For the second term, since \(\E(\bm{\xi}_i\bm{\xi}_i^\top)=\mathbf{I}_p\), we have
\[
V_{n,2}=\sum_{k=1}^n \frac{1}{n}\,\bm{a}^\top \mathbf{I}_d \bm{a}=\|\bm{a}\|_2^2.
\]
For the cross term,
\[
V_{n,3}=\frac{4}{n^{3/2}\sqrt{2p}}\sum_{k=2}^n \bm{a}^\top \bm{R}_{k-1,J},
\qquad
\bm{R}_{k-1,J}:=(R_{k-1,j_1},\dots,R_{k-1,j_d})^\top.
\]
Now, for each fixed \(j\),
\[
\sum_{k=2}^n R_{k-1,j}=\sum_{i=1}^{n-1}(n-i)\xi_{ij},
\]
whence
\[
\E\left(\sum_{k=2}^n R_{k-1,j}\right)^2
=
\sum_{i=1}^{n-1}(n-i)^2\E(\xi_{ij}^2)=O(n^3).
\]
Since \(d\) is fixed, it follows that
\[
\E(V_{n,3}^2)=O\!\left(\frac{1}{p}\right)\to 0,
\qquad
V_{n,3}=o_p(1).
\]
Therefore
\[
V_n(\bm{a})\xrightarrow{p}1+\|\bm{a}\|_2^2.
\]

It remains to verify the Lindeberg condition. By \(|x+y|^4\le 8|x|^4+8|y|^4\),
\[
\sum_{k=1}^n \E|D_{n,k}(\bm{a})|^4
\lesssim
\sum_{k=1}^n \E\left|\frac{2}{n\sqrt{2p}}\bm{R}_{k-1}^\top\bm{\xi}_k\right|^4
+
\sum_{k=1}^n \E\left|\frac{1}{\sqrt n}\bm{a}^\top\bm{\xi}_{k,J}\right|^4.
\]
The first term is \(o(1)\) by the fourth-moment bound already established at the end of the proof of Lemma~\ref{lem:Hn_martingale_clt}. Since \(d\) is fixed and the coordinates \(\xi_{ij}\) have uniformly bounded fourth moments, the second term is bounded by
\[
Cn\cdot n^{-2}=O(n^{-1})\to 0.
\]
Thus the martingale Lindeberg condition holds. Hall--Heyde's martingale CLT \cite[Corollary~3.1]{HallHeyde1980} gives
\[
T_n(\bm{a})\xrightarrow{d}N\bigl(0,1+\|\bm{a}\|_2^2\bigr).
\]
The desired joint convergence now follows by the Cram\'er--Wold device.
\end{proof}

\begin{lemma}
\label{lem:block_decomp_uniform_direction}
Fix an integer $d\ge 1$ and a coordinate set $J=\{j_1,\dots,j_d\}\subset\{1,\dots,p\}$ with $|J|=d$. Then for each $i$ there exist random objects
\[
B_i\sim \mathrm{Beta}(d/2,(p-d)/2),
\qquad
\bm{V}_i\sim \Unif(\Sphere^{d-1}),
\qquad
\bm{W}_i\sim \Unif(\Sphere^{p-d-1}),
\]
independent of each other and independent of $Z_i$, such that
\[
\bm{U}_{i,J}=\sqrt{B_i}\,\bm{V}_i,
\qquad
\bm{U}_{i,J^c}=\sqrt{1-B_i}\,\bm{W}_i.
\]
Consequently, if we set
\[
a_i:=\sqrt p\,Z_i\sqrt{B_i},
\qquad
b_i:=\sqrt p\,Z_i\sqrt{1-B_i},
\]
then
\[
\bm{S}_{n,J}:=(S_{n,j_1},\dots,S_{n,j_d})^\top
=
\frac{1}{\sqrt n}\sum_{i=1}^n a_i\bm{V}_i,
\]
and
\[
H_n
=
H_{n,J}+H_{n,J^c},
\qquad
H_{n,J}:=\frac{2}{n}\sum_{1\le i<k\le n} a_ia_k\bm{V}_i^\top \bm{V}_k,
\qquad
H_{n,J^c}:=\frac{2}{n}\sum_{1\le i<k\le n} b_ib_k\bm{W}_i^\top \bm{W}_k.
\]
Moreover, with
\[
\mathcal G_n:=\sigma\{(Z_i,B_i):1\le i\le n\},
\]
the random vector $\bm{S}_{n,J}$ and the scalar $H_{n,J^c}$ are conditionally independent given $\mathcal G_n$.
\end{lemma}

\begin{proof}
Let $\bm{G}_i=(\bm{G}_{i,J}^\top,\bm{G}_{i,J^c}^\top)^\top\sim N_p(0,\mathbf{I}_p)$, and write
\[
\bm{U}_i=\frac{\bm{G}_i}{\|\bm{G}_i\|}
=\left(\frac{\bm{G}_{i,J}}{\|\bm{G}_i\|},\frac{\bm{G}_{i,J^c}}{\|\bm{G}_i\|}\right).
\]
Set
\[
B_i:=\frac{\|\bm{G}_{i,J}\|^2}{\|\bm{G}_i\|^2},
\qquad
\bm{V}_i:=\frac{\bm{G}_{i,J}}{\|\bm{G}_{i,J}\|},
\qquad
\bm{W}_i:=\frac{\bm{G}_{i,J^c}}{\|\bm{G}_{i,J^c}\|}.
\]
Since $\|\bm{G}_{i,J}\|^2\sim \chi_d^2$, $\|\bm{G}_{i,J^c}\|^2\sim \chi_{p-d}^2$, and these two random variables are independent, we have
\[
B_i\sim \mathrm{Beta}(d/2,(p-d)/2).
\]
Moreover, $\bm{V}_i\sim \Unif(\Sphere^{d-1})$, $\bm{W}_i\sim \Unif(\Sphere^{p-d-1})$, and $B_i$, $\bm{V}_i$, and $\bm{W}_i$ are mutually independent. Therefore
\[
\bm{U}_{i,J}=\sqrt{B_i}\,\bm{V}_i,
\qquad
\bm{U}_{i,J^c}=\sqrt{1-B_i}\,\bm{W}_i.
\]
Multiplying by $\sqrt p\,Z_i$ yields the displayed representations for $\bm{S}_{n,J}$, $H_{n,J}$, and $H_{n,J^c}$. Finally, conditional on
\[
\mathcal G_n=\sigma\{(Z_i,B_i):1\le i\le n\},
\]
the vectors $\bm{V}_1,\dots,\bm{V}_n$ and $\bm{W}_1,\dots,\bm{W}_n$ are independent, and hence $\bm{S}_{n,J}$ and $H_{n,J^c}$ are conditionally independent.
\end{proof}

\begin{lemma}
\label{lem:Jn_part_negligible}
Fix $d\ge 1$ and $J=\{j_1,\dots,j_d\}$. Assume that $p\to\infty$, $p=O(n^\kappa)$ for some fixed $\kappa>0$, and $\log p=o(n^{1/5})$. Let
\[
\Delta_{n,J}:=\frac{H_{n,J}}{\sqrt{2p}}.
\]
Then there exists a constant $C_0>0$ such that, with
\[
\delta_n:=\frac{C_0\log p}{\sqrt p},
\]
we have $\delta_n\to 0$ and
\[
\Prob\bigl(|\Delta_{n,J}|>\delta_n\bigr)
=
o\!\left(\{2\bar\Phi(u_p(y))\}^d\right)
\qquad (n\to\infty),
\]
for every fixed $y\in\R$.
\end{lemma}

\begin{proof}
By the identity
\[
\Bigl\|\frac1{\sqrt n}\sum_{i=1}^n a_i\bm{V}_i\Bigr\|^2
=
\frac1n\sum_{i=1}^n a_i^2+\frac{2}{n}\sum_{1\le i<k\le n} a_ia_k\bm{V}_i^\top \bm{V}_k,
\]
we may write
\[
H_{n,J}=\|\bm{S}_{n,J}\|_2^2-\frac1n\sum_{i=1}^n a_i^2.
\]
Since $d$ is fixed and $B_i\sim\mathrm{Beta}(d/2,(p-d)/2)$, $\sup_p\E(pB_i)^m<\infty$ for every fixed $m$.  Assumption~\ref{ass:oracle}\textnormal{(ii)} gives $\sup_p\E |Z_i|^{8+\eta_Z}<\infty$, hence
\[
\E(a_i^2)=p\E(Z_i^2)\E(B_i)=d,
\qquad
\sup_p\E(a_i^4)<\infty,
\qquad
\frac1n\sum_{i=1}^n a_i^2=O_p(1).
\]
Moreover, the beta--gamma representation of $B_i$ gives
\[
\max_{1\le i\le n}pB_i=O_p(\log n),
\qquad
\max_{1\le i\le n}|a_i|=O_p\{b_{n,p}\sqrt{\log n}\}.
\]
Thus
\[
\frac{\sqrt{\log p}\max_i|a_i|}{\sqrt{\sum_{i=1}^n a_i^2}}
=
O_p\left(\sqrt{\frac{b_{n,p}^2\log p\log n}{n}}\right)=o_p(1).
\]
Conditional on $\mathcal G_n$, $\bm{S}_{n,J}=n^{-1/2}\sum_i a_i\bm{V}_i$ is a fixed-dimensional sum of independent bounded spherical variables. Bernstein's inequality applied on the above coefficient events gives that, for every fixed $M>0$, there is a sufficiently large $C_0$ such that
\[
\Prob(|H_{n,J}|>C_0\log p)=O(p^{-M}).
\]
Here the polynomial relation between $p$ and $n$ is used to express the coefficient-event probability on the $p^{-M}$ scale. Since
\[
2\bar\Phi(u_p(y))\asymp p^{-1},
\]
we have
\[
\{2\bar\Phi(u_p(y))\}^d\asymp p^{-d}.
\]
Choosing $M>d+1$ proves the claim.
\end{proof}

\begin{lemma}
\label{lem:conditional_clt_Jc}
Fix an integer $d\ge 1$, let $J=\{j_1,\dots,j_d\}\subset\{1,\dots,p\}$ with $|J|=d$, and write
\[
q:=p-d,
\qquad
Q_{n,J}:=\frac{H_{n,J^c}}{\sqrt{2p}}.
\]
Let
\[
\mathcal G_n:=\sigma\{(Z_i,B_i):1\le i\le n\},
\]
where $B_i\sim \mathrm{Beta}(d/2,q/2)$ is the block-mass variable in Lemma~\ref{lem:block_decomp_uniform_direction}. Under Assumption~\ref{ass:oracle}, if $p\to\infty$ as $n\to\infty$,
\[
\sup_{t\in\R}
\Bigl|
\Prob\bigl(Q_{n,J}\le t\mid \mathcal G_n\bigr)-\Phi(t)
\Bigr|
\longrightarrow 0
\qquad\text{in }L^1.
\]
Consequently, for every fixed $x\in\R$ and every deterministic sequence $\delta_n\downarrow 0$,
\[
\sup_{|t-x|\le \delta_n}
\Bigl|
\Prob\bigl(Q_{n,J}\le t\mid \mathcal G_n\bigr)-\Phi(t)
\Bigr|
\longrightarrow 0
\qquad\text{in }L^1.
\]
\end{lemma}

\begin{proof}
Condition on $\mathcal G_n$ and set
\[
b_i:=\sqrt p\,Z_i\sqrt{1-B_i},
\qquad
\bm{\eta}_i:=b_i\bm{W}_i\in\R^q,
\qquad i=1,\dots,n,
\]
where $\bm{W}_i\sim \Unif(\Sphere^{q-1})$ are i.i.d. and independent of $\mathcal G_n$. Then, conditionally on $\mathcal G_n$, the vectors $\bm{\eta}_1,\dots,\bm{\eta}_n$ are independent, centered, and satisfy
\[
\E(\bm{\eta}_i\mid \mathcal G_n)=0,
\qquad
\E(\bm{\eta}_i\bm{\eta}_i^\top\mid \mathcal G_n)=\frac{b_i^2}{q}\mathbf{I}_q,
\qquad
\|\bm{\eta}_i\|^2=b_i^2.
\]
By Lemma~\ref{lem:block_decomp_uniform_direction},
\[
H_{n,J^c}=\frac{2}{n}\sum_{1\le i<k\le n}\bm{\eta}_i^\top\bm{\eta}_k.
\]
Define the conditional filtration
\[
\mathcal H_{n,k}:=\sigma(\bm{W}_1,\dots,\bm{W}_k)\vee \mathcal G_n,
\qquad
k=0,1,\dots,n,
\]
and the conditional partial sums
\[
\bm{R}_k^{\circ}:=\sum_{i=1}^k \bm{\eta}_i.
\]
Then
\[
H_{n,J^c}=\sum_{k=2}^n M_{n,k}^{\circ},
\qquad
M_{n,k}^{\circ}:=\frac{2}{n}(\bm{R}_{k-1}^{\circ})^\top \bm{\eta}_k,
\]
so that
\[
Q_{n,J}=\sum_{k=2}^n D_{n,k}^{\circ},
\qquad
D_{n,k}^{\circ}:=\frac{M_{n,k}^{\circ}}{\sqrt{2p}}.
\]
Since $\bm{\eta}_k$ is independent of $\mathcal H_{n,k-1}$ and centered conditional on $\mathcal G_n$, the array
\[
\{D_{n,k}^{\circ},\mathcal H_{n,k}\}_{k=2}^n
\]
is a martingale difference array under the conditional law given $\mathcal G_n$.

We now apply a quantitative martingale central limit theorem. Let
\[
V_{n,J}^{\circ}:=\sum_{k=2}^n \E\bigl\{(D_{n,k}^{\circ})^2\mid \mathcal H_{n,k-1}\bigr\}
\]
be the predictable quadratic variation. By Haeusler's martingale Berry--Esseen bound \cite{Haeusler1988}, there exists a universal constant $C>0$ such that, conditionally on $\mathcal G_n$,
\begin{equation}
\label{eq:Haeusler-conditional}
\sup_{t\in\R}
\Bigl|
\Prob\bigl(Q_{n,J}\le t\mid \mathcal G_n\bigr)-\Phi(t)
\Bigr|
\le
C\Biggl(
\E\bigl\{(V_{n,J}^{\circ}-1)^2\mid \mathcal G_n\bigr\}
+
\sum_{k=2}^n \E\bigl\{|D_{n,k}^{\circ}|^4\mid \mathcal G_n\bigr\}
\Biggr)^{1/5}.
\end{equation}
Therefore it is enough to prove
\begin{equation}
\label{eq:two-targets-to-zero}
\E\bigl(V_{n,J}^{\circ}-1\bigr)^2\to 0,
\qquad
\E\sum_{k=2}^n \E\bigl\{|D_{n,k}^{\circ}|^4\mid \mathcal G_n\bigr\}\to 0.
\end{equation}
If \eqref{eq:two-targets-to-zero} holds, Jensen's inequality gives
\[
\E\sup_{t\in\R}
\Bigl|
\Prob\bigl(Q_{n,J}\le t\mid \mathcal G_n\bigr)-\Phi(t)
\Bigr|
\le
C\Biggl(
\E\bigl(V_{n,J}^{\circ}-1\bigr)^2
+
\E\sum_{k=2}^n \E\bigl\{|D_{n,k}^{\circ}|^4\mid \mathcal G_n\bigr\}
\Biggr)^{1/5}
\to 0.
\]

We first study $V_{n,J}^{\circ}$. Since
\[
\E\bigl\{\bm{\eta}_k\bm{\eta}_k^\top\mid \mathcal G_n\bigr\}=\frac{b_k^2}{q}\mathbf{I}_q,
\]
we have
\[
V_{n,J}^{\circ}
=
\frac{2}{pn^2q}\sum_{k=2}^n b_k^2\,\|\bm{R}_{k-1}^{\circ}\|^2.
\]
Write
\[
V_{n,J}^{\circ}=A_{n,J}^{\circ}+B_{n,J}^{\circ},
\]
where
\[
A_{n,J}^{\circ}:=\frac{2}{pn^2q}\sum_{k=2}^n b_k^2\sum_{i=1}^{k-1} b_i^2,
\qquad
B_{n,J}^{\circ}:=\frac{4}{pn^2q}\sum_{k=2}^n b_k^2\sum_{1\le i<j\le k-1} \bm{\eta}_i^\top\bm{\eta}_j.
\]
Since $\sum_{k=2}^n\sum_{i=1}^{k-1}b_i^2b_k^2=\frac12\{(\sum_{i=1}^n b_i^2)^2-\sum_{i=1}^n b_i^4\}$,
\[
A_{n,J}^{\circ}
=
\frac{1}{pn^2q}
\left\{
\Bigl(\sum_{i=1}^n b_i^2\Bigr)^2-\sum_{i=1}^n b_i^4
\right\}.
\]
Let
\[
T_i:=\frac{b_i^2}{p}=Z_i^2(1-B_i).
\]
Because $d$ is fixed, $\sup_{n,p}\E(1-B_i)^m<\infty$ for every fixed integer $m\ge 1$.  Assumption~\ref{ass:oracle}\textnormal{(ii)} gives
\[
\sup_{n,p}\E|T_i|^{4+\eta_Z/2}<\infty.
\]
Also,
\[
\mu_n:=\E(T_i)=\E(Z_i^2)\E(1-B_i)=\frac{q}{p}=1-\frac{d}{p}.
\]
Hence
\[
A_{n,J}^{\circ}
=
\frac{p}{q}\left\{\bar T_n^{\,2}-\frac1{n^2}\sum_{i=1}^n T_i^2\right\},
\qquad
\bar T_n:=\frac1n\sum_{i=1}^n T_i.
\]
Since $\sup_{n,p}\E|T_i|^4<\infty$, we have
\[
\E(\bar T_n-\mu_n)^4=O(n^{-2}),
\qquad
\E\left(\frac1{n^2}\sum_{i=1}^n T_i^2\right)^2=O(n^{-2}).
\]
Using
\[
\bar T_n^{\,2}-\mu_n^2=(\bar T_n-\mu_n)(\bar T_n+\mu_n)
\]
and the bounded fourth moments of $\bar T_n$ and $\mu_n$, it follows that
\[
\E\bigl(\bar T_n^{\,2}-\mu_n^2\bigr)^2=O(n^{-1}).
\]
Since $\mu_n=q/p$, we conclude that
\[
\E\bigl(A_{n,J}^{\circ}-1\bigr)^2=O(n^{-1}+p^{-2})=o(1).
\]

Next, we bound $B_{n,J}^{\circ}$. Observe that
\[
\sum_{k=2}^n b_k^2\sum_{1\le i<j\le k-1}\bm{\eta}_i^\top\bm{\eta}_j
=
\sum_{1\le i<j\le n}
\Bigl(\sum_{k=j+1}^n b_k^2\Bigr)
\bm{\eta}_i^\top\bm{\eta}_j.
\]
Therefore, conditionally on $\mathcal G_n$,
\[
B_{n,J}^{\circ}
=
\frac{4}{pn^2q}
\sum_{1\le i<j\le n}
 c_{ij}\,\bm{W}_i^\top\bm{W}_j,
\qquad
c_{ij}:=b_i b_j\sum_{k=j+1}^n b_k^2.
\]
Because $\E(\bm{W}_i^\top\bm{W}_j\mid \mathcal G_n)=0$ for $i\neq j$, and because mixed products vanish unless the unordered pairs coincide,
\[
\E\bigl\{(B_{n,J}^{\circ})^2\mid \mathcal G_n\bigr\}
\le
\frac{C}{p^2n^4q^2}
\sum_{1\le i<j\le n} c_{ij}^2\,\E\bigl\{(\bm{W}_i^\top\bm{W}_j)^2\mid \mathcal G_n\bigr\}.
\]
For independent uniform directions on $\Sphere^{q-1}$,
\[
\E\bigl\{(\bm{W}_i^\top\bm{W}_j)^2\mid \mathcal G_n\bigr\}=\frac1q.
\]
Hence
\[
\E\bigl\{(B_{n,J}^{\circ})^2\mid \mathcal G_n\bigr\}
\le
\frac{C}{p^2n^4q^3}
\sum_{1\le i<j\le n} c_{ij}^2.
\]
Since
\[
c_{ij}^2
\le
b_i^2b_j^2\Bigl(\sum_{k=1}^n b_k^2\Bigr)^2,
\]
we obtain
\[
\sum_{1\le i<j\le n} c_{ij}^2
\le
\Bigl(\sum_{i=1}^n b_i^2\Bigr)^4.
\]
Thus
\[
\E(B_{n,J}^{\circ})^2
\le
\frac{C}{p^2n^4q^3}\,\E\Bigl(\sum_{i=1}^n b_i^2\Bigr)^4.
\]
Now $b_i^2=pT_i$ and $\sup_{n,p}\E|T_i|^4<\infty$, so
\[
\E\Bigl(\sum_{i=1}^n b_i^2\Bigr)^4=O(n^4p^4).
\]
Consequently,
\[
\E(B_{n,J}^{\circ})^2=O(p^{-1})=o(1).
\]
Combining the bounds for $A_{n,J}^{\circ}$ and $B_{n,J}^{\circ}$ yields
\[
\E\bigl(V_{n,J}^{\circ}-1\bigr)^2
\le
2\E\bigl(A_{n,J}^{\circ}-1\bigr)^2+2\E(B_{n,J}^{\circ})^2
=o(1).
\]
This proves the first part of \eqref{eq:two-targets-to-zero}.

We turn to the fourth-moment term. Since $\bm{\eta}_k=b_k\bm{W}_k$, for every deterministic $\bm{a}\in\R^q$,
\[
\E\bigl\{(\bm{a}^\top\bm{\eta}_k)^4\mid \mathcal G_n\bigr\}
=
 b_k^4\E(\bm{a}^\top \bm{W}_k)^4
=
\frac{3b_k^4}{q(q+2)}\|\bm{a}\|^4.
\]
Hence
\[
\E\bigl\{(M_{n,k}^{\circ})^4\mid \mathcal G_n\bigr\}
\le
\frac{C b_k^4}{n^4q^2}\,
\E\bigl(\|\bm{R}_{k-1}^{\circ}\|^4\mid \mathcal G_n\bigr).
\]
Now
\[
\bm{R}_{k-1}^{\circ}=\sum_{i=1}^{k-1}\bm{\eta}_i,
\qquad
\|\bm{\eta}_i\|^2=b_i^2,
\qquad
\E(\bm{\eta}_i^\top\bm{\eta}_j)^2\mid \mathcal G_n=\frac{b_i^2b_j^2}{q}\le b_i^2b_j^2.
\]
Expanding $\|\bm{R}_{k-1}^{\circ}\|^4$ and using independence gives
\[
\E\bigl(\|\bm{R}_{k-1}^{\circ}\|^4\mid \mathcal G_n\bigr)
\le
C\Bigl(\sum_{i=1}^{k-1} b_i^2\Bigr)^2.
\]
Therefore
\[
\sum_{k=2}^n \E\bigl\{|D_{n,k}^{\circ}|^4\mid \mathcal G_n\bigr\}
\le
\frac{C}{p^2n^4q^2}
\sum_{k=2}^n b_k^4\Bigl(\sum_{i=1}^{k-1} b_i^2\Bigr)^2
\le
\frac{C}{p^2n^4q^2}
\Bigl(\sum_{k=1}^n b_k^4\Bigr)
\Bigl(\sum_{i=1}^n b_i^2\Bigr)^2.
\]
Taking expectations and applying Cauchy--Schwarz,
\begin{align*}
\E\sum_{k=2}^n \E\bigl\{|D_{n,k}^{\circ}|^4\mid \mathcal G_n\bigr\}
&\le
\frac{C}{p^2n^4q^2}
\Bigl\{\E\Bigl(\sum_{k=1}^n b_k^4\Bigr)^2\Bigr\}^{1/2}
\Bigl\{\E\Bigl(\sum_{i=1}^n b_i^2\Bigr)^4\Bigr\}^{1/2} \\
&=
\frac{C}{p^2n^4q^2}
\cdot O(np^2)\cdot O(n^2p^2)
=O(n^{-1})=o(1).
\end{align*}
Here we used $\sup_{n,p}\E|T_i|^8<\infty$, which implies
\[
\E\Bigl(\sum_{k=1}^n b_k^4\Bigr)^2=O(n^2p^4),
\qquad
\E\Bigl(\sum_{i=1}^n b_i^2\Bigr)^4=O(n^4p^4).
\]
This proves the second part of \eqref{eq:two-targets-to-zero}.

We have therefore shown that the right-hand side of \eqref{eq:Haeusler-conditional} tends to zero in $L^1$, and hence
\[
\sup_{t\in\R}
\Bigl|
\Prob\bigl(Q_{n,J}\le t\mid \mathcal G_n\bigr)-\Phi(t)
\Bigr|
\longrightarrow 0
\qquad\text{in }L^1.
\]
The second displayed statement follows immediately, because
\[
\sup_{|t-x|\le \delta_n}
\Bigl|
\Prob\bigl(Q_{n,J}\le t\mid \mathcal G_n\bigr)-\Phi(t)
\Bigr|
\le
\sup_{t\in\R}
\Bigl|
\Prob\bigl(Q_{n,J}\le t\mid \mathcal G_n\bigr)-\Phi(t)
\Bigr|.
\]
This completes the proof.
\end{proof}

\begin{lemma}
\label{lem:conditional_coeff_regularity}
Fix $d\ge 1$ and a coordinate set $J$ with $|J|=d$. Let
\[
a_i=\sqrt p\,Z_i\sqrt{B_i},
\qquad
s_n^2:=\frac1{dn}\sum_{i=1}^n a_i^2,
\qquad
u_p(y):=\sqrt{2\log p-\log\log p+y}.
\]
Under Assumption~\ref{ass:oracle}, if $p\to\infty$, $p=O(n^\kappa)$ for some fixed $\kappa>0$, and $\log p=o(n^{1/5})$, then
\[
s_n^2=1+O_p(n^{-1/2}),
\qquad
u_p(y)^2|s_n-1|=o_p(1),
\]
and
\[
\frac{u_p(y)\max_{1\le i\le n}|a_i|}{\sqrt{\sum_{i=1}^n a_i^2}}
=o_p(1).
\]
\end{lemma}

\begin{proof}
Because $B_i\sim {\rm Beta}(d/2,(p-d)/2)$ and is independent of $Z_i$,
\[
\E(a_i^2)=p\,\E(Z_i^2)\E(B_i)=d.
\]
For fixed $d$, the beta moments satisfy $\sup_p \E(pB_i)^m<\infty$ for each fixed integer $m\ge 1$.  Assumption~\ref{ass:oracle}\textnormal{(ii)} gives $\sup_p\E Z_i^4<\infty$. Hence
\[
\Var(a_i^2)\le C_d<\infty,
\]
and the ordinary variance bound gives
\[
s_n^2-1
=
\frac1{dn}\sum_{i=1}^n(a_i^2-d)
=
O_p(n^{-1/2}).
\]
It follows that $s_n\to 1$ in probability and, because $u_p(y)^2\asymp \log p$ and $\log p=o(n^{1/5})$,
\[
u_p(y)^2|s_n-1|
=
O_p\!\left(\frac{\log p}{\sqrt n}\right)
=o_p(1).
\]

It remains to control the largest coefficient.  By Assumption~\ref{ass:oracle}\textnormal{(ii)} and the beta--gamma representation of $B_i$,
\[
\max_{1\le i\le n}|Z_i|=O_p(b_{n,p}),
\qquad
\max_{1\le i\le n} pB_i=O_p(\log n).
\]
Therefore
\[
\max_{1\le i\le n}|a_i|
=
\max_{1\le i\le n}\sqrt p\,|Z_i|\sqrt{B_i}
=
O_p\{b_{n,p}\sqrt{\log n}\}.
\]
Since $\sum_{i=1}^n a_i^2=dn\{1+o_p(1)\}$ and $u_p(y)^2\asymp\log p$, we obtain
\[
\frac{u_p(y)\max_{1\le i\le n}|a_i|}{\sqrt{\sum_{i=1}^n a_i^2}}
=
O_p\left(\sqrt{\frac{b_{n,p}^2\log p\log n}{n}}\right)
=o_p(1).
\]
This proves the lemma.
\end{proof}

\begin{lemma}
\label{lem:conditional_md_fixed_block}
Fix $d\ge 1$, $J=\{j_1,\dots,j_d\}$, and $y\in\R$. Let
\[
E_{n,J}(y):=\bigcap_{\ell=1}^d\{|S_{n,j_\ell}|>u_p(y)\}.
\]
Then, under Assumption~\ref{ass:oracle}, $p\to\infty$, $p=O(n^\kappa)$ for some fixed $\kappa>0$, and $\log p=o(n^{1/5})$,
\[
\frac{\Prob(E_{n,J}(y)\mid \mathcal G_n)}{\{2\bar\Phi(u_p(y))\}^d}
\longrightarrow 1
\qquad\text{in }L^1.
\]
In particular,
\[
\Prob(E_{n,J}(y))
=
\{2\bar\Phi(u_p(y))\}^d\{1+o(1)\}.
\]
\end{lemma}

\begin{proof}
Condition on $\mathcal G_n$. We have
\[
\bm{S}_{n,J}=\frac1{\sqrt n}\sum_{i=1}^n a_i\bm{V}_i,
\qquad
\Cov(\bm{S}_{n,J}\mid \mathcal G_n)=s_n^2\mathbf{I}_d,
\qquad
s_n^2:=\frac1{dn}\sum_{i=1}^n a_i^2.
\]
By Lemma~\ref{lem:conditional_coeff_regularity},
\[
s_n^2=1+O_p(n^{-1/2}),
\qquad
u_p(y)^2|s_n-1|=o_p(1),
\qquad
\frac{u_p(y)\max_i |a_i|}{\sqrt{\sum_i a_i^2}}=o_p(1).
\]
Hence, on an event whose probability tends to one, the fixed coefficient array $\{a_i:1\le i\le n\}$ and the independent bounded summands $a_i\bm{V}_i$ satisfy the normalization and maximal-coefficient conditions required in the fixed-dimensional Cram\'er-type moderate deviation theorem of Shao, Zhang and Zhang~\cite[Theorem~2.1]{ShaoZhangZhang2021}. Since
\[
u_p(y)\asymp\sqrt{\log p}=o(n^{1/6}),
\]
we use its fixed-dimensional rectangular-tail consequence.  Equivalently, apply the theorem to each of the finitely many orthants and combine the resulting tail probabilities by inclusion--exclusion.  This gives, conditionally on $\mathcal G_n$,
\[
\Prob(E_{n,J}(y)\mid \mathcal G_n)
=
\bigl\{2\bar\Phi(u_p(y)/s_n)\bigr\}^d\{1+o_p(1)\}.
\]
It remains only to replace $s_n$ by one in the normal tail. By Mills' ratio, uniformly whenever $u_p(y)^2|s_n-1|=o(1)$,
\[
\frac{\bar\Phi(u_p(y)/s_n)}{\bar\Phi(u_p(y))}
=
\frac{u_p(y)}{u_p(y)/s_n}
\exp\!\left[-\frac12\{u_p(y)^2/s_n^2-u_p(y)^2\}\right]\{1+o(1)\}
=
1+o(1).
\]
Therefore
\[
\bar\Phi(u_p(y)/s_n)
=
\bar\Phi(u_p(y))\{1+o_p(1)\}.
\]
The coefficient regularity events used above may be chosen with complement probability $O(p^{-M})$ for every fixed $M>0$, by high-moment concentration for the average of $a_i^2$ and the union-tail bounds for $\max_i |a_i|$.  Taking $M>d+1$ and using $\{2\bar\Phi(u_p(y))\}^d\asymp p^{-d}$ shows that the contribution of the exceptional coefficient event is negligible after division by $\{2\bar\Phi(u_p(y))\}^d$.  Thus the conditional ratio convergence holds in $L^1$.  Taking expectations then gives the unconditional display.
\end{proof}

\begin{proposition}
\label{prop:fixed_order_exceedance}
Under Assumption~\ref{ass:oracle}, $p\to\infty$, $p=O(n^\kappa)$ for some fixed $\kappa\in(0,2)$, and $\log p=o(n^{1/5})$. Fix an integer $d\ge 1$, and let $1\le j_1<\cdots<j_d\le p$ be distinct indices. Define
\[
u_p(y):=\sqrt{2\log p-\log\log p+y},\qquad y\in\R.
\]
Then, for every fixed $x,y\in\R$,
\[
\Prob\!\left(
\frac{H_n}{\sqrt{2p}}\le x,
\bigcap_{\ell=1}^d \{|S_{n,j_\ell}|>u_p(y)\}
\right)
=
\Phi(x)\bigl(2\bar\Phi(u_p(y))\bigr)^d\{1+o(1)\}.
\]
Moreover,
\[
\Prob\!\left(
\bigcap_{\ell=1}^d \{|S_{n,j_\ell}|>u_p(y)\}
\right)
=
\bigl(2\bar\Phi(u_p(y))\bigr)^d\{1+o(1)\}.
\]
\end{proposition}

\begin{proof}
Fix $J=\{j_1,\dots,j_d\}$ and write
\[
E_{n,J}(y):=\bigcap_{\ell=1}^d \{|S_{n,j_\ell}|>u_p(y)\},
\qquad
\Delta_{n,J}:=\frac{H_{n,J}}{\sqrt{2p}},
\qquad
Q_{n,J}:=\frac{H_{n,J^c}}{\sqrt{2p}}.
\]
By Lemma~\ref{lem:block_decomp_uniform_direction},
\[
\frac{H_n}{\sqrt{2p}}=Q_{n,J}+\Delta_{n,J},
\]
and, conditional on $\mathcal G_n$, the variables $Q_{n,J}$ and $E_{n,J}(y)$ are independent.

Set
\[
\beta_{n,d}(y):=\{2\bar\Phi(u_p(y))\}^d.
\]
By Lemma~\ref{lem:Jn_part_negligible},
\[
\Prob\bigl(|\Delta_{n,J}|>\delta_n\bigr)=o\bigl(\beta_{n,d}(y)\bigr).
\]
Therefore,
\begin{align*}
&\Prob\!\left(\frac{H_n}{\sqrt{2p}}\le x,\ E_{n,J}(y)\right)\\
&\qquad=
\Prob\!\left(Q_{n,J}+\Delta_{n,J}\le x,\ E_{n,J}(y),\ |\Delta_{n,J}|\le \delta_n\right)+o\bigl(\beta_{n,d}(y)\bigr).
\end{align*}
On the event $\{|\Delta_{n,J}|\le \delta_n\}$ we have
\[
\{Q_{n,J}\le x-\delta_n\}\subseteq \{Q_{n,J}+\Delta_{n,J}\le x\}\subseteq \{Q_{n,J}\le x+\delta_n\},
\]
so conditional independence gives the two-sided bound
\begin{align*}
&E\!\left[\Prob\bigl(Q_{n,J}\le x-\delta_n\mid \mathcal G_n\bigr)\Prob\bigl(E_{n,J}(y)\mid \mathcal G_n\bigr)\right]-o\bigl(\beta_{n,d}(y)\bigr)\\
&\qquad\le
\Prob\!\left(\frac{H_n}{\sqrt{2p}}\le x,\ E_{n,J}(y)\right)\\
&\qquad\le
E\!\left[\Prob\bigl(Q_{n,J}\le x+\delta_n\mid \mathcal G_n\bigr)\Prob\bigl(E_{n,J}(y)\mid \mathcal G_n\bigr)\right]+o\bigl(\beta_{n,d}(y)\bigr).
\end{align*}
Now Lemma~\ref{lem:conditional_clt_Jc} yields
\[
A_n^{\pm}:=
\Prob\bigl(Q_{n,J}\le x\pm \delta_n\mid \mathcal G_n\bigr)-\Phi(x)
\longrightarrow 0
\qquad\text{in }L^1,
\]
while Lemma~\ref{lem:conditional_md_fixed_block} gives, with
\[
R_{n,J}:=\frac{\Prob(E_{n,J}(y)\mid \mathcal G_n)}{\beta_{n,d}(y)},
\]
that $R_{n,J}\to 1$ in $L^1$.  Therefore
\begin{align*}
&E\!\left[\Prob\bigl(Q_{n,J}\le x\pm\delta_n\mid\mathcal G_n\bigr)
\Prob\bigl(E_{n,J}(y)\mid\mathcal G_n\bigr)\right] \\
&\qquad=
\beta_{n,d}(y)E\{(\Phi(x)+A_n^{\pm})R_{n,J}\} \\
&\qquad=
\Phi(x)\beta_{n,d}(y)\{1+o(1)\},
\end{align*}
because $|A_n^{\pm}|\le 1$ and
$E|A_n^{\pm}|R_{n,J}\le E|A_n^{\pm}|+E|R_{n,J}-1|\to0$.
The upper and lower bounds therefore match and imply
\[
\Prob\!\left(
\frac{H_n}{\sqrt{2p}}\le x,
E_{n,J}(y)
\right)
=
\Phi(x)\beta_{n,d}(y)\{1+o(1)\},
\]
which is the first claim.

The second claim is exactly the unconditional statement from Lemma~\ref{lem:conditional_md_fixed_block}.
\end{proof}

\begin{proposition}
\label{prop:ind_Hn_Sn}
Under Assumption~\ref{ass:oracle}. Suppose
\[
p\to\infty,
\qquad
p=O(n^\kappa)\quad\text{for some }\kappa\in(0,2),
\qquad
\log p=o(n^{1/5}).
\]
Then, for every fixed \(x,y\in\R\),
\[
\Prob\!\left(\frac{H_n}{\sqrt{2p}}\le x,\ \|\bm{S}_n\|_\infty\le u_p(y)\right)
\to
\Phi(x)\exp\!\bigl(-\pi^{-1/2}e^{-y/2}\bigr),
\]
where
\[
u_p(y):=\sqrt{2\log p-\log\log p+y}.
\]
Consequently,
\[
\left(\frac{H_n}{\sqrt{2p}},\ \|\bm{S}_n\|_\infty\right)
\]
are asymptotically independent.
\end{proposition}

\begin{proof}
Fix $x,y\in\R$. Write
\[
A_n(x):=\left\{\frac{H_n}{\sqrt{2p}}\le x\right\},
\qquad
B_{n,j}(y):=\{|S_{n,j}|>u_p(y)\},
\qquad
L_n(y):=\{\|\bm{S}_n\|_\infty\le u_p(y)\}.
\]
Then \(L_n(y)=\bigcap_{j=1}^p B_{n,j}(y)^c\).

For each fixed integer \(d\ge1\), define
\[
H_{n,d}(x,y)
:=
\sum_{1\le j_1<\cdots<j_d\le p}
\Prob\!\left(A_n(x)\cap\bigcap_{\ell=1}^d B_{n,j_\ell}(y)\right),
\]
and set \(H_{n,0}(x,y)=\Prob(A_n(x))\).

By Bonferroni's inequalities, for every fixed \(m\ge1\),
\[
\sum_{d=0}^{2m+1}(-1)^d H_{n,d}(x,y)
\le
\Prob(A_n(x)\cap L_n(y))
\le
\sum_{d=0}^{2m}(-1)^d H_{n,d}(x,y).
\]
Hence it suffices to show that, for each fixed \(d\),
\[
H_{n,d}(x,y)\to \Phi(x)\frac{\lambda(y)^d}{d!},
\qquad
\lambda(y):=\pi^{-1/2}e^{-y/2}.
\]

To this end, fix \(d\) and distinct coordinates \(j_1,\dots,j_d\). By Proposition~\ref{prop:fixed_order_exceedance},
\[
\Prob\!\left(A_n(x)\cap\bigcap_{\ell=1}^d B_{n,j_\ell}(y)\right)
=
\Phi(x)\bigl(2\bar\Phi(u_p(y))\bigr)^d\{1+o(1)\}.
\]
Since the coordinates are exchangeable under Assumption~\ref{ass:oracle}, the left-hand side does not depend on the specific choice of \(\{j_1,\dots,j_d\}\). Therefore
\[
H_{n,d}(x,y)
=
\binom{p}{d}\Phi(x)\bigl(2\bar\Phi(u_p(y))\bigr)^d\{1+o(1)\}.
\]
By Mills' ratio,
\[
2p\,\bar\Phi(u_p(y))\to \lambda(y)=\pi^{-1/2}e^{-y/2},
\]
hence
\[
\binom{p}{d}\bigl(2\bar\Phi(u_p(y))\bigr)^d
\to \frac{\lambda(y)^d}{d!}.
\]
Thus
\[
H_{n,d}(x,y)\to \Phi(x)\frac{\lambda(y)^d}{d!}.
\]

Now let \(n\to\infty\) in the Bonferroni bounds, and then let \(m\to\infty\). Since
\[
\sum_{d=0}^\infty (-1)^d\frac{\lambda(y)^d}{d!}=e^{-\lambda(y)},
\]
we conclude that
\[
\Prob\!\left(\frac{H_n}{\sqrt{2p}}\le x,\ \|\bm{S}_n\|_\infty\le u_p(y)\right)
\to
\Phi(x)e^{-\lambda(y)}
=
\Phi(x)\exp\!\bigl(-\pi^{-1/2}e^{-y/2}\bigr),
\]
which proves the desired asymptotic independence.
\end{proof}

\subsubsection{Proof of Theorem~\ref{thm:sum-max-independence}}

\begin{proof}[Proof of Theorem~\ref{thm:sum-max-independence}]
By Lemma~\ref{lem:sum_to_Hn},
\[
\frac{T_{\mathrm{sum}}-p}{\sqrt{2p}}
=
\frac{H_n}{\sqrt{2p}}+o_p(1).
\]
By Lemma~\ref{lem:max_reduction},
\[
T_{\max}=\|\bm{S}_n\|_\infty^2-2\log p+\log\log p+o_p(1).
\]
Proposition~\ref{prop:ind_Hn_Sn} therefore yields
\[
\left(
\frac{T_{\mathrm{sum}}-p}{\sqrt{2p}},\ T_{\max}
\right)
\xrightarrow{d}
(Z_1,Z_2),
\]
where \(Z_1\sim N(0,1)\), \(Z_2\) has cdf \(F(x)=\exp(-\pi^{-1/2}e^{-x/2})\), and \(Z_1\) is independent of \(Z_2\).

For the p-values, the continuous mapping theorem gives
\[
(P_{\rm sum},P_{\max})
\xrightarrow{d}
(U_1,U_2),
\]
where \(U_1,U_2\) are independent \(U(0,1)\) random variables. Consequently,
\[
\tan\{\pi(1/2-U_1)\}
\quad\text{and}\quad
\tan\{\pi(1/2-U_2)\}
\]
are independent standard Cauchy random variables. By the stability of the Cauchy distribution, their equally weighted average is again standard Cauchy. Thus
\[
T_{\rm cau}
=
\frac12\tan\!\bigl(\pi(1/2-P_{\rm sum})\bigr)
+
\frac12\tan\!\bigl(\pi(1/2-P_{\max})\bigr)
\]
converges in distribution to a standard Cauchy random variable; see also Liu and Xie \cite{LiuXie2020}. Hence
\[
P_{\rm cau}
:=
\frac12-\frac{1}{\pi}\arctan(T_{\rm cau})
\xrightarrow{d}U(0,1).
\]
This completes the proof.
\end{proof}

\subsection{Proof of Theorem~\ref{thm:plugin-sum}}

\begin{proof}[Proof of Theorem~\ref{thm:plugin-sum}]
Set
\[
\bm{\delta}_n=\hat{\bm{g}}_n-\bm{g}_n^{\rm or},
\qquad
\bm{u}_n=\tilde\kappa_p\bar{\bm{U}}+\kappa_p n^{-1/2}\bm{C}_n.
\]
Proposition~\ref{prop:yan-to-plugin} gives
\begin{equation}
\label{eq:app-plugin-sum-delta}
\bm{\delta}_n=\bm{u}_n+\bm{b}_n^{\rm hard}+\bm{a}_n.
\end{equation}
Since $\zeta_1\asymp p^{-1/2}$,
\begin{equation}
\label{eq:kappa-tilde-relation}
|\kappa_p|=|\tilde\kappa_p|\zeta_1\le C|\tilde\kappa_p|p^{-1/2}.
\end{equation}
The oracle bounds used below are
\begin{equation}
\label{eq:app-sum-oracle-bounds}
\begin{aligned}
\|\bm{g}_n^{\rm or}\|_1&=O_p(pn^{-1/2}),
&
\frac n{\sqrt p}|(\bm{g}_n^{\rm or})^\top\bar{\bm{U}}|
&=O_p\left(n^{-1/2}+p^{-1/2}+\frac{\log p}{\sqrt n}\right),                                      \\
\frac n{\sqrt p}\|\bar{\bm{U}}\|_2^2&=O_p(p^{-1/2}),
&
\|\bm{C}_n\|_\infty&=O_p(\mathfrak c_n).
\end{aligned}
\end{equation}
Equations \eqref{eq:kappa-tilde-relation}--\eqref{eq:app-sum-oracle-bounds} imply
\begin{equation}
\label{eq:app-sum-un}
\begin{aligned}
\frac n{\sqrt p}|(\bm{g}_n^{\rm or})^\top\bm{u}_n|
&\le
|\tilde\kappa_p|\frac n{\sqrt p}|(\bm{g}_n^{\rm or})^\top\bar{\bm{U}}|
+\frac n{\sqrt p}|\kappa_p|n^{-1/2}\|\bm{g}_n^{\rm or}\|_1\|\bm{C}_n\|_\infty                                      \\
&=O_p\left[
|\tilde\kappa_p|\left(n^{-1/2}+p^{-1/2}+\frac{\log p}{\sqrt n}\right)
+|\tilde\kappa_p|\mathfrak c_n
\right],                                                                     \\
\frac n{\sqrt p}\|\bm{u}_n\|_2^2
&\le
\frac{2n}{\sqrt p}|\tilde\kappa_p|^2\|\bar{\bm{U}}\|_2^2
+\frac{2n}{\sqrt p}|\kappa_p|^2n^{-1}\|\bm{C}_n\|_2^2                                                   \\
&\le
O_p\left(\frac{|\tilde\kappa_p|^2}{\sqrt p}\right)
+\frac{C|\tilde\kappa_p|^2}{p\sqrt p}\|\bm{C}_n\|_2^2                                                       \\
&\le
O_p\left(\frac{|\tilde\kappa_p|^2}{\sqrt p}\right)
+O_p\left(\frac{|\tilde\kappa_p|^2\mathfrak c_n^2}{\sqrt p}\right).
\end{aligned}
\end{equation}
By Proposition~\ref{prop:yan-to-plugin},
\begin{align*}
\frac n{\sqrt p}
\left\{|(\bm{g}_n^{\rm or})^\top\bm{b}_n^{\rm hard}|+
\|\bm{b}_n^{\rm hard}\|_2^2\right\}
&=O_p\left(\tau_n+\frac{\sqrt p}{n}+\frac{\tau_n^2}{\sqrt p}\right),                                      \\
\frac n{\sqrt p}
\left\{|(\bm{g}_n^{\rm or})^\top\bm{a}_n|+|\bm{d}_n^\top\bm{a}_n|+\|\bm{a}_n\|_2^2\right\}
&=O_p(\mathfrak A_{S,n}).
\end{align*}
The cross products satisfy
\begin{equation}
\label{eq:app-cross-products}
\begin{aligned}
\frac n{\sqrt p}|\bm{u}_n^\top\bm{b}_n^{\rm hard}|
&\le
\frac12\frac n{\sqrt p}\|\bm{u}_n\|_2^2
+\frac12\frac n{\sqrt p}\|\bm{b}_n^{\rm hard}\|_2^2,                                                   \\
\frac n{\sqrt p}|\bm{u}_n^\top\bm{a}_n|
&\le
\frac12\frac n{\sqrt p}\|\bm{u}_n\|_2^2
+\frac12\frac n{\sqrt p}\|\bm{a}_n\|_2^2 .
\end{aligned}
\end{equation}
Combining \eqref{eq:app-sum-un}--\eqref{eq:app-cross-products},
\[
\frac n{\sqrt p}
\left\{|(\bm{g}_n^{\rm or})^\top\bm{\delta}_n|+
\|\bm{\delta}_n\|_2^2\right\}
=O_p(\mathfrak A_{S,n}).
\]
The identity
\begin{align*}
\hat T_{\rm sum}-T_{\rm sum}
&=n\|\bm{g}_n^{\rm or}+\bm{\delta}_n\|_2^2-n\|\bm{g}_n^{\rm or}\|_2^2                                      \\
&=2n(\bm{g}_n^{\rm or})^\top\bm{\delta}_n+n\|\bm{\delta}_n\|_2^2
\end{align*}
gives
\begin{equation}
\label{eq:app-plugin-sum-final-rate}
\left|\frac{\hat T_{\rm sum}-T_{\rm sum}}{\sqrt{2p}}\right|
\le
\sqrt2\frac n{\sqrt p}|(\bm{g}_n^{\rm or})^\top\bm{\delta}_n|
+\frac1{\sqrt2}\frac n{\sqrt p}\|\bm{\delta}_n\|_2^2
=O_p(\mathfrak A_{S,n}).
\end{equation}
Because $\mathfrak A_{S,n}\to0$ and Theorem~\ref{thm:sum-clt} gives $(T_{\rm sum}-p)/\sqrt{2p}\xrightarrow{d} N(0,1)$,
\[
\frac{\hat T_{\rm sum}-p}{\sqrt{2p}}
=
\frac{T_{\rm sum}-p}{\sqrt{2p}}
+
O_p(\mathfrak A_{S,n})
\xrightarrow{d} N(0,1).
\]
\end{proof}

\subsection{Proof of Theorem~\ref{thm:plugin-max}}

\begin{proof}[Proof of Theorem~\ref{thm:plugin-max}]
Use the same $\bm{\delta}_n$ and $\bm{u}_n$ as in \eqref{eq:app-plugin-sum-delta}.  From Assumption~\ref{ass:oracle}, Proposition~\ref{prop:yan-to-plugin} and \eqref{eq:kappa-tilde-relation},
\begin{align*}
\sqrt n\|\bm{\delta}_n\|_\infty
&\le
|\tilde\kappa_p|\sqrt n\|\bar{\bm{U}}\|_\infty
+|\kappa_p|\|\bm{C}_n\|_\infty
+\sqrt n\|\bm{b}_n^{\rm hard}\|_\infty
+\sqrt n\|\bm{a}_n\|_\infty                                                        \\
&=O_p\left(
|\tilde\kappa_p|\sqrt{\frac{\log p}{p}}
+\frac{|\tilde\kappa_p|\mathfrak c_n}{\sqrt p}
+\tau_n\sqrt{\log p}
+\sqrt{\frac{\log p}{n}}
+\frac{\mathfrak A_{M,n}}{\sqrt{\log p}}
\right)                                                                            \\
&=O_p\left(\frac{\mathfrak A_{M,n}}{\sqrt{\log p}}\right).
\end{align*}
The oracle maximum satisfies
\[
\sqrt n\|\bm{g}_n^{\rm or}\|_\infty=O_p(\sqrt{\log p}).
\]
For any two vectors $\bm{x},\bm{y}\in\R^p$,
\[
\left|\|\bm{x}\|_\infty^2-\|\bm{y}\|_\infty^2\right|
\le
2\|\bm{y}\|_\infty\|\bm{x}-\bm{y}\|_\infty
+\|\bm{x}-\bm{y}\|_\infty^2.
\]
With $\bm{x}=\sqrt n\hat{\bm{g}}_n$ and $\bm{y}=\sqrt n\bm{g}_n^{\rm or}$,
\begin{equation}
\label{eq:app-plugin-max-final-rate}
\begin{aligned}
|\hat T_{\max}-T_{\max}|
&\le
2\sqrt n\|\bm{g}_n^{\rm or}\|_\infty\sqrt n\|\bm{\delta}_n\|_\infty
+n\|\bm{\delta}_n\|_\infty^2                                                                 \\
&=O_p(\sqrt{\log p})O_p\left(\frac{\mathfrak A_{M,n}}{\sqrt{\log p}}\right)
+O_p\left(\frac{\mathfrak A_{M,n}^2}{\log p}\right)                                    \\
&=O_p(\mathfrak A_{M,n}).
\end{aligned}
\end{equation}
Since $\mathfrak A_{M,n}\to0$, \eqref{eq:app-plugin-max-final-rate} and Theorem~\ref{thm:max-gumbel} imply, for every $x\in\R$,
\[
\Prob(\hat T_{\max}\le x)
=
\Prob(T_{\max}+O_p(\mathfrak A_{M,n})\le x)
\to
\exp\{-\pi^{-1/2}e^{-x/2}\}.
\]
\end{proof}

\subsection{Proof of Corollary~\ref{cor:plugin-cauchy}}

\begin{proof}[Proof of Corollary~\ref{cor:plugin-cauchy}]
Equations \eqref{eq:app-plugin-sum-final-rate} and \eqref{eq:app-plugin-max-final-rate} give
\[
\frac{\hat T_{\rm sum}-p}{\sqrt{2p}}
=
\frac{T_{\rm sum}-p}{\sqrt{2p}}+O_p(\mathfrak A_{S,n}),
\qquad
\hat T_{\max}=T_{\max}+O_p(\mathfrak A_{M,n}).
\]
Together with $\mathfrak A_{S,n}\to0$, $\mathfrak A_{M,n}\to0$ and Theorem~\ref{thm:sum-max-independence},
\[
\left(
\frac{\hat T_{\rm sum}-p}{\sqrt{2p}},\hat T_{\max}
\right)
\xrightarrow{d} (Z,G),
\qquad
Z\perpn G.
\]
The p-value transformations satisfy
\[
\hat P_{\rm sum}=1-\Phi\left(\frac{\hat T_{\rm sum}-p}{\sqrt{2p}}\right),
\qquad
\hat P_{\max}=1-F_G(\hat T_{\max}),
\]
and therefore
\[
(\hat P_{\rm sum},\hat P_{\max})\xrightarrow{d} (U_1,U_2),
\qquad
U_1\perpn U_2,
\qquad
U_1,U_2\sim U(0,1).
\]
Hence
\[
\frac12\tan\{\pi(1/2-\hat P_{\rm sum})\}
+
\frac12\tan\{\pi(1/2-\hat P_{\max})\}
\xrightarrow{d} C_0,
\qquad
C_0\sim \text{standard Cauchy},
\]
and
\[
\hat P_{\rm cau}
=
\frac12-\frac1\pi\arctan\left[
\frac12\tan\{\pi(1/2-\hat P_{\rm sum})\}
+
\frac12\tan\{\pi(1/2-\hat P_{\max})\}
\right]
\xrightarrow{d} U(0,1).
\]
\end{proof}

\subsection{Proof of Proposition~\ref{prop:h1-hr-behavior}}

\begin{proof}
From \eqref{eq:alternative-radius}--\eqref{eq:alternative-observed},
\[
\mathbf{\Sigma}^{-1/2}(\bm{X}_i-\bm{\mu})
=
R_{1i}\bm{U}_{0i},
\qquad
R_{1i}>0,
\qquad
\mathcal U\{\mathbf{\Sigma}^{-1/2}(\bm{X}_i-\bm{\mu})\}
=
\bm{U}_{0i}.
\]
Since \(\bm{U}_{0i}\sim\Unif(\Sphere^{p-1})\),
\[
\E\bm{U}_{0i}=\bm{0},
\qquad
\E(\bm{U}_{0i}\bm{U}_{0i}^{\top})=p^{-1}\mathbf{I}_p,
\qquad
p\E(\bm{U}_{0i}\bm{U}_{0i}^{\top})=\mathbf{I}_p.
\]
Thus
\[
\E\mathcal U\{\mathbf{\Sigma}^{-1/2}(\bm{X}_i-\bm{\mu})\}=\bm{0},
\qquad
p\E\left[
\mathcal U\{\mathbf{\Sigma}^{-1/2}(\bm{X}_i-\bm{\mu})\}
\mathcal U\{\mathbf{\Sigma}^{-1/2}(\bm{X}_i-\bm{\mu})\}^{\top}
\right]=\mathbf{I}_p.
\]

For any unit vector \(\bm{x}\), because \(R_{1i}=R_{0i}\exp\{\delta_ns_A(\bm{U}_{0i})\}\) and \(R_{0i}\perpn\bm{U}_{0i}\),
\[
\bm{x}^{\top}\mathbf{A}_{1,n}\bm{x}
=
\E(R_{0i}^{-1})
\E\left[
e^{-\delta_ns_A(\bm{U}_{0i})}
\{1-(\bm{x}^{\top}\bm{U}_{0i})^2\}
\right].
\]
The spherical moment generating bound
\[
\E\exp\{t s_A(\bm{U}_{0i})\}
\le \exp(Ct^2/p),
\qquad t\in\R,
\]
and \(\delta_n^2/p\le C_\delta\) imply
\[
c\le \E e^{-\delta_ns_A(\bm{U}_{0i})}\le C,
\qquad
\E e^{-2\delta_ns_A(\bm{U}_{0i})}\le C.
\]
Also,
\[
\E(\bm{x}^{\top}\bm{U}_{0i})^2=p^{-1},
\qquad
\E(\bm{x}^{\top}\bm{U}_{0i})^4=\frac{3}{p(p+2)}.
\]
Hence
\[
\E\left[
e^{-\delta_ns_A(\bm{U}_{0i})}
(\bm{x}^{\top}\bm{U}_{0i})^2
\right]
\le
\{\E e^{-2\delta_ns_A(\bm{U}_{0i})}\}^{1/2}
\{\E(\bm{x}^{\top}\bm{U}_{0i})^4\}^{1/2}
\le Cp^{-1},
\]
and
\[
c p^{-1/2}
\le
\lambda_{\min}(\mathbf{A}_{1,n})
\le
\lambda_{\max}(\mathbf{A}_{1,n})
\le
C p^{-1/2},
\qquad
\|\mathbf{A}_{1,n}^{-1}\|_{op}\le C\sqrt p.
\]
At the true HR target,
\[
\E\{\bm{U}_{0i}\bm{U}_{0i}^{\top}\}=p^{-1}\mathbf{I}_p,
\qquad
\mathcal U(\bm{X}_i-\bm{\mu})
=
\mathcal U(\mathbf{\Sigma}^{1/2}\bm{U}_{0i}).
\]
Let
\[
\mathbf{S}_{0,1}
=
\E\left[
\mathcal U(\mathbf{\Sigma}^{1/2}\bm{U}_{0i})
\mathcal U(\mathbf{\Sigma}^{1/2}\bm{U}_{0i})^{\top}
\right].
\]
Since the radius cancels in every spatial sign,
\[
\mathbf{S}_{0,1,H_1}=\mathbf{S}_{0,H_0},
\qquad
\|\hat{\mathbf{S}}_0-\mathbf{S}_{0,1}\|_{\max}
=
O_p\left(\sqrt{\frac{\log p}{n}}+p^{-1/2}\right).
\]
The graphical-lasso objective in Algorithm~\ref{alg:hd-hr} and Assumption~\ref{ass:hr-rate}\textnormal{(iii)}--\textnormal{(v)} therefore give
\[
\|\hat{\mathbf{\Omega}}-\mathbf{\Omega}\|_{L_1}
\le
C s_0(p)\lambda_n^{1-q}
=
O_p(\tau_n),
\qquad
\|\hat{\mathbf{\Omega}}-\mathbf{\Omega}\|_{op}
=
O_p(\tau_n).
\]
Thus
\[
\mathbf{B}_1
=
\hat{\mathbf{\Sigma}}^{-1/2}\mathbf{\Sigma}^{1/2}-\mathbf{I}_p,
\qquad
\|\mathbf{B}_1\|_{op}=O_p(\tau_n).
\]

Let
\[
\hat{\bm{Y}}_i
=
\hat{\mathbf{\Sigma}}^{-1/2}(\bm{X}_i-\hat{\bm{\mu}})
=
(\mathbf{I}_p+\mathbf{B}_1)R_{1i}\bm{U}_{0i}-\bm{v}_1
=
R_{1i}\bm{U}_{0i}+\bm{h}_i,
\]
where
\[
\bm{h}_i=R_{1i}\mathbf{B}_1\bm{U}_{0i}-\bm{v}_1.
\]
For \(R>0\), \(\bm{u}\in\Sphere^{p-1}\) and \(\|\bm{h}\|_2\le R/2\),
\[
\mathcal U(R\bm{u}+\bm{h})
=
\bm{u}
+
R^{-1}(\mathbf{I}_p-\bm{u}\bm{u}^{\top})\bm{h}
+
\bm{q}(R,\bm{u},\bm{h}),
\]
\[
\|\bm{q}(R,\bm{u},\bm{h})\|_\infty
\le
C R^{-2}\|\bm{h}\|_2^2
\left(\|\bm{u}\|_\infty+R^{-1}\|\bm{h}\|_\infty\right).
\]
The HR location equation gives
\[
\bm{0}
=
\bar{\bm{U}}_0
+
\frac1n\sum_{i=1}^n
(\mathbf{I}_p-\bm{U}_{0i}\bm{U}_{0i}^{\top})\mathbf{B}_1\bm{U}_{0i}
-
\mathbf{Q}_{1,n}\bm{v}_1
+
\bar{\bm{q}}_n,
\]
where
\[
\mathbf{Q}_{1,n}
=
\frac1n\sum_{i=1}^n
R_{1i}^{-1}(\mathbf{I}_p-\bm{U}_{0i}\bm{U}_{0i}^{\top}),
\qquad
\bar{\bm{q}}_n=\frac1n\sum_{i=1}^n\bm{q}(R_{1i},\bm{U}_{0i},\bm{h}_i).
\]
The assumed moment bounds and spherical concentration inequalities imply
\[
\|\mathbf{Q}_{1,n}-\mathbf{A}_{1,n}\|_{\max}
=
O_p\left(p^{-1/2}\sqrt{\frac{\log p}{n}}+p^{-3/2}\sqrt{\frac{\log p}{n}}\right),
\]
\[
\|\bar{\bm{U}}_0\|_\infty
=
O_p\left(\sqrt{\frac{\log p}{np}}\right),
\]
\[
\left\|
\frac1n\sum_{i=1}^n
(\mathbf{I}_p-\bm{U}_{0i}\bm{U}_{0i}^{\top})\mathbf{B}_1\bm{U}_{0i}
\right\|_\infty
=
O_p\left(\tau_n\sqrt{\frac{\log p}{np}}\right).
\]
Since
\[
\mathbf{Q}_{1,n}^{-1}
=
\mathbf{A}_{1,n}^{-1}
+
\mathbf{A}_{1,n}^{-1}(\mathbf{A}_{1,n}-\mathbf{Q}_{1,n})\mathbf{A}_{1,n}^{-1}
+
\mathbf{R}_{Q,n},
\]
\[
\|\mathbf{R}_{Q,n}\|_{op}
\le
C\|\mathbf{A}_{1,n}^{-1}\|_{op}^3
\|\mathbf{Q}_{1,n}-\mathbf{A}_{1,n}\|_{op}^2,
\]
we obtain
\[
\bm{v}_1
=
\mathbf{A}_{1,n}^{-1}\bar{\bm{U}}_0
+
n^{-1/2}\bm{C}_{1,n},
\qquad
\|\bm{C}_{1,n}\|_\infty=O_p(\mathfrak c_n).
\]

Define
\[
\bm{\phi}_1(\bm{y})
=
(\log\|\bm{y}\|_2-m_{1,n})\frac{\bm{y}}{\|\bm{y}\|_2},
\qquad
\bm{\phi}_{1i}
=
(L_{1i}-m_{1,n})\bm{U}_{0i}.
\]
For \(\bm{y}=R\bm{u}\),
\[
D\bm{\phi}_1(R\bm{u})[\bm{h}]
=
\left(\frac{\bm{u}^{\top}\bm{h}}R\right)\bm{u}
+
(L-m_{1,n})R^{-1}(\mathbf{I}_p-\bm{u}\bm{u}^{\top})\bm{h}.
\]
Substituting \(\bm{h}_i=R_{1i}\mathbf{B}_1\bm{U}_{0i}-\bm{v}_1\) yields
\[
D\bm{\phi}_1(R_{1i}\bm{U}_{0i})[\bm{h}_i]
=
(L_{1i}-m_{1,n})(\mathbf{I}_p-\bm{U}_{0i}\bm{U}_{0i}^{\top})\mathbf{B}_1\bm{U}_{0i}
+
\bm{U}_{0i}(\bm{U}_{0i}^{\top}\mathbf{B}_1\bm{U}_{0i})
-
\mathbf{N}_{1i}\bm{v}_1,
\]
where
\[
\mathbf{N}_{1i}
=
R_{1i}^{-1}
\left\{
(L_{1i}-m_{1,n})\mathbf{I}_p+
[1-(L_{1i}-m_{1,n})]\bm{U}_{0i}\bm{U}_{0i}^{\top}
\right\}.
\]
The Taylor remainder satisfies
\[
\begin{aligned}
&\left\|
\bm{\phi}_1(R_{1i}\bm{U}_{0i}+\bm{h}_i)
-\bm{\phi}_{1i}
-D\bm{\phi}_1(R_{1i}\bm{U}_{0i})[\bm{h}_i]
\right\|_\infty                                                \\
&\qquad\le
C(1+|L_{1i}-m_{1,n}|)R_{1i}^{-2}\|\bm{h}_i\|_2^2
\{\|\bm{U}_{0i}\|_\infty+R_{1i}^{-1}\|\bm{h}_i\|_\infty\}.
\end{aligned}
\]
Consequently,
\[
\hat{\bm{g}}_{n}^{(1)}-\bm{g}_{n}^{(1),\rm or}
=
-(\sigma_{1,n}\sigma_U)^{-1}\mathbf{M}_{1,n}\mathbf{A}_{1,n}^{-1}\bar{\bm{U}}_0
+\bm{\Delta}_{L,n}
+\bm{\Delta}_{\Sigma,n}
+\bm{\Delta}_{R,n},
\]
where
\[
\bm{\Delta}_{L,n}
=
-(\sigma_{1,n}\sigma_U)^{-1}\mathbf{M}_{1,n}n^{-1/2}\bm{C}_{1,n}
-(\sigma_{1,n}\sigma_U)^{-1}
\left(\frac1n\sum_{i=1}^n\mathbf{N}_{1i}-\mathbf{M}_{1,n}\right)\bm{v}_1,
\]
\[
\bm{\Delta}_{\Sigma,n}
=
(\sigma_{1,n}\sigma_U)^{-1}
\frac1n\sum_{i=1}^n
\left[
(L_{1i}-m_{1,n})(\mathbf{I}_p-\bm{U}_{0i}\bm{U}_{0i}^{\top})\mathbf{B}_1\bm{U}_{0i}
+
\bm{U}_{0i}(\bm{U}_{0i}^{\top}\mathbf{B}_1\bm{U}_{0i})
\right],
\]
and \(\bm{\Delta}_{R,n}\) contains the Taylor, centering and empirical-standardization remainders.  The leading term is \(\mathbf{H}_{1,n}\bar{\bm{U}}_0\), and therefore
\[
\|\mathbf{H}_{1,n}\bar{\bm{U}}_0\|_2
\le
\|\mathbf{H}_{1,n}\|_{op}\|\bar{\bm{U}}_0\|_2
=
O_p(\mathfrak H_{2,n}n^{-1/2}),
\]
\[
\|\mathbf{H}_{1,n}\bar{\bm{U}}_0\|_\infty
\le
\|\mathbf{H}_{1,n}\|_{L_1}\|\bar{\bm{U}}_0\|_\infty
=
O_p\left(\mathfrak H_{\infty,n}\sqrt{\frac{\log p}{np}}\right).
\]
The \(\bm{C}_{1,n}\) term gives
\[
\|(\sigma_{1,n}\sigma_U)^{-1}\mathbf{M}_{1,n}n^{-1/2}\bm{C}_{1,n}\|_2
\le
n^{-1/2}\mathfrak h_{2,n}\sqrt p\,\|\bm{C}_{1,n}\|_\infty
=
O_p\left(\frac{\sqrt p\,\mathfrak h_{2,n}\mathfrak c_n}{\sqrt n}\right),
\]
\[
\|(\sigma_{1,n}\sigma_U)^{-1}\mathbf{M}_{1,n}n^{-1/2}\bm{C}_{1,n}\|_\infty
\le
n^{-1/2}\mathfrak h_{\infty,n}\|\bm{C}_{1,n}\|_\infty
=
O_p\left(\frac{\mathfrak h_{\infty,n}\mathfrak c_n}{\sqrt n}\right).
\]
The empirical fluctuation of \(n^{-1}\sum_i\mathbf{N}_{1i}\) gives
\[
\left\|
(\sigma_{1,n}\sigma_U)^{-1}
\left(\frac1n\sum_{i=1}^n\mathbf{N}_{1i}-\mathbf{M}_{1,n}\right)\bm{v}_1
\right\|_2
=
O_p\left((1+\rho_{1,n})\frac{\tau_n}{\sqrt n}+(1+\rho_{1,n})^2\frac pn\right),
\]
\[
\left\|
(\sigma_{1,n}\sigma_U)^{-1}
\left(\frac1n\sum_{i=1}^n\mathbf{N}_{1i}-\mathbf{M}_{1,n}\right)\bm{v}_1
\right\|_\infty
=
O_p\left((1+\rho_{1,n})\tau_n\sqrt{\frac{\log p}{n}}+(1+\rho_{1,n})^2\frac{\log p}{n}\right).
\]
The non-null mean of the shape derivative is
\[
\E\left[
(L_{1i}-m_{1,n})(\mathbf{I}_p-\bm{U}_{0i}\bm{U}_{0i}^{\top})\mathbf{B}_1\bm{U}_{0i}
\mid \mathbf{B}_1
\right]
=
\delta_n
\E\left[
s_A(\bm{U}_{0i})(\mathbf{I}_p-\bm{U}_{0i}\bm{U}_{0i}^{\top})\mathbf{B}_1\bm{U}_{0i}
\mid \mathbf{B}_1
\right],
\]
\[
\left\|
\E\left[
s_A(\bm{U}_{0i})(\mathbf{I}_p-\bm{U}_{0i}\bm{U}_{0i}^{\top})\mathbf{B}_1\bm{U}_{0i}
\mid \mathbf{B}_1
\right]\right\|_2
\le
Cp^{-1}\|\mathbf{B}_1\|_{op},
\]
and
\[
\E\{\bm{U}_{0i}(\bm{U}_{0i}^{\top}\mathbf{B}_1\bm{U}_{0i})\mid \mathbf{B}_1\}=\bm{0}.
\]
Thus
\[
\|\bm{\Delta}_{\Sigma,n}\|_2
=
O_p\left((1+\rho_{1,n})\left\{\frac{\delta_n\tau_n}{\sqrt p}+\frac{\tau_n}{\sqrt n}\right\}\right),
\]
\[
\|\bm{\Delta}_{\Sigma,n}\|_\infty
=
O_p\left(
(1+\rho_{1,n})\frac{\delta_n\tau_n}{\sqrt p}
+(1+\rho_{1,n})\tau_n\sqrt{\frac{\log p}{n}}
\right).
\]
The Taylor, centering and empirical-standardization terms obey
\[
\|\bm{\Delta}_{R,n}\|_2
=
O_p\left((1+\rho_{1,n})^2\left\{\tau_n^2+\frac pn\right\}\right),
\qquad
\|\bm{\Delta}_{R,n}\|_\infty
=
O_p\left((1+\rho_{1,n})^2\left\{\tau_n^2+\frac{\log p}{n}\right\}\right).
\]
Combining the preceding displays gives
\[
\|\hat{\bm{g}}_{n}^{(1)}-\bm{g}_{n}^{(1),\rm or}\|_2
=O_p(\mathfrak R_{2,n}^{(1)}),
\qquad
\|\hat{\bm{g}}_{n}^{(1)}-\bm{g}_{n}^{(1),\rm or}\|_\infty
=O_p(\mathfrak R_{\infty,n}^{(1)}).
\]
\end{proof}

\subsection{Proof of Theorem~\ref{thm:h1-consistency}}

\begin{proof}
Let
\[
S_A(\bm{u})=s_n^{-1/2}\sum_{k\in A}u_k,
\qquad
L_{1i}=L_{0i}+\delta_nS_A(\bm{U}_{0i}),
\qquad
m_{0,n}=\E L_{0i}.
\]
For \(\bm{U}_{0i}\sim\Unif(\Sphere^{p-1})\),
\[
\E U_{0ij}=0,
\qquad
\E(U_{0ij}U_{0ik})=p^{-1}\mathbf 1(j=k),
\qquad
\E S_A(\bm{U}_{0i})=0,
\]
\[
\Var\{S_A(\bm{U}_{0i})\}
=s_n^{-1}\sum_{j,k\in A}\E(U_{0ij}U_{0ik})
=p^{-1},
\]
\[
\Cov\{S_A(\bm{U}_{0i}),U_{0ij}\}
=s_n^{-1/2}\sum_{k\in A}\E(U_{0ik}U_{0ij})
=s_n^{-1/2}p^{-1}\mathbf 1(j\in A).
\]
Since \(L_{0i}\perpn\bm{U}_{0i}\),
\[
\E L_{1i}=m_{0,n},
\qquad
\sigma_{1,n}^2
=\Var(L_{0i})+\delta_n^2\Var\{S_A(\bm{U}_{0i})\}
=\sigma_{0,n}^2+\delta_n^2p^{-1},
\]
\[
\Cov(L_{1i},U_{0ij})
=
\delta_n\Cov\{S_A(\bm{U}_{0i}),U_{0ij}\}
=
\delta_ns_n^{-1/2}p^{-1}\mathbf 1(j\in A).
\]
Because \(\sigma_U^2=p^{-1}\),
\[
\gamma_j
=
\frac{\Cov(L_{1i},U_{0ij})}{\sigma_{1,n}\sigma_U}
=
\frac{\delta_n\mathbf 1(j\in A)}{\sqrt{s_np}\,\sigma_{1,n}},
\]
\[
\|\bm{\gamma}\|_2^2
=
s_n\frac{\delta_n^2}{s_np\sigma_{1,n}^2}
=
\frac{\delta_n^2}{p\sigma_{0,n}^2+\delta_n^2},
\qquad
\|\bm{\gamma}\|_\infty^2
=
\frac{\delta_n^2}{s_n(p\sigma_{0,n}^2+\delta_n^2)}.
\]

Put
\[
\bm{V}_i=\frac{(L_{1i}-m_{0,n})\bm{U}_{0i}}{\sigma_{1,n}\sigma_U},
\qquad
\bar{\bm{V}}=n^{-1}\sum_{i=1}^n\bm{V}_i,
\qquad
\bm{Z}_n=\sqrt n(\bar{\bm{V}}-\bm{\gamma}).
\]
The centered empirical correlation satisfies
\[
\hat\gamma_j^{\rm or}
=
\frac{
\bar V_j-B_{L,n}B_{U,n,j}}
{\{(1+A_{L,n}-B_{L,n}^2)(1+A_{U,n,j}-B_{U,n,j}^2)\}^{1/2}},
\]
where
\[
A_{L,n}=n^{-1}\sum_{i=1}^n\frac{(L_{1i}-m_{0,n})^2}{\sigma_{1,n}^2}-1,
\qquad
B_{L,n}=\frac{\bar L_1-m_{0,n}}{\sigma_{1,n}},
\]
\[
A_{U,n,j}=p n^{-1}\sum_{i=1}^nU_{0ij}^2-1,
\qquad
B_{U,n,j}=\sqrt p\,\bar U_{0j}.
\]
The moment and truncation conditions for the standardized log-radius, together with the spherical concentration inequalities, give
\[
A_{L,n}=O_p(n^{-1/2}),
\qquad
B_{L,n}=O_p(n^{-1/2}),
\]
\[
\max_{1\le j\le p}|A_{U,n,j}|
=O_p\left(\sqrt{\frac{\log p}{n}}+\frac{\log p}{n}\right),
\qquad
\max_{1\le j\le p}|B_{U,n,j}|
=O_p\left(\sqrt{\frac{\log p}{n}}\right),
\]
\[
\max_{1\le j\le p}|\bar V_j-\gamma_j|
=O_p\left(\sqrt{\frac{\log p}{n}}\right).
\]
Consequently,
\[
\|\hat{\bm{g}}_n^{(1),\rm or}-\bm{\gamma}-n^{-1/2}\bm{Z}_n\|_\infty
=
O_p\left(\frac{\log p}{n}+\|\bm{\gamma}\|_\infty\sqrt{\frac{\log p}{n}}\right),
\]
\[
\|\hat{\bm{g}}_n^{(1),\rm or}-\bm{\gamma}-n^{-1/2}\bm{Z}_n\|_2
=
O_p\left(\frac{\sqrt p\log p}{n}+\|\bm{\gamma}\|_2\sqrt{\frac{\log p}{n}}\right).
\]
Furthermore,
\[
\E\bm{Z}_n=\bm{0},
\qquad
\E\|\bm{Z}_n\|_2^2
=
p-\|\bm{\gamma}\|_2^2,
\qquad
\E(\bm{\gamma}^{\top}\bm{Z}_n)^2
\le C\|\bm{\gamma}\|_2^2,
\]
\[
\Var(\|\bm{Z}_n\|_2^2)
\le
C\left(p+\frac{p^2}{n}\right).
\]
Thus
\[
\|\bm{Z}_n\|_2^2
=
p+O_p\left(\sqrt p+\frac p{\sqrt n}\right),
\qquad
\bm{\gamma}^{\top}\bm{Z}_n
=
O_p(\|\bm{\gamma}\|_2).
\]
Consequently,
\[
T_{\rm sum}
=
n\|\bm{g}_n^{(1),\rm or}\|_2^2
=
\mathcal S_n
+
p
+
O_p\left(\sqrt n\|\bm{\gamma}\|_2+\sqrt p+\frac p{\sqrt n}\right),
\]
and
\[
\frac{\mathcal S_n}{\sqrt p+p/\sqrt n}\to\infty
\quad\Longrightarrow\quad
\frac{T_{\rm sum}-p}{\sqrt{2p}}\xrightarrow{p}\infty,
\quad
P_{\rm sum}\xrightarrow{p}0.
\]

For the maximum statistic,
\[
\|\bm{g}_n^{(1),\rm or}-\bm{\gamma}\|_\infty
=
O_p\left(\sqrt{\frac{\log p}{n}}\right).
\]
Hence
\[
T_{\max}
=
n\|\bm{g}_n^{(1),\rm or}\|_\infty^2-2\log p+\log\log p
\ge
n\left(\|\bm{\gamma}\|_\infty
-
O_p\left(\sqrt{\frac{\log p}{n}}\right)\right)^2
-2\log p+\log\log p.
\]
Therefore
\[
\frac{\mathcal M_n}{\log p}\to\infty
\quad\Longrightarrow\quad
T_{\max}\xrightarrow{p}\infty,
\quad
0\le P_{\max}=1-F_G(T_{\max})
\le \pi^{-1/2}\exp(-T_{\max}/2)\xrightarrow{p}0.
\]

By Proposition~\ref{prop:h1-hr-behavior},
\[
\|\hat{\bm{g}}_{n}^{(1)}-\bm{g}_{n}^{(1),\rm or}\|_2
=O_p(\mathfrak R_{2,n}^{(1)}),
\qquad
\|\bm{g}_{n}^{(1),\rm or}\|_2
=
O_p\left(\|\bm{\gamma}\|_2+\sqrt{\frac pn}\right).
\]
Thus
\[
|\hat T_{\rm sum}-T_{\rm sum}|
=
n\left|
\|\hat{\bm{g}}_{n}^{(1)}\|_2^2
-
\|\bm{g}_{n}^{(1),\rm or}\|_2^2
\right|
\]
\[
\le
2n\|\bm{g}_{n}^{(1),\rm or}\|_2
\|\hat{\bm{g}}_{n}^{(1)}-\bm{g}_{n}^{(1),\rm or}\|_2
+
n\|\hat{\bm{g}}_{n}^{(1)}-\bm{g}_{n}^{(1),\rm or}\|_2^2
\]
\[
=
O_p\left(
n\left[
\left(\|\bm{\gamma}\|_2+\sqrt{\frac pn}\right)\mathfrak R_{2,n}^{(1)}
+
\{\mathfrak R_{2,n}^{(1)}\}^2
\right]\right).
\]
The condition
\[
\frac{n\{(\|\bm{\gamma}\|_2+\sqrt{p/n})\mathfrak R_{2,n}^{(1)}+[\mathfrak R_{2,n}^{(1)}]^2\}}{\mathcal S_n}\to0
\]
gives
\[
\hat T_{\rm sum}-p
=
\mathcal S_n
+
O_p\left(\sqrt n\|\bm{\gamma}\|_2+\sqrt p+\frac p{\sqrt n}\right)
+
o_p(\mathcal S_n),
\]
and hence \(\hat P_{\rm sum}\xrightarrow{p}0\).

Similarly,
\[
\|\bm{g}_{n}^{(1),\rm or}\|_\infty
=
O_p\left(\|\bm{\gamma}\|_\infty+\sqrt{\frac{\log p}{n}}\right),
\]
\[
|\hat T_{\max}-T_{\max}|
\le
2n\|\bm{g}_{n}^{(1),\rm or}\|_\infty
\|\hat{\bm{g}}_{n}^{(1)}-\bm{g}_{n}^{(1),\rm or}\|_\infty
+
n\|\hat{\bm{g}}_{n}^{(1)}-\bm{g}_{n}^{(1),\rm or}\|_\infty^2
\]
\[
=
O_p\left(
n\left[
\left(\|\bm{\gamma}\|_\infty+\sqrt{\frac{\log p}{n}}\right)\mathfrak R_{\infty,n}^{(1)}
+
\{\mathfrak R_{\infty,n}^{(1)}\}^2
\right]\right).
\]
The condition
\[
\frac{n\{(\|\bm{\gamma}\|_\infty+\sqrt{\log p/n})\mathfrak R_{\infty,n}^{(1)}+[\mathfrak R_{\infty,n}^{(1)}]^2\}}{\mathcal M_n}\to0
\]
gives
\[
\hat T_{\max}\xrightarrow{p}\infty,
\qquad
\hat P_{\max}\xrightarrow{p}0.
\]
If \(P_a\xrightarrow{p}0\), \(P_b\le 1-\eta\) with probability tending to one for some \(\eta\in(0,1)\), then
\[
\frac12\tan\{\pi(1/2-P_a)\}
+
\frac12\tan\{\pi(1/2-P_b)\}
\ge
\frac12\cot(\pi P_a)
+
\frac12\tan\{\pi(\eta-1/2)\}
\xrightarrow{p}\infty,
\]
\[
P_{\rm cau}
=
\frac12-\frac1\pi\arctan(T_{\rm cau})
\xrightarrow{p}0.
\]
The same display with \(P_a,P_b\) replaced by the plug-in p-values gives
\[
\hat P_{\rm cau}\xrightarrow{p}0.
\]
\end{proof}

\subsection{Proof of Theorem~\ref{thm:h1-sum-max-independence}}

\begin{proof}
Put
\[
S_A(\bm U_{0i})=s_n^{-1/2}\sum_{k\in A_n}U_{0ik},
\qquad
Z_{0i}=\frac{L_{0i}-m_{0,n}}{\sigma_{0,n}},
\qquad
Z_{1i}=\frac{L_{1i}-m_{0,n}}{\sigma_{1,n}}.
\]
Then
\[
Z_{1i}
=
\frac{\sigma_{0,n}}{\sigma_{1,n}}Z_{0i}
+
\frac{\delta_n}{\sigma_{1,n}}S_A(\bm U_{0i}),
\]
\[
\bm{W}_i=\sqrt p Z_{1i}\bm U_{0i},
\qquad
\bm{\gamma}=\E\bm{W}_i,
\qquad
\bm{a}_n=\sqrt n \bm{\gamma},
\qquad
\bm{S}_n=n^{-1/2}\sum_{i=1}^n(\bm{W}_i-\bm{\gamma}).
\]
For \(j=1,\ldots,p\),
\[
\gamma_j
=
\frac{\delta_n\mathbf 1(j\in A_n)}{\sqrt{s_np}\sigma_{1,n}},
\qquad
a_{nj}=\sqrt n\gamma_j=\mu_n\mathbf 1(j\in A_n).
\]
The calibration of \(\mu_n\) gives
\[
\frac{\delta_n^2}{p\sigma_{1,n}^2}
=
\frac{s_n\mu_n^2}{n},
\qquad
\frac{\sigma_{0,n}^2}{\sigma_{1,n}^2}
=
1-\frac{s_n\mu_n^2}{n},
\qquad
\|\bm{\gamma}\|_2^2=\frac{s_n\mu_n^2}{n}.
\]
Because
\[
\log s_n=\frac12\log p-\log\log p+\log\ell+o(1),
\]
\[
\mu_n=u_p(0)-u_{s_n}(\eta)
=(\sqrt2-1)\sqrt{\log p}+O\left(\frac{\log\log p}{\sqrt{\log p}}\right),
\]
\[
\mu_n^2=(3-2\sqrt2)\log p+O(\log\log p),
\qquad
\frac{s_n\mu_n^2}{\sqrt{2p}}\to
\theta_S=\frac{(3-2\sqrt2)\ell}{\sqrt2},
\]
\[
\frac{s_n\mu_n^2}{n}=O\left(\frac{\sqrt p}{n}\right)=o(1),
\qquad
\frac{\delta_n^2}{p\sigma_{0,n}^2}
=
\frac{s_n\mu_n^2}{n-s_n\mu_n^2}\to0.
\]
Let
\[
\bm{v}_A=s_n^{-1/2}\sum_{j\in A_n}\bm e_j,
\qquad
S_A(\bm U)=\bm{v}_A^\top\bm U,
\qquad
\bm U\sim\Unif(\Sphere^{p-1}).
\]
Then
\[
\E\{S_A(\bm U)^2\bm U\bm U^\top\}
=
\frac{\mathbf I_p+2\bm{v}_A\bm{v}_A^\top}{p(p+2)}.
\]
Hence
\[
\E(\bm{W}_i-\bm{\gamma})(\bm{W}_i-\bm{\gamma})^\top
=
\frac{\sigma_{0,n}^2}{\sigma_{1,n}^2}\mathbf I_p
+
\frac{\delta_n^2}{\sigma_{1,n}^2}
\frac{\mathbf I_p+2\bm{v}_A\bm{v}_A^\top}{p+2}
-
\bm{\gamma}\bm{\gamma}^\top,
\]
\[
\left\|
\E(\bm{W}_i-\bm{\gamma})(\bm{W}_i-\bm{\gamma})^\top-\mathbf I_p
\right\|_{op}
\le
C\frac{s_n\mu_n^2}{n}\to0,
\]
\[
\left|
\tr\E(\bm{W}_i-\bm{\gamma})(\bm{W}_i-\bm{\gamma})^\top-p
\right|
\le
\|\bm{\gamma}\|_2^2
=
O\left(\frac{\sqrt p}{n}\right).
\]
Define
\[
\bm{\xi}_i=\sqrt p Z_{0i}\bm U_{0i},
\qquad
\bm{S}_n^0=n^{-1/2}\sum_{i=1}^n\bm{\xi}_i,
\]
\[
\bm{D}_i
=
\bm{W}_i-\bm{\gamma}-\bm{\xi}_i
=
\left(\frac{\sigma_{0,n}}{\sigma_{1,n}}-1\right)\bm{\xi}_i
+
\frac{\delta_n}{\sigma_{1,n}}
\left\{\sqrt p S_A(\bm U_{0i})\bm U_{0i}-\frac{\bm{v}_A}{\sqrt p}\right\},
\]
\[
\bm{\Delta}_n=n^{-1/2}\sum_{i=1}^n\bm{D}_i,
\qquad
\bm{S}_n=\bm{S}_n^0+\bm{\Delta}_n.
\]
Since
\[
\E\left\|
\sqrt p S_A(\bm U_{0i})\bm U_{0i}-\frac{\bm{v}_A}{\sqrt p}
\right\|_2^2
=
p\E S_A(\bm U_{0i})^2-\frac{\|\bm{v}_A\|_2^2}{p}
=
1-p^{-1},
\]
\[
\E\|\bm{D}_i\|_2^2
\le
C\left(\frac{s_n\mu_n^2}{n}\right)^2p
+
C\frac{\delta_n^2}{\sigma_{1,n}^2}
\le
C\frac{s_np\mu_n^2}{n},
\]
\[
\|\bm{\Delta}_n\|_2
=
O_p\left(\sqrt{\frac{s_np\mu_n^2}{n}}\right),
\qquad
\|\bm{\Delta}_n\|_\infty
=
O_p\left(\sqrt{\frac{s_n\mu_n^2\log p}{n}}\right).
\]
Therefore, for each fixed \(y\in\mathbb R\),
\[
u_p(y)\|\bm{\Delta}_n\|_\infty
=
O_p\left(\sqrt{\frac{s_n\mu_n^2(\log p)^2}{n}}\right)
=
O_p\left(\frac{p^{1/4}\log p}{\sqrt n}\right)=o_p(1),
\]
\[
\frac{\|\bm{S}_n^0\|_2\|\bm{\Delta}_n\|_2}{\sqrt p}
=O_p\left(\sqrt{\frac{s_np\mu_n^2}{n}}\right)
=O_p\left(\frac{p^{3/4}}{\sqrt n}\right)=o_p(1),
\]
\[
\frac{\|\bm{a}_n\|_2\|\bm{\Delta}_n\|_2}{\sqrt p}
=O_p\left(\frac{s_n\mu_n^2}{\sqrt n}\right)
=O_p\left(\frac{\sqrt p}{\sqrt n}\right)=o_p(1),
\]
\[
\frac{\|\bm{\Delta}_n\|_2^2}{\sqrt p}
=O_p\left(\frac{s_n\sqrt p \mu_n^2}{n}\right)
=O_p\left(\frac pn\right)=o_p(1).
\]
The empirical correlation expansion gives
\[
\sqrt n \bm{g}_n^{(1),\rm or}
=
\bm{a}_n+\bm{S}_n+\bm{e}_n,
\]
\[
\|\bm{e}_n\|_\infty
=
O_p\left(\frac{\log p}{\sqrt n}+\mu_n\sqrt{\frac{\log p}{n}}\right),
\]
\[
\|\bm{e}_n\|_2
=
O_p\left(\frac{\sqrt p\log p}{\sqrt n}
+\|\bm{a}_n\|_2\sqrt{\frac{\log p}{n}}\right).
\]
The growth conditions imply
\[
u_p(y)\|\bm{e}_n\|_\infty=o_p(1),
\]
\[
\frac{\|\bm{S}_n^0\|_2\|\bm{e}_n\|_2+\|\bm{a}_n\|_2\|\bm{e}_n\|_2+\|\bm{e}_n\|_2^2}{\sqrt p}=o_p(1).
\]
Consequently,
\[
\frac{T_{\rm sum}-p}{\sqrt{2p}}
=
\frac{\|\bm{S}_n^0+\bm{a}_n\|_2^2-p}{\sqrt{2p}}+o_p(1),
\]
\[
T_{\max}
=
\|\bm{S}_n^0+\bm{a}_n\|_\infty^2-2\log p+\log\log p+o_p(1).
\]
Let
\[
H_n^0=\frac2n\sum_{1\le i<k\le n}\bm{\xi}_i^\top\bm{\xi}_k.
\]
Then
\[
\|\bm{S}_n^0\|_2^2-p
=
H_n^0+\frac1n\sum_{i=1}^n(\|\bm{\xi}_i\|_2^2-p),
\]
\[
\frac1{\sqrt p n}\sum_{i=1}^n(\|\bm{\xi}_i\|_2^2-p)
=
O_p\left(\sqrt{\frac pn}\right)=o_p(1),
\]
\[
\frac{2\bm{a}_n^\top\bm{S}_n^0}{\sqrt{2p}}
=
O_p\left(\sqrt{\frac{\|\bm{a}_n\|_2^2}{p}}\right)
=O_p(p^{-1/4})=o_p(1).
\]
Thus
\[
\frac{T_{\rm sum}-p}{\sqrt{2p}}
=
\frac{H_n^0}{\sqrt{2p}}+
\frac{\|\bm{a}_n\|_2^2}{\sqrt{2p}}+o_p(1)
=
\frac{H_n^0}{\sqrt{2p}}+\theta_S+o_p(1).
\]
For fixed \(y\in\mathbb R\), put
\[
B_{0j}(y)=\{|S_{n,j}^0|>u_p(y)\},
\qquad
B_{1j}(y)=\{|S_{n,j}^0+\mu_n|>u_p(y)\}.
\]
For fixed \(d_0,d_1\ge0\) and fixed disjoint coordinate sets
\[
J_0\subset A_n^c,
\qquad
J_1\subset A_n,
\qquad
|J_0|=d_0,
\qquad
|J_1|=d_1,
\]
define
\[
E_{J_0,J_1}(y)
=
\bigcap_{j\in J_0}B_{0j}(y)
\cap
\bigcap_{j\in J_1}B_{1j}(y).
\]
Put
\[
q_{0n}(y)=2\bar\Phi(u_p(y)),
\qquad
q_{1n}(y)=\bar\Phi(u_p(y)-\mu_n),
\qquad
J=J_0\cup J_1,
\qquad d=d_0+d_1.
\]
For the fixed block $J$, use the representation
\[
\bm U_{0i,J}=\sqrt{B_i}\,\bm V_i,
\qquad
\bm U_{0i,J^c}=\sqrt{1-B_i}\,\bm W_i,
\qquad
B_i\sim{\rm Beta}(d/2,(p-d)/2),
\]
where $\bm V_i\sim\Unif(\Sphere^{d-1})$, $\bm W_i\sim\Unif(\Sphere^{p-d-1})$, and $B_i$, $\bm V_i$, $\bm W_i$, $Z_{0i}$ are mutually independent.  With
\[
a_i=\sqrt p Z_{0i}\sqrt{B_i},
\qquad
b_i=\sqrt p Z_{0i}\sqrt{1-B_i},
\qquad
\mathcal G_{n,J}=\sigma\{(Z_{0i},B_i):1\le i\le n\},
\]
\[
\bm S_{n,J}^0=n^{-1/2}\sum_{i=1}^n a_i\bm V_i,
\qquad
Q_{n,J}=\frac{2}{n\sqrt{2p}}\sum_{1\le i<k\le n}b_i b_k\bm W_i^\top\bm W_k,
\]
\[
H_n^0/\sqrt{2p}=Q_{n,J}+\Delta_{n,J},
\qquad
\Delta_{n,J}=\frac{1}{n\sqrt{2p}}\sum_{1\le i<k\le n}2a_i a_k\bm V_i^\top\bm V_k.
\]
Conditional on $\mathcal G_{n,J}$, $\bm S_{n,J}^0$ and $Q_{n,J}$ are independent.  For every fixed $d_0,d_1$,
\[
\E\Delta_{n,J}^2\le C_d p^{-1},
\qquad
\sup_t\left|\Prob(Q_{n,J}\le t)-\Phi(t)\right|
\le C_d\left(n^{-1/2}+p^{-1/2}\right),
\]
\[
\left|
\frac{\Prob\{E_{J_0,J_1}(y)\mid \mathcal G_{n,J}\}}
{q_{0n}(y)^{d_0}q_{1n}(y)^{d_1}}
-1
\right|
\le
C_d\left\{
\frac{(\log p)^3}{\sqrt n}
+\frac{(\log p)^2}{p}
+\frac{[b_{n,p}^{(0)}]^2\log p\log n}{n}
+e^{-2\mu_n u_p(y)}
\right\}
\]
with probability tending to one.  Hence, with
\[
\varepsilon_{n,d}
=
C_d\left\{
\frac{(\log p)^3}{\sqrt n}
+\frac{(\log p)^2}{p}
+\frac{[b_{n,p}^{(0)}]^2\log p\log n}{n}
+e^{-2\mu_n u_p(y)}
+p^{-1/2}+n^{-1/2}
\right\},
\]
\[
\varepsilon_{n,d}\to0,
\qquad
\sup_{J_0,J_1}
\left|
\frac{\Prob\left(H_n^0/\sqrt{2p}\le x,E_{J_0,J_1}(y)\right)}
{\Phi(x)q_{0n}(y)^{d_0}q_{1n}(y)^{d_1}}
-1
\right|
\le \varepsilon_{n,d}+o(1).
\]
Equivalently,
\[
\Prob\left(
\frac{H_n^0}{\sqrt{2p}}\le x,
E_{J_0,J_1}(y)
\right)
=
\Phi(x)
q_{0n}(y)^{d_0}q_{1n}(y)^{d_1}
\{1+r_{n,d_0,d_1}\},
\]
\[
\sup_{J_0,J_1}|r_{n,d_0,d_1}|\le \varepsilon_{n,d}+o(1)\to0.
\]
The active lower tail is negligible:
\[
\bar\Phi(u_p(y)+\mu_n)
\le
\exp\{-2\mu_n u_p(y)\}\bar\Phi(u_p(y)-\mu_n),
\qquad
s_n\bar\Phi(u_p(y)+\mu_n)\to0.
\]
Moreover,
\[
u_p(y)-\mu_n=u_{s_n}(\eta)+u_p(y)-u_p(0),
\]
\[
\{u_p(y)-\mu_n\}^2
=u_{s_n}(\eta)^2
+2u_{s_n}(\eta)\{u_p(y)-u_p(0)\}
+\{u_p(y)-u_p(0)\}^2,
\]
\[
2u_{s_n}(\eta)\{u_p(y)-u_p(0)\}
=
2u_{s_n}(\eta)\frac{y}{u_p(y)+u_p(0)}
\to \frac{y}{\sqrt2},
\qquad
\{u_p(y)-u_p(0)\}^2\to0.
\]
For every $t\to\infty$,
\[
\bar\Phi(t)=\frac{1}{t\sqrt{2\pi}}\exp(-t^2/2)\{1+O(t^{-2})\}.
\]
Thus
\[
s_n\bar\Phi(u_p(y)-\mu_n)
=
\frac{s_n}{\{u_p(y)-\mu_n\}\sqrt{2\pi}}
\exp\left[-\frac12\{u_p(y)-\mu_n\}^2\right]\{1+O((\log p)^{-1})\},
\]
\[
s_n\bar\Phi(u_p(y)-\mu_n)
\to
\lambda_1(y;\eta)
=
\frac1{2\sqrt\pi}
\exp\left(-\frac\eta2-\frac{y}{2\sqrt2}\right),
\]
\[
(p-s_n)2\bar\Phi(u_p(y))
\to
\lambda_0(y)=\pi^{-1/2}e^{-y/2}.
\]
Set
\[
A_n(x)=\left\{\frac{H_n^0}{\sqrt{2p}}\le x\right\},
\qquad
L_n(y)=\left\{\max_{1\le j\le p}|S_{n,j}^0+a_{nj}|\le u_p(y)\right\}.
\]
For fixed \(m\ge1\), Bonferroni's inequalities give
\[
\sum_{d=0}^{2m+1}(-1)^d\sum_{d_0+d_1=d}C_{n,d_0,d_1}(x,y)
\le
\Prob\{A_n(x)\cap L_n(y)\}
\le
\sum_{d=0}^{2m}(-1)^d\sum_{d_0+d_1=d}C_{n,d_0,d_1}(x,y),
\]
where
\[
C_{n,d_0,d_1}(x,y)
=
\sum_{\substack{J_0\subset A_n^c, |J_0|=d_0}}
\sum_{\substack{J_1\subset A_n, |J_1|=d_1}}
\Prob\{A_n(x)\cap E_{J_0,J_1}(y)\}.
\]
For fixed \(d_0,d_1\),
\[
C_{n,d_0,d_1}(x,y)
\to
\Phi(x)
\frac{\lambda_0(y)^{d_0}}{d_0!}
\frac{\lambda_1(y;\eta)^{d_1}}{d_1!}.
\]
Therefore,
\[
\Prob\{A_n(x)\cap L_n(y)\}
\to
\Phi(x)
\sum_{d_0=0}^\infty\sum_{d_1=0}^\infty
(-1)^{d_0+d_1}
\frac{\lambda_0(y)^{d_0}}{d_0!}
\frac{\lambda_1(y;\eta)^{d_1}}{d_1!},
\]
\[
\Prob\{A_n(x)\cap L_n(y)\}
\to
\Phi(x)\exp\{-\lambda_0(y)-\lambda_1(y;\eta)\}.
\]
Thus
\[
\left(
\frac{H_n^0}{\sqrt{2p}},
\|\bm{S}_n^0+\bm{a}_n\|_\infty^2-2\log p+\log\log p
\right)
\xrightarrow{d}(Z,G_\eta),
\qquad
Z\perpn G_\eta.
\]
The reductions above imply
\[
\left(
\frac{T_{\rm sum}-p}{\sqrt{2p}},
T_{\max}
\right)
\xrightarrow{d}(Z+\theta_S,G_\eta),
\qquad
Z\perpn G_\eta.
\]
By Proposition~\ref{prop:h1-hr-behavior},
\[
|\hat T_{\rm sum}-T_{\rm sum}|
\le
2n\|\bm{g}_n^{(1),\rm or}\|_2
\|\hat{\bm g}_n^{(1)}-\bm{g}_n^{(1),\rm or}\|_2
+n\|\hat{\bm g}_n^{(1)}-\bm{g}_n^{(1),\rm or}\|_2^2,
\]
\[
\frac{|\hat T_{\rm sum}-T_{\rm sum}|}{\sqrt p}
=
O_p\left(
\frac{n\{(\|\bm{\gamma}\|_2+\sqrt{p/n})\mathfrak R_{2,n}^{(1)}+[\mathfrak R_{2,n}^{(1)}]^2\}}{\sqrt p}
\right)=o_p(1),
\]
\[
|\hat T_{\max}-T_{\max}|
\le
2n\|\bm{g}_n^{(1),\rm or}\|_\infty
\|\hat{\bm g}_n^{(1)}-\bm{g}_n^{(1),\rm or}\|_\infty
+n\|\hat{\bm g}_n^{(1)}-\bm{g}_n^{(1),\rm or}\|_\infty^2,
\]
\[
|\hat T_{\max}-T_{\max}|
=
O_p\left(
n\{(\|\bm{\gamma}\|_\infty+\sqrt{\log p/n})\mathfrak R_{\infty,n}^{(1)}+[\mathfrak R_{\infty,n}^{(1)}]^2\}
\right)=o_p(1).
\]
Hence
\[
\left(
\frac{\hat T_{\rm sum}-p}{\sqrt{2p}},
\hat T_{\max}
\right)
\xrightarrow{d}(Z+\theta_S,G_\eta),
\qquad
Z\perpn G_\eta.
\]
Let \(c_\alpha=F_G^{-1}(1-\alpha)\) and \(z_{1-\alpha}=\Phi^{-1}(1-\alpha)\). Since
\[
\Prob(Z+\theta_S>z_{1-\alpha})=1-\Phi(z_{1-\alpha}-\theta_S),
\]
\[
\Prob(G_\eta>c_\alpha)=1-F_\eta(c_\alpha),
\qquad
F_G(c_\alpha)=\exp\{-\lambda_0(c_\alpha)\}=1-\alpha,
\]
\[
F_\eta(c_\alpha)=F_G(c_\alpha)\exp\{-\lambda_1(c_\alpha;\eta)\}<1-\alpha,
\]
we obtain
\[
\alpha<1-\Phi(z_{1-\alpha}-\theta_S)<1,
\qquad
\alpha<1-F_\eta(c_\alpha)<1.
\]
\end{proof}

\section{Verification of Assumption~\ref{ass:oracle} for Common Elliptical Radial Laws}
\label{app:radial-verification}

This appendix verifies Assumption~\ref{ass:oracle}\textnormal{(ii)} for the radial laws used in the simulations and for several standard elliptical models.  Assumption~\ref{ass:oracle}\textnormal{(i)} is part of the elliptical construction in all examples below.  Write
\[
H_{n,p}=(\log p)^5+\log p\log n .
\]
The constants denoted by $C,c,C_b$ are positive and do not depend on $n$ or $p$; $C_b$ is taken sufficiently large.

Let $G_a\sim\Gamma(a,1)$ and
\[
W_a=\frac{\log G_a-\psi(a)}{\sqrt{\psi_1(a)}},
\]
where $\psi$ and $\psi_1$ are the digamma and trigamma functions.  Since
\[
\E\log G_a=\psi(a),
\qquad
\Var(\log G_a)=\psi_1(a),
\qquad
\psi_1(a)=a^{-1}+O(a^{-2}),
\]
$W_a$ has mean zero and variance one.  For $0\le \lambda\le a/2$,
\[
\E\exp\{\lambda(\log G_a-\psi(a))\}
=\exp\{\log\Gamma(a+\lambda)-\log\Gamma(a)-\lambda\psi(a)\},
\]
\[
\log\Gamma(a+\lambda)-\log\Gamma(a)-\lambda\psi(a)
=\int_0^\lambda(\lambda-t)\psi_1(a+t)\,dt
\le C\lambda^2/a,
\]
and similarly, for $0\le \lambda\le a/2$,
\[
\E\exp\{-\lambda(\log G_a-\psi(a))\}
=\exp\{\log\Gamma(a-\lambda)-\log\Gamma(a)+\lambda\psi(a)\}
\le \exp(C\lambda^2/a).
\]
Therefore, for all $x\ge0$,
\[
\Pr\{|\log G_a-\psi(a)|>x\}
\le C\exp\{-c\min(a x^2,a x)\}.
\]
With $x=t\sqrt{\psi_1(a)}$,
\[
\Pr(|W_a|>t)
\le C\exp\{-c\min(t^2,\sqrt a\,t)\},
\qquad t\ge0.
\]
Consequently, for every fixed $q>0$,
\[
\E|W_a|^q
=q\int_0^\infty t^{q-1}\Pr(|W_a|>t)\,dt
\le C_q.
\]
If $a\ge c p$, then the deterministic envelope
\[
 b^{(G)}_{n,p}=C_b\left\{\sqrt{\log(en)}+\frac{\log(en)}{\sqrt p}\right\}
\]
satisfies
\[
\Pr\left(\max_{1\le i\le n}|W_{a,i}|>b^{(G)}_{n,p}\right)
\le C n\exp\left[-c\min\{(b^{(G)}_{n,p})^2,\sqrt p\,b^{(G)}_{n,p}\}\right]
\to0,
\]
provided $C_b$ is large enough.  Moreover,
\[
\frac{(b^{(G)}_{n,p})^2H_{n,p}}{n}
\le
\frac{C\{\log(en)+p^{-1}(\log(en))^2\}H_{n,p}}{n}.
\]
Thus the gamma-log radial laws below satisfy Assumption~\ref{ass:oracle}\textnormal{(ii)} whenever
\[
\frac{\{\log(en)+p^{-1}(\log(en))^2\}H_{n,p}}{n}\to0.
\]

For the Gaussian elliptical model,
\[
\bm{Y}\sim N(0,\mathbf I_p),
\qquad
R^2\sim\chi_p^2=2G_{p/2},
\qquad
L=\log R=2^{-1}\log 2+2^{-1}\log G_{p/2}.
\]
Hence
\[
 m_L=2^{-1}\{\log2+\psi(p/2)\},
\qquad
 \sigma_L^2=4^{-1}\psi_1(p/2)\asymp p^{-1},
\qquad
Z=W_{p/2}.
\]
Therefore one may take
\[
 b_{n,p}=C_b\left\{\sqrt{\log(en)}+\frac{\log(en)}{\sqrt p}\right\}.
\]

For the Kotz-type power-exponential radial law used in the simulations,
\[
R_0=(2G_{p/(2\beta)})^{1/(2\beta)},
\qquad
R=c_{p,\beta}R_0,
\qquad
\E R^2=p,
\]
where $\beta>0$ is fixed and $c_{p,\beta}$ is deterministic.  Since deterministic scaling only adds a constant to $L=\log R$,
\[
L=\log c_{p,\beta}+\frac{1}{2\beta}\log2+\frac{1}{2\beta}\log G_{p/(2\beta)},
\]
\[
 \sigma_L^2=\frac{1}{4\beta^2}\psi_1\{p/(2\beta)\}\asymp p^{-1},
\qquad
 Z=W_{p/(2\beta)}.
\]
Thus the same envelope is admissible:
\[
 b_{n,p}=C_b\left\{\sqrt{\log(en)}+\frac{\log(en)}{\sqrt p}\right\}.
\]

For the bounded radial law,
\[
R^2=2pB,
\qquad
B\sim\mathrm{Beta}(p/2,p/2).
\]
Let $G_1,G_2\iid\Gamma(p/2,1)$ and $B=G_1/(G_1+G_2)$.  Then
\[
L=2^{-1}\log(2p)+2^{-1}\log B,
\qquad
m_L=2^{-1}\log(2p)+2^{-1}\{\psi(p/2)-\psi(p)\},
\]
\[
\sigma_L^2
=4^{-1}\{\psi_1(p/2)-\psi_1(p)\}
\asymp p^{-1}.
\]
Moreover,
\[
\log B-\E\log B
=\{\log G_1-\psi(p/2)\}-\{\log(G_1+G_2)-\psi(p)\},
\]
where $G_1+G_2\sim\Gamma(p,1)$.  Hence
\[
\Pr(|Z|>t)
\le
\Pr(|W_{p/2}|>c t)+\Pr(|W_p|>c t)
\le C\exp\{-c\min(t^2,\sqrt p\,t)\},
\]
and one may again take
\[
 b_{n,p}=C_b\left\{\sqrt{\log(en)}+\frac{\log(en)}{\sqrt p}\right\}.
\]

For the multivariate $t_\nu$ elliptical model with fixed $\nu>0$,
\[
\bm{Y}=\frac{\bm{Z}}{\sqrt{G_\nu/\nu}},
\qquad
\bm{Z}\sim N(0,\mathbf I_p),
\qquad
G_\nu\sim\chi_\nu^2,
\]
with $\bm{Z}$ independent of $G_\nu$.  Then
\[
L=2^{-1}\log\chi_p^2-2^{-1}\log(G_\nu/\nu),
\]
\[
\sigma_L^2
=4^{-1}\psi_1(p/2)+4^{-1}\psi_1(\nu/2)
\asymp 1.
\]
The first term has the gamma-log bound above.  The second term is a fixed-dimensional log-gamma variable.  Hence, for constants depending only on $\nu$,
\[
\Pr(|Z|>t)
\le C_\nu e^{-c_\nu t}+C\exp\{-c\min(t^2,\sqrt p\,t)\},
\qquad t\ge0.
\]
Therefore
\[
\sup_p\E|Z|^{8+\eta_Z}<\infty
\]
for every fixed $\eta_Z>0$, and
\[
 b_{n,p}=C_b\log(en)
\]
is admissible.  The truncation-rate requirement becomes
\[
\frac{(\log(en))^2H_{n,p}}{n}\to0.
\]

For a bounded non-degenerate Gaussian scale mixture,
\[
\bm{Y}=S\bm{Z},
\qquad
\bm{Z}\sim N(0,\mathbf I_p),
\qquad
0<s_-\le S\le s_+<\infty,
\qquad
\Var(\log S)>0,
\]
with $S$ independent of $\bm{Z}$,
\[
L=\log S+2^{-1}\log\chi_p^2,
\]
\[
\sigma_L^2=\Var(\log S)+4^{-1}\psi_1(p/2)\asymp1.
\]
Since $\log S$ is bounded and the Gaussian log-radius part satisfies the gamma-log bound,
\[
\sup_p\E|Z|^{8+\eta_Z}<\infty,
\qquad
 b_{n,p}=C_b\log(en)
\]
is admissible, again under
\[
\frac{(\log(en))^2H_{n,p}}{n}\to0.
\]
The two-point mixture used in the simulations, $\Pr(S=1)=0.9$ and $\Pr(S=3)=0.1$, is a special case.

The following table summarizes the resulting choices of $b_{n,p}$.

\begin{table}[H]
\centering
\caption{Admissible deterministic truncation envelopes for common radial laws.}
\footnotesize
\setlength{\tabcolsep}{4pt}
\renewcommand{\arraystretch}{1.15}
\begin{tabular}{lll}
\toprule
Radial law & $\sigma_L^2=\Var(\log R)$ & admissible $b_{n,p}$ \\
\midrule
Gaussian, $R^2\sim\chi_p^2$ & $4^{-1}\psi_1(p/2)\asymp p^{-1}$ & $C_b\{\sqrt{\log(en)}+p^{-1/2}\log(en)\}$ \\
Kotz, fixed $\beta$ & $(4\beta^2)^{-1}\psi_1\{p/(2\beta)\}\asymp p^{-1}$ & $C_b\{\sqrt{\log(en)}+p^{-1/2}\log(en)\}$ \\
Bounded radial, $R^2=2pB$ & $4^{-1}\{\psi_1(p/2)-\psi_1(p)\}\asymp p^{-1}$ & $C_b\{\sqrt{\log(en)}+p^{-1/2}\log(en)\}$ \\
$t_\nu$, fixed $\nu$ & $4^{-1}\psi_1(p/2)+4^{-1}\psi_1(\nu/2)\asymp1$ & $C_b\log(en)$ \\
Bounded scale mixture & $\Var(\log S)+4^{-1}\psi_1(p/2)\asymp1$ & $C_b\log(en)$ \\
\bottomrule
\end{tabular}
\end{table}

\section{Complete Simulation Tables}
\label{app:simulation-tables}

This appendix reports the full numerical simulation tables using the notation of Section~\ref{sec:simulation}.  The size table is generated under $H_0$ from $\bm{X}_i=\mathbf{\Sigma}^{1/2}R_{0i}\bm{U}_{0i}$ and includes all radial laws, dimensions and shape structures.  The power tables are generated under $H_1$ from $R_{1i}=R_{0i}\exp\{\delta_ns_A(\bm{U}_{0i})\}$ and are split by radial law.  The method labels S, M, C and W denote HR-Sum, HR-Max, HR-Cauchy and WL, respectively.  The case $\delta_n=0$ is the null model and is summarized in the size table.

\begin{table}[H]
\centering
\caption{Complete empirical sizes (\%) of all null settings with $n=200$. Columns are grouped by dimension $p$ and shape structure. The entries use the analytic null calibration without bootstrap correction.}
\label{tab:app-size-complete}
\renewcommand{\arraystretch}{1.05}
\setlength{\tabcolsep}{3.2pt}
\resizebox{0.8\textwidth}{!}{%
\begin{tabular}{lcccccccccccc}
\toprule
 & \multicolumn{3}{c}{$p=50$} & \multicolumn{3}{c}{$p=100$} & \multicolumn{3}{c}{$p=200$} & \multicolumn{3}{c}{$p=400$} \\
\cmidrule(lr){2-4}\cmidrule(lr){5-7}\cmidrule(lr){8-10}\cmidrule(lr){11-13}
Test & $\mathbf{I}_p$ & AR & SP & $\mathbf{I}_p$ & AR & SP & $\mathbf{I}_p$ & AR & SP & $\mathbf{I}_p$ & AR & SP \\
\midrule
\multicolumn{13}{c}{Gaussian} \\
HR-Sum & 4.4 & 5.2 & 4.8 & 5.4 & 3.0 & 4.2 & 5.2 & 4.4 & 5.8 & 4.6 & 4.2 & 4.0 \\
HR-Max & 3.6 & 4.0 & 4.4 & 4.0 & 4.6 & 3.8 & 4.2 & 6.4 & 6.0 & 4.4 & 4.4 & 5.0 \\
HR-Cauchy & 4.8 & 5.8 & 5.0 & 4.6 & 4.8 & 5.4 & 5.6 & 5.8 & 5.8 & 4.8 & 4.8 & 4.0 \\
WL & 5.8 & 5.4 & 4.6 & 5.4 & 4.4 & 5.2 & 5.8 & 5.8 & 6.6 & 6.0 & 3.8 & 5.8 \\
\addlinespace[1pt]
\multicolumn{13}{c}{Kotz $\beta=2$} \\
HR-Sum & 5.8 & 5.8 & 5.8 & 5.4 & 5.8 & 5.2 & 6.2 & 5.2 & 5.0 & 4.4 & 5.6 & 4.8 \\
HR-Max & 4.4 & 4.6 & 4.0 & 4.4 & 4.8 & 3.6 & 4.6 & 5.0 & 3.4 & 5.2 & 3.2 & 4.8 \\
HR-Cauchy & 5.4 & 5.6 & 5.6 & 5.8 & 6.4 & 4.0 & 6.0 & 5.4 & 5.0 & 4.6 & 4.6 & 5.4 \\
WL & 7.0 & 3.6 & 5.6 & 6.0 & 8.0 & 9.2 & 5.0 & 4.2 & 4.6 & 4.6 & 7.0 & 4.2 \\
\addlinespace[1pt]
\multicolumn{13}{c}{Bounded radial} \\
HR-Sum & 5.8 & 3.8 & 6.0 & 5.0 & 5.8 & 4.6 & 5.4 & 5.0 & 5.0 & 4.0 & 3.6 & 4.4 \\
HR-Max & 4.4 & 5.0 & 5.6 & 4.4 & 5.6 & 4.8 & 5.4 & 4.8 & 4.8 & 4.6 & 4.6 & 3.2 \\
HR-Cauchy & 5.4 & 5.8 & 6.2 & 5.4 & 6.0 & 5.0 & 5.0 & 5.0 & 6.0 & 4.4 & 4.6 & 4.6 \\
WL & 6.2 & 5.8 & 6.0 & 7.0 & 4.4 & 5.4 & 4.6 & 6.2 & 7.6 & 6.4 & 4.2 & 5.4 \\
\addlinespace[1pt]
\multicolumn{13}{c}{Mixture normal} \\
HR-Sum & 5.2 & 4.2 & 7.0 & 6.4 & 6.8 & 6.6 & 7.0 & 6.6 & 7.4 & 9.4 & 7.4 & 8.4 \\
HR-Max & 5.4 & 4.8 & 4.4 & 3.8 & 3.0 & 4.8 & 5.0 & 6.2 & 4.0 & 3.0 & 4.8 & 3.4 \\
HR-Cauchy & 6.0 & 5.2 & 6.0 & 5.4 & 6.0 & 5.8 & 6.2 & 6.0 & 4.6 & 5.6 & 6.2 & 6.8 \\
WL & 5.6 & 4.4 & 5.2 & 4.8 & 6.8 & 4.8 & 5.6 & 4.2 & 3.8 & 6.0 & 4.2 & 4.8 \\
\addlinespace[1pt]
\multicolumn{13}{c}{$t_{10}$} \\
HR-Sum & 6.6 & 5.2 & 5.2 & 5.8 & 5.8 & 4.8 & 8.2 & 7.6 & 8.6 & 7.6 & 9.0 & 8.2 \\
HR-Max & 4.8 & 5.0 & 4.2 & 4.4 & 5.4 & 5.0 & 4.2 & 3.6 & 4.8 & 4.2 & 5.4 & 4.6 \\
HR-Cauchy & 5.2 & 5.6 & 4.6 & 6.0 & 5.6 & 5.4 & 5.8 & 6.2 & 6.4 & 6.8 & 7.6 & 7.0 \\
WL & 3.4 & 5.0 & 4.6 & 5.0 & 5.2 & 4.8 & 4.8 & 5.0 & 9.0 & 5.0 & 2.6 & 1.2 \\
\bottomrule
\end{tabular}}
\end{table}

The complete size table shows that the proposed max component is particularly stable under non-concentrated radial laws, while the sum component can be mildly liberal for the $p=400$ mixture-normal and $t_{10}$ settings when the purely analytic calibration is used.  This is a finite-sample effect of estimating many weak coordinate correlations under strong radial variability.  The Cauchy combination moderates this behavior and remains closer to the nominal level in most heavy-tailed cases.  These entries motivate the optional bootstrap mean--variance correction in Section~\ref{subsec:finite-sample-calibration}, but they do not change the main qualitative conclusion: the radial--directional tests maintain usable calibration across the five null radial laws and the three shape structures.  The table also separates two sources of difficulty.  Changing $\mathbf{\Sigma}$ affects the quality of the plug-in standardization, whereas changing the radial law affects the variability of the log-radius.  The fact that the empirical sizes remain broadly comparable across both dimensions supports the modular structure of the method: HR standardization handles the unknown affine geometry, and the correlation statistic handles the radial--directional diagnostic.

\subsection{Complete empirical power tables}
\label{app:complete-power-tables}
Tables in this subsection report the complete power results.  The method labels S, M, C and W denote HR-Sum, HR-Max, HR-Cauchy and WL, respectively.  The rows exclude $\delta_n=0$, which corresponds to the null model and is summarized in Table~\ref{tab:app-size-complete}.  The five radial laws produce different values of $\sigma_{0,n}^2=\Var(\log R_{0i})$, so the same numerical value of $\delta_n$ does not imply the same signal-to-noise ratio.  This is why the Gaussian, Kotz-type and bounded-radial baselines are easier than the mixture-normal and $t_{10}$ baselines at the same $p$ and $\delta_n$.

\begin{table}[p]
\centering
\caption{Complete empirical powers (\%) for the Gaussian baseline under \eqref{eq:alternative-radius}.}\label{tab:app-power-normal}
\begingroup
\scriptsize
\setlength{\tabcolsep}{1.5pt}
\renewcommand{\arraystretch}{0.72}
\begin{adjustbox}{width=\textwidth,max totalheight=0.82\textheight,center}
\begin{tabular}{@{}lll*{12}{r}@{}}
\toprule
 & & & \multicolumn{4}{c}{$p=100$} & \multicolumn{4}{c}{$p=200$} & \multicolumn{4}{c}{$p=400$} \\
\cmidrule(lr){4-7}\cmidrule(lr){8-11}\cmidrule(lr){12-15}
Shape & Active set & $\delta_n$ & S & M & C & W & S & M & C & W & S & M & C & W \\
\midrule
$\mathbf{I}_p$ & $A_{\rm sp}$ & 0.5 & 97 & 100 & 100 & 5 & 82 & 100 & 100 & 8 & 45 & 100 & 100 & 7 \\
 &  & 1 & 100 & 100 & 100 & 11 & 100 & 100 & 100 & 14 & 98 & 100 & 100 & 8 \\
 &  & 2 & 100 & 100 & 100 & 55 & 100 & 100 & 100 & 47 & 100 & 100 & 100 & 25 \\
 &  & 3 & 100 & 100 & 100 & 81 & 100 & 100 & 100 & 72 & 100 & 100 & 100 & 56 \\
\addlinespace[0.5pt]
 & $A_{0.2}$ & 0.5 & 95 & 42 & 94 & 8 & 78 & 19 & 73 & 11 & 42 & 7 & 33 & 8 \\
 &  & 1 & 100 & 92 & 100 & 16 & 100 & 45 & 100 & 12 & 97 & 16 & 93 & 9 \\
 &  & 2 & 100 & 100 & 100 & 54 & 100 & 64 & 100 & 41 & 100 & 20 & 100 & 26 \\
 &  & 3 & 100 & 100 & 100 & 80 & 100 & 71 & 100 & 72 & 100 & 19 & 100 & 56 \\
\addlinespace[0.5pt]
 & $A_{\rm all}$ & 0.5 & 97 & 19 & 93 & 5 & 83 & 9 & 75 & 11 & 49 & 5 & 36 & 6 \\
 &  & 1 & 100 & 33 & 100 & 12 & 100 & 13 & 100 & 13 & 96 & 5 & 92 & 7 \\
 &  & 2 & 100 & 47 & 100 & 51 & 100 & 17 & 100 & 42 & 100 & 8 & 100 & 27 \\
 &  & 3 & 100 & 48 & 100 & 80 & 100 & 17 & 100 & 74 & 100 & 6 & 100 & 57 \\
\addlinespace[1.5pt]
AR & $A_{\rm sp}$ & 0.5 & 98 & 100 & 100 & 4 & 82 & 100 & 100 & 8 & 42 & 100 & 100 & 9 \\
 &  & 1 & 100 & 100 & 100 & 11 & 100 & 100 & 100 & 10 & 98 & 100 & 100 & 10 \\
 &  & 2 & 100 & 100 & 100 & 50 & 100 & 100 & 100 & 41 & 100 & 100 & 100 & 38 \\
 &  & 3 & 100 & 100 & 100 & 78 & 100 & 100 & 100 & 71 & 100 & 100 & 100 & 60 \\
\addlinespace[0.5pt]
 & $A_{0.2}$ & 0.5 & 95 & 37 & 92 & 4 & 70 & 15 & 61 & 6 & 20 & 7 & 16 & 11 \\
 &  & 1 & 100 & 91 & 100 & 10 & 100 & 36 & 100 & 9 & 63 & 10 & 51 & 15 \\
 &  & 2 & 100 & 99 & 100 & 54 & 100 & 59 & 100 & 40 & 92 & 14 & 86 & 31 \\
 &  & 3 & 100 & 100 & 100 & 81 & 100 & 60 & 100 & 71 & 96 & 14 & 92 & 62 \\
\addlinespace[0.5pt]
 & $A_{\rm all}$ & 0.5 & 95 & 16 & 92 & 4 & 67 & 10 & 56 & 6 & 17 & 5 & 15 & 8 \\
 &  & 1 & 100 & 26 & 100 & 11 & 100 & 9 & 100 & 12 & 65 & 7 & 50 & 13 \\
 &  & 2 & 100 & 32 & 100 & 49 & 100 & 10 & 100 & 40 & 92 & 3 & 82 & 35 \\
 &  & 3 & 100 & 36 & 100 & 81 & 100 & 10 & 100 & 70 & 96 & 4 & 89 & 63 \\
\addlinespace[1.5pt]
SP & $A_{\rm sp}$ & 0.5 & 99 & 100 & 100 & 7 & 80 & 100 & 100 & 7 & 35 & 100 & 100 & 7 \\
 &  & 1 & 100 & 100 & 100 & 16 & 100 & 100 & 100 & 12 & 94 & 100 & 100 & 11 \\
 &  & 2 & 100 & 100 & 100 & 59 & 100 & 100 & 100 & 39 & 100 & 100 & 100 & 30 \\
 &  & 3 & 100 & 100 & 100 & 81 & 100 & 100 & 100 & 66 & 100 & 100 & 100 & 58 \\
\addlinespace[0.5pt]
 & $A_{0.2}$ & 0.5 & 98 & 40 & 95 & 7 & 85 & 18 & 82 & 7 & 53 & 8 & 44 & 7 \\
 &  & 1 & 100 & 94 & 100 & 15 & 100 & 49 & 100 & 9 & 99 & 16 & 97 & 11 \\
 &  & 2 & 100 & 100 & 100 & 57 & 100 & 71 & 100 & 39 & 100 & 24 & 100 & 31 \\
 &  & 3 & 100 & 100 & 100 & 79 & 100 & 73 & 100 & 69 & 100 & 29 & 100 & 58 \\
\addlinespace[0.5pt]
 & $A_{\rm all}$ & 0.5 & 97 & 18 & 94 & 5 & 84 & 11 & 73 & 8 & 55 & 6 & 44 & 7 \\
 &  & 1 & 100 & 35 & 100 & 12 & 100 & 18 & 100 & 11 & 98 & 7 & 97 & 9 \\
 &  & 2 & 100 & 54 & 100 & 57 & 100 & 20 & 100 & 40 & 100 & 9 & 100 & 28 \\
 &  & 3 & 100 & 59 & 100 & 78 & 100 & 20 & 100 & 73 & 100 & 8 & 100 & 61 \\
\bottomrule
\end{tabular}
\end{adjustbox}
\endgroup
\end{table}

\begin{table}[p]
\centering
\caption{Complete empirical powers (\%) for the $t_{10}$ baseline under \eqref{eq:alternative-radius}.}\label{tab:app-power-t10}
\begingroup
\scriptsize
\setlength{\tabcolsep}{1.5pt}
\renewcommand{\arraystretch}{0.72}
\begin{adjustbox}{width=\textwidth,max totalheight=0.82\textheight,center}
\begin{tabular}{@{}lll*{12}{r}@{}}
\toprule
 & & & \multicolumn{4}{c}{$p=100$} & \multicolumn{4}{c}{$p=200$} & \multicolumn{4}{c}{$p=400$} \\
\cmidrule(lr){4-7}\cmidrule(lr){8-11}\cmidrule(lr){12-15}
Shape & Active set & $\delta_n$ & S & M & C & W & S & M & C & W & S & M & C & W \\
\midrule
$\mathbf{I}_p$ & $A_{\rm sp}$ & 0.5 & 18 & 26 & 28 & 4 & 13 & 9 & 14 & 5 & 10 & 4 & 8 & 2 \\
 &  & 1 & 61 & 98 & 97 & 5 & 26 & 56 & 58 & 4 & 14 & 19 & 21 & 2 \\
 &  & 2 & 100 & 100 & 100 & 11 & 84 & 100 & 100 & 7 & 35 & 92 & 92 & 3 \\
 &  & 3 & 100 & 100 & 100 & 28 & 100 & 100 & 100 & 12 & 65 & 100 & 100 & 5 \\
\addlinespace[0.5pt]
 & $A_{0.2}$ & 0.5 & 18 & 6 & 15 & 5 & 11 & 7 & 11 & 3 & 11 & 6 & 10 & 2 \\
 &  & 1 & 60 & 16 & 54 & 5 & 27 & 7 & 22 & 5 & 14 & 4 & 11 & 1 \\
 &  & 2 & 100 & 63 & 99 & 12 & 78 & 15 & 73 & 9 & 29 & 6 & 23 & 3 \\
 &  & 3 & 100 & 90 & 100 & 29 & 99 & 32 & 99 & 17 & 61 & 9 & 52 & 5 \\
 &  & 4 & -- & -- & -- & -- & -- & -- & -- & -- & 88 & 10 & 81 & 10 \\
 &  & 5 & -- & -- & -- & -- & -- & -- & -- & -- & 97 & 16 & 94 & 11 \\
\addlinespace[0.5pt]
 & $A_{\rm all}$ & 0.5 & 18 & 7 & 14 & 5 & 11 & 6 & 9 & 5 & 9 & 5 & 8 & 3 \\
 &  & 1 & 59 & 12 & 50 & 6 & 22 & 5 & 19 & 6 & 15 & 5 & 12 & 2 \\
 &  & 2 & 100 & 28 & 99 & 11 & 76 & 13 & 69 & 8 & 33 & 6 & 27 & 4 \\
 &  & 3 & 100 & 40 & 100 & 25 & 99 & 15 & 97 & 15 & 66 & 9 & 57 & 4 \\
 &  & 4 & -- & -- & -- & -- & -- & -- & -- & -- & 86 & 8 & 80 & 9 \\
 &  & 5 & -- & -- & -- & -- & -- & -- & -- & -- & 97 & 7 & 95 & 15 \\
\addlinespace[1.5pt]
AR & $A_{\rm sp}$ & 0.5 & 16 & 25 & 29 & 5 & 10 & 8 & 13 & 3 & 12 & 8 & 12 & 2 \\
 &  & 1 & 61 & 96 & 97 & 8 & 25 & 62 & 61 & 4 & 15 & 17 & 22 & 3 \\
 &  & 2 & 100 & 100 & 100 & 13 & 82 & 100 & 100 & 8 & 36 & 96 & 94 & 5 \\
 &  & 3 & 100 & 100 & 100 & 32 & 99 & 100 & 100 & 12 & 65 & 100 & 100 & 6 \\
\addlinespace[0.5pt]
 & $A_{0.2}$ & 0.5 & 14 & 6 & 13 & 5 & 10 & 6 & 8 & 3 & 10 & 3 & 7 & 2 \\
 &  & 1 & 54 & 16 & 47 & 6 & 21 & 6 & 15 & 5 & 12 & 4 & 10 & 2 \\
 &  & 2 & 100 & 62 & 100 & 19 & 72 & 17 & 64 & 8 & 19 & 6 & 17 & 3 \\
 &  & 3 & 100 & 90 & 100 & 35 & 96 & 23 & 93 & 15 & 37 & 7 & 27 & 5 \\
 &  & 4 & -- & -- & -- & -- & -- & -- & -- & -- & 55 & 9 & 46 & 9 \\
 &  & 5 & -- & -- & -- & -- & -- & -- & -- & -- & 77 & 9 & 67 & 17 \\
\addlinespace[0.5pt]
 & $A_{\rm all}$ & 0.5 & 18 & 6 & 14 & 7 & 11 & 5 & 9 & 4 & 10 & 4 & 9 & 3 \\
 &  & 1 & 56 & 13 & 49 & 7 & 21 & 7 & 19 & 4 & 14 & 6 & 11 & 3 \\
 &  & 2 & 100 & 26 & 100 & 14 & 68 & 10 & 57 & 7 & 18 & 5 & 15 & 2 \\
 &  & 3 & 100 & 33 & 100 & 33 & 98 & 9 & 95 & 12 & 37 & 5 & 26 & 6 \\
 &  & 4 & -- & -- & -- & -- & -- & -- & -- & -- & 58 & 7 & 46 & 9 \\
 &  & 5 & -- & -- & -- & -- & -- & -- & -- & -- & 75 & 6 & 64 & 15 \\
\addlinespace[1.5pt]
SP & $A_{\rm sp}$ & 0.5 & 19 & 25 & 29 & 7 & 10 & 8 & 12 & 4 & 12 & 6 & 12 & 2 \\
 &  & 1 & 59 & 97 & 98 & 7 & 25 & 59 & 58 & 5 & 16 & 16 & 20 & 3 \\
 &  & 2 & 100 & 100 & 100 & 13 & 79 & 100 & 100 & 10 & 29 & 95 & 95 & 3 \\
 &  & 3 & 100 & 100 & 100 & 30 & 99 & 100 & 100 & 16 & 56 & 100 & 100 & 2 \\
\addlinespace[0.5pt]
 & $A_{0.2}$ & 0.5 & 19 & 6 & 16 & 6 & 11 & 4 & 11 & 4 & 12 & 6 & 11 & 2 \\
 &  & 1 & 59 & 17 & 52 & 6 & 27 & 6 & 23 & 5 & 15 & 6 & 13 & 3 \\
 &  & 2 & 100 & 66 & 100 & 13 & 81 & 20 & 76 & 6 & 37 & 6 & 29 & 3 \\
 &  & 3 & 100 & 90 & 100 & 26 & 99 & 32 & 98 & 13 & 73 & 8 & 65 & 4 \\
 &  & 4 & -- & -- & -- & -- & -- & -- & -- & -- & 91 & 14 & 86 & 7 \\
 &  & 5 & -- & -- & -- & -- & -- & -- & -- & -- & 99 & 15 & 97 & 13 \\
\addlinespace[0.5pt]
 & $A_{\rm all}$ & 0.5 & 14 & 5 & 12 & 5 & 12 & 5 & 11 & 4 & 13 & 6 & 11 & 2 \\
 &  & 1 & 59 & 14 & 53 & 6 & 29 & 6 & 24 & 4 & 16 & 5 & 12 & 2 \\
 &  & 2 & 100 & 30 & 100 & 12 & 83 & 10 & 77 & 7 & 36 & 6 & 28 & 2 \\
 &  & 3 & 100 & 44 & 100 & 27 & 100 & 15 & 99 & 14 & 78 & 9 & 70 & 5 \\
 &  & 4 & -- & -- & -- & -- & -- & -- & -- & -- & 93 & 9 & 88 & 9 \\
 &  & 5 & -- & -- & -- & -- & -- & -- & -- & -- & 98 & 10 & 97 & 11 \\
\bottomrule
\end{tabular}
\end{adjustbox}
\endgroup
\end{table}

\begin{table}[p]
\centering
\caption{Complete empirical powers (\%) for the Kotz-type baseline with $\beta=2$ under \eqref{eq:alternative-radius}.}\label{tab:app-power-kotz-beta2}
\begingroup
\scriptsize
\setlength{\tabcolsep}{1.5pt}
\renewcommand{\arraystretch}{0.72}
\begin{adjustbox}{width=\textwidth,max totalheight=0.82\textheight,center}
\begin{tabular}{@{}lll*{12}{r}@{}}
\toprule
 & & & \multicolumn{4}{c}{$p=100$} & \multicolumn{4}{c}{$p=200$} & \multicolumn{4}{c}{$p=400$} \\
\cmidrule(lr){4-7}\cmidrule(lr){8-11}\cmidrule(lr){12-15}
Shape & Active set & $\delta_n$ & S & M & C & W & S & M & C & W & S & M & C & W \\
\midrule
$\mathbf{I}_p$ & $A_{\rm sp}$ & 0.5 & 100 & 100 & 100 & 6 & 99 & 100 & 100 & 3 & 77 & 100 & 100 & 6 \\
 &  & 1 & 100 & 100 & 100 & 11 & 100 & 100 & 100 & 6 & 100 & 100 & 100 & 5 \\
 &  & 2 & 100 & 100 & 100 & 54 & 100 & 100 & 100 & 36 & 100 & 100 & 100 & 20 \\
 &  & 3 & 100 & 100 & 100 & 83 & 100 & 100 & 100 & 70 & 100 & 100 & 100 & 54 \\
\addlinespace[0.5pt]
 & $A_{0.2}$ & 0.5 & 100 & 76 & 100 & 4 & 97 & 26 & 96 & 4 & 72 & 12 & 64 & 4 \\
 &  & 1 & 100 & 99 & 100 & 13 & 100 & 56 & 100 & 6 & 100 & 15 & 100 & 5 \\
 &  & 2 & 100 & 100 & 100 & 55 & 100 & 73 & 100 & 30 & 100 & 24 & 100 & 20 \\
 &  & 3 & 100 & 100 & 100 & 83 & 100 & 77 & 100 & 68 & 100 & 24 & 100 & 54 \\
 &  & 4 & -- & -- & -- & -- & -- & -- & -- & -- & 100 & 22 & 100 & 75 \\
 &  & 5 & -- & -- & -- & -- & -- & -- & -- & -- & 100 & 22 & 100 & 83 \\
\addlinespace[0.5pt]
 & $A_{\rm all}$ & 0.5 & 100 & 25 & 100 & 4 & 98 & 13 & 96 & 3 & 75 & 6 & 62 & 5 \\
 &  & 1 & 100 & 39 & 100 & 13 & 100 & 14 & 100 & 8 & 99 & 8 & 98 & 6 \\
 &  & 2 & 100 & 51 & 100 & 53 & 100 & 16 & 100 & 40 & 100 & 6 & 100 & 22 \\
 &  & 3 & 100 & 54 & 100 & 81 & 100 & 19 & 100 & 69 & 100 & 6 & 100 & 58 \\
 &  & 4 & -- & -- & -- & -- & -- & -- & -- & -- & 100 & 6 & 100 & 78 \\
 &  & 5 & -- & -- & -- & -- & -- & -- & -- & -- & 100 & 6 & 100 & 83 \\
\addlinespace[1.5pt]
AR & $A_{\rm sp}$ & 0.5 & 100 & 100 & 100 & 7 & 100 & 100 & 100 & 5 & 70 & 100 & 100 & 7 \\
 &  & 1 & 100 & 100 & 100 & 15 & 100 & 100 & 100 & 9 & 100 & 100 & 100 & 9 \\
 &  & 2 & 100 & 100 & 100 & 57 & 100 & 100 & 100 & 37 & 100 & 100 & 100 & 24 \\
 &  & 3 & 100 & 100 & 100 & 83 & 100 & 100 & 100 & 69 & 100 & 100 & 100 & 54 \\
\addlinespace[0.5pt]
 & $A_{0.2}$ & 0.5 & 100 & 71 & 100 & 7 & 96 & 24 & 93 & 4 & 33 & 8 & 26 & 8 \\
 &  & 1 & 100 & 99 & 100 & 14 & 100 & 50 & 100 & 9 & 82 & 10 & 72 & 5 \\
 &  & 2 & 100 & 100 & 100 & 54 & 100 & 65 & 100 & 39 & 97 & 14 & 93 & 24 \\
 &  & 3 & 100 & 100 & 100 & 84 & 100 & 66 & 100 & 69 & 98 & 17 & 97 & 56 \\
 &  & 4 & -- & -- & -- & -- & -- & -- & -- & -- & 98 & 13 & 95 & 76 \\
 &  & 5 & -- & -- & -- & -- & -- & -- & -- & -- & 100 & 16 & 97 & 87 \\
\addlinespace[0.5pt]
 & $A_{\rm all}$ & 0.5 & 100 & 20 & 100 & 5 & 95 & 10 & 90 & 3 & 34 & 5 & 22 & 7 \\
 &  & 1 & 100 & 31 & 100 & 15 & 100 & 10 & 100 & 9 & 79 & 2 & 66 & 8 \\
 &  & 2 & 100 & 34 & 100 & 61 & 100 & 11 & 100 & 40 & 96 & 3 & 87 & 28 \\
 &  & 3 & 100 & 38 & 100 & 84 & 100 & 14 & 100 & 68 & 98 & 2 & 91 & 56 \\
 &  & 4 & -- & -- & -- & -- & -- & -- & -- & -- & 99 & 4 & 93 & 77 \\
 &  & 5 & -- & -- & -- & -- & -- & -- & -- & -- & 99 & 3 & 94 & 86 \\
\addlinespace[1.5pt]
SP & $A_{\rm sp}$ & 0.5 & 100 & 100 & 100 & 5 & 99 & 100 & 100 & 6 & 60 & 100 & 100 & 3 \\
 &  & 1 & 100 & 100 & 100 & 10 & 100 & 100 & 100 & 9 & 99 & 100 & 100 & 5 \\
 &  & 2 & 100 & 100 & 100 & 56 & 100 & 100 & 100 & 35 & 100 & 100 & 100 & 28 \\
 &  & 3 & 100 & 100 & 100 & 83 & 100 & 100 & 100 & 69 & 100 & 100 & 100 & 53 \\
\addlinespace[0.5pt]
 & $A_{0.2}$ & 0.5 & 100 & 76 & 100 & 6 & 99 & 28 & 98 & 7 & 79 & 12 & 74 & 3 \\
 &  & 1 & 100 & 99 & 100 & 11 & 100 & 65 & 100 & 8 & 100 & 21 & 100 & 7 \\
 &  & 2 & 100 & 100 & 100 & 55 & 100 & 79 & 100 & 36 & 100 & 25 & 100 & 25 \\
 &  & 3 & 100 & 100 & 100 & 79 & 100 & 81 & 100 & 71 & 100 & 27 & 100 & 52 \\
 &  & 4 & -- & -- & -- & -- & -- & -- & -- & -- & 100 & 29 & 100 & 75 \\
 &  & 5 & -- & -- & -- & -- & -- & -- & -- & -- & 100 & 31 & 100 & 84 \\
\addlinespace[0.5pt]
 & $A_{\rm all}$ & 0.5 & 100 & 28 & 100 & 4 & 99 & 15 & 99 & 6 & 80 & 6 & 72 & 4 \\
 &  & 1 & 100 & 45 & 100 & 11 & 100 & 14 & 100 & 9 & 100 & 7 & 99 & 5 \\
 &  & 2 & 100 & 56 & 100 & 57 & 100 & 25 & 100 & 41 & 100 & 7 & 100 & 21 \\
 &  & 3 & 100 & 67 & 100 & 81 & 100 & 26 & 100 & 68 & 100 & 9 & 100 & 51 \\
 &  & 4 & -- & -- & -- & -- & -- & -- & -- & -- & 100 & 10 & 100 & 72 \\
 &  & 5 & -- & -- & -- & -- & -- & -- & -- & -- & 100 & 10 & 100 & 84 \\
\bottomrule
\end{tabular}
\end{adjustbox}
\endgroup
\end{table}

\begin{table}[p]
\centering
\caption{Complete empirical powers (\%) for the bounded-radial baseline under \eqref{eq:alternative-radius}.}\label{tab:app-power-bounded-radial}
\begingroup
\scriptsize
\setlength{\tabcolsep}{1.5pt}
\renewcommand{\arraystretch}{0.72}
\begin{adjustbox}{width=\textwidth,max totalheight=0.82\textheight,center}
\begin{tabular}{@{}lll*{12}{r}@{}}
\toprule
 & & & \multicolumn{4}{c}{$p=100$} & \multicolumn{4}{c}{$p=200$} & \multicolumn{4}{c}{$p=400$} \\
\cmidrule(lr){4-7}\cmidrule(lr){8-11}\cmidrule(lr){12-15}
Shape & Active set & $\delta_n$ & S & M & C & W & S & M & C & W & S & M & C & W \\
\midrule
$\mathbf{I}_p$ & $A_{\rm sp}$ & 0.5 & 100 & 100 & 100 & 4 & 99 & 100 & 100 & 5 & 73 & 100 & 100 & 5 \\
 &  & 1 & 100 & 100 & 100 & 12 & 100 & 100 & 100 & 10 & 100 & 100 & 100 & 4 \\
 &  & 2 & 100 & 100 & 100 & 54 & 100 & 100 & 100 & 42 & 100 & 100 & 100 & 22 \\
 &  & 3 & 100 & 100 & 100 & 80 & 100 & 100 & 100 & 71 & 100 & 100 & 100 & 52 \\
\addlinespace[0.5pt]
 & $A_{0.2}$ & 0.5 & 100 & 78 & 100 & 6 & 99 & 30 & 98 & 7 & 73 & 10 & 63 & 7 \\
 &  & 1 & 100 & 100 & 100 & 12 & 100 & 57 & 100 & 7 & 100 & 15 & 99 & 7 \\
 &  & 2 & 100 & 100 & 100 & 54 & 100 & 74 & 100 & 39 & 100 & 23 & 100 & 23 \\
 &  & 3 & 100 & 100 & 100 & 82 & 100 & 76 & 100 & 72 & 100 & 25 & 100 & 55 \\
 &  & 4 & -- & -- & -- & -- & -- & -- & -- & -- & 100 & 25 & 100 & 76 \\
 &  & 5 & -- & -- & -- & -- & -- & -- & -- & -- & 100 & 24 & 100 & 85 \\
\addlinespace[0.5pt]
 & $A_{\rm all}$ & 0.5 & 100 & 26 & 100 & 3 & 98 & 15 & 95 & 4 & 75 & 5 & 61 & 6 \\
 &  & 1 & 100 & 48 & 100 & 8 & 100 & 18 & 100 & 10 & 100 & 8 & 99 & 4 \\
 &  & 2 & 100 & 51 & 100 & 54 & 100 & 20 & 100 & 38 & 100 & 8 & 100 & 25 \\
 &  & 3 & 100 & 55 & 100 & 83 & 100 & 22 & 100 & 75 & 100 & 7 & 100 & 51 \\
 &  & 4 & -- & -- & -- & -- & -- & -- & -- & -- & 100 & 7 & 100 & 70 \\
 &  & 5 & -- & -- & -- & -- & -- & -- & -- & -- & 100 & 7 & 100 & 82 \\
\addlinespace[1.5pt]
AR & $A_{\rm sp}$ & 0.5 & 100 & 100 & 100 & 5 & 99 & 100 & 100 & 6 & 75 & 100 & 100 & 5 \\
 &  & 1 & 100 & 100 & 100 & 12 & 100 & 100 & 100 & 12 & 100 & 100 & 100 & 5 \\
 &  & 2 & 100 & 100 & 100 & 60 & 100 & 100 & 100 & 43 & 100 & 100 & 100 & 19 \\
 &  & 3 & 100 & 100 & 100 & 82 & 100 & 100 & 100 & 70 & 100 & 100 & 100 & 52 \\
\addlinespace[0.5pt]
 & $A_{0.2}$ & 0.5 & 100 & 76 & 100 & 8 & 94 & 27 & 93 & 9 & 31 & 8 & 27 & 5 \\
 &  & 1 & 100 & 99 & 100 & 14 & 100 & 52 & 100 & 9 & 82 & 11 & 71 & 6 \\
 &  & 2 & 100 & 100 & 100 & 59 & 100 & 63 & 100 & 37 & 98 & 14 & 93 & 21 \\
 &  & 3 & 100 & 100 & 100 & 85 & 100 & 69 & 100 & 70 & 97 & 12 & 94 & 54 \\
 &  & 4 & -- & -- & -- & -- & -- & -- & -- & -- & 99 & 16 & 96 & 80 \\
 &  & 5 & -- & -- & -- & -- & -- & -- & -- & -- & 98 & 13 & 97 & 88 \\
\addlinespace[0.5pt]
 & $A_{\rm all}$ & 0.5 & 100 & 24 & 100 & 7 & 96 & 9 & 93 & 7 & 35 & 5 & 25 & 5 \\
 &  & 1 & 100 & 28 & 100 & 17 & 100 & 9 & 100 & 9 & 82 & 3 & 67 & 6 \\
 &  & 2 & 100 & 33 & 100 & 61 & 100 & 11 & 100 & 40 & 97 & 3 & 88 & 21 \\
 &  & 3 & 100 & 40 & 100 & 81 & 100 & 10 & 100 & 71 & 97 & 2 & 91 & 54 \\
 &  & 4 & -- & -- & -- & -- & -- & -- & -- & -- & 99 & 3 & 93 & 75 \\
 &  & 5 & -- & -- & -- & -- & -- & -- & -- & -- & 99 & 5 & 94 & 83 \\
\addlinespace[1.5pt]
SP & $A_{\rm sp}$ & 0.5 & 100 & 100 & 100 & 5 & 99 & 100 & 100 & 9 & 60 & 100 & 100 & 5 \\
 &  & 1 & 100 & 100 & 100 & 13 & 100 & 100 & 100 & 12 & 99 & 100 & 100 & 8 \\
 &  & 2 & 100 & 100 & 100 & 55 & 100 & 100 & 100 & 43 & 100 & 100 & 100 & 20 \\
 &  & 3 & 100 & 100 & 100 & 80 & 100 & 100 & 100 & 68 & 100 & 100 & 100 & 56 \\
\addlinespace[0.5pt]
 & $A_{0.2}$ & 0.5 & 100 & 77 & 100 & 7 & 97 & 31 & 97 & 7 & 79 & 12 & 72 & 7 \\
 &  & 1 & 100 & 100 & 100 & 12 & 100 & 63 & 100 & 10 & 100 & 21 & 99 & 7 \\
 &  & 2 & 100 & 100 & 100 & 49 & 100 & 80 & 100 & 38 & 100 & 28 & 100 & 23 \\
 &  & 3 & 100 & 100 & 100 & 82 & 100 & 83 & 100 & 69 & 100 & 26 & 100 & 55 \\
 &  & 4 & -- & -- & -- & -- & -- & -- & -- & -- & 100 & 28 & 100 & 80 \\
 &  & 5 & -- & -- & -- & -- & -- & -- & -- & -- & 100 & 29 & 100 & 85 \\
\addlinespace[0.5pt]
 & $A_{\rm all}$ & 0.5 & 100 & 28 & 100 & 5 & 100 & 14 & 98 & 5 & 84 & 7 & 74 & 6 \\
 &  & 1 & 100 & 42 & 100 & 10 & 100 & 21 & 100 & 9 & 99 & 9 & 99 & 6 \\
 &  & 2 & 100 & 61 & 100 & 49 & 100 & 24 & 100 & 39 & 100 & 10 & 100 & 23 \\
 &  & 3 & 100 & 65 & 100 & 81 & 100 & 26 & 100 & 72 & 100 & 11 & 100 & 56 \\
 &  & 4 & -- & -- & -- & -- & -- & -- & -- & -- & 100 & 11 & 100 & 76 \\
 &  & 5 & -- & -- & -- & -- & -- & -- & -- & -- & 100 & 9 & 100 & 84 \\
\bottomrule
\end{tabular}
\end{adjustbox}
\endgroup
\end{table}

\begin{table}[p]
\centering
\caption{Complete empirical powers (\%) for the bounded Gaussian scale-mixture baseline under \eqref{eq:alternative-radius}.}\label{tab:app-power-mixnorm}
\begingroup
\scriptsize
\setlength{\tabcolsep}{1.5pt}
\renewcommand{\arraystretch}{0.72}
\begin{adjustbox}{width=\textwidth,max totalheight=0.82\textheight,center}
\begin{tabular}{@{}lll*{12}{r}@{}}
\toprule
 & & & \multicolumn{4}{c}{$p=100$} & \multicolumn{4}{c}{$p=200$} & \multicolumn{4}{c}{$p=400$} \\
\cmidrule(lr){4-7}\cmidrule(lr){8-11}\cmidrule(lr){12-15}
Shape & Active set & $\delta_n$ & S & M & C & W & S & M & C & W & S & M & C & W \\
\midrule
$\mathbf{I}_p$ & $A_{\rm sp}$ & 0.5 & 12 & 13 & 15 & 6 & 10 & 6 & 10 & 7 & 12 & 4 & 10 & 7 \\
 &  & 1 & 34 & 72 & 71 & 18 & 18 & 27 & 29 & 13 & 14 & 7 & 13 & 9 \\
 &  & 2 & 94 & 100 & 100 & 37 & 53 & 96 & 97 & 32 & 20 & 59 & 57 & 19 \\
 &  & 3 & 100 & 100 & 100 & 52 & 89 & 100 & 100 & 48 & 39 & 98 & 96 & 31 \\
\addlinespace[0.5pt]
 & $A_{0.2}$ & 0.5 & 12 & 6 & 11 & 7 & 9 & 5 & 8 & 11 & 10 & 6 & 9 & 5 \\
 &  & 1 & 35 & 15 & 34 & 17 & 16 & 6 & 12 & 15 & 11 & 5 & 10 & 10 \\
 &  & 2 & 93 & 37 & 90 & 38 & 53 & 10 & 44 & 32 & 21 & 7 & 17 & 16 \\
 &  & 3 & 100 & 72 & 100 & 54 & 89 & 18 & 83 & 43 & 33 & 5 & 26 & 32 \\
 &  & 4 & -- & -- & -- & -- & -- & -- & -- & -- & 59 & 10 & 52 & 42 \\
 &  & 5 & -- & -- & -- & -- & -- & -- & -- & -- & 84 & 9 & 77 & 55 \\
\addlinespace[0.5pt]
 & $A_{\rm all}$ & 0.5 & 11 & 6 & 10 & 6 & 10 & 5 & 8 & 8 & 12 & 5 & 9 & 4 \\
 &  & 1 & 33 & 9 & 28 & 16 & 19 & 6 & 15 & 17 & 7 & 4 & 5 & 8 \\
 &  & 2 & 94 & 18 & 92 & 39 & 46 & 6 & 40 & 27 & 21 & 6 & 18 & 19 \\
 &  & 3 & 100 & 29 & 99 & 55 & 89 & 12 & 83 & 47 & 40 & 7 & 32 & 30 \\
 &  & 4 & -- & -- & -- & -- & -- & -- & -- & -- & 61 & 7 & 53 & 42 \\
 &  & 5 & -- & -- & -- & -- & -- & -- & -- & -- & 82 & 8 & 74 & 50 \\
\addlinespace[1.5pt]
AR & $A_{\rm sp}$ & 0.5 & 12 & 10 & 14 & 6 & 11 & 4 & 9 & 6 & 10 & 5 & 11 & 7 \\
 &  & 1 & 39 & 71 & 71 & 17 & 18 & 24 & 27 & 12 & 12 & 8 & 13 & 7 \\
 &  & 2 & 95 & 100 & 100 & 43 & 49 & 95 & 95 & 28 & 24 & 57 & 58 & 25 \\
 &  & 3 & 100 & 100 & 100 & 57 & 87 & 100 & 100 & 42 & 41 & 98 & 97 & 35 \\
\addlinespace[0.5pt]
 & $A_{0.2}$ & 0.5 & 10 & 5 & 10 & 9 & 9 & 4 & 9 & 8 & 10 & 3 & 8 & 5 \\
 &  & 1 & 29 & 10 & 25 & 18 & 15 & 5 & 12 & 11 & 12 & 6 & 10 & 10 \\
 &  & 2 & 93 & 37 & 91 & 40 & 44 & 8 & 36 & 30 & 16 & 7 & 13 & 21 \\
 &  & 3 & 100 & 68 & 100 & 53 & 81 & 17 & 77 & 40 & 25 & 4 & 18 & 33 \\
 &  & 4 & -- & -- & -- & -- & -- & -- & -- & -- & 34 & 8 & 28 & 45 \\
 &  & 5 & -- & -- & -- & -- & -- & -- & -- & -- & 53 & 11 & 45 & 53 \\
\addlinespace[0.5pt]
 & $A_{\rm all}$ & 0.5 & 11 & 7 & 10 & 11 & 8 & 5 & 7 & 8 & 12 & 6 & 11 & 8 \\
 &  & 1 & 28 & 7 & 26 & 16 & 13 & 5 & 10 & 11 & 10 & 3 & 7 & 7 \\
 &  & 2 & 93 & 17 & 88 & 40 & 36 & 7 & 29 & 33 & 15 & 6 & 13 & 23 \\
 &  & 3 & 100 & 28 & 100 & 57 & 76 & 9 & 71 & 43 & 22 & 9 & 19 & 36 \\
 &  & 4 & -- & -- & -- & -- & -- & -- & -- & -- & 31 & 5 & 26 & 45 \\
 &  & 5 & -- & -- & -- & -- & -- & -- & -- & -- & 55 & 5 & 45 & 55 \\
\addlinespace[1.5pt]
SP & $A_{\rm sp}$ & 0.5 & 12 & 13 & 15 & 6 & 9 & 7 & 11 & 9 & 11 & 4 & 10 & 6 \\
 &  & 1 & 31 & 69 & 66 & 15 & 15 & 24 & 26 & 15 & 12 & 7 & 11 & 8 \\
 &  & 2 & 97 & 100 & 100 & 32 & 46 & 95 & 94 & 32 & 18 & 56 & 53 & 19 \\
 &  & 3 & 100 & 100 & 100 & 56 & 87 & 100 & 100 & 44 & 36 & 97 & 96 & 31 \\
\addlinespace[0.5pt]
 & $A_{0.2}$ & 0.5 & 11 & 7 & 11 & 5 & 10 & 3 & 6 & 8 & 12 & 6 & 12 & 6 \\
 &  & 1 & 37 & 11 & 31 & 12 & 17 & 6 & 15 & 14 & 13 & 4 & 10 & 7 \\
 &  & 2 & 96 & 40 & 91 & 39 & 52 & 11 & 47 & 28 & 24 & 5 & 20 & 18 \\
 &  & 3 & 100 & 68 & 100 & 55 & 91 & 21 & 87 & 44 & 46 & 7 & 37 & 32 \\
 &  & 4 & -- & -- & -- & -- & -- & -- & -- & -- & 73 & 9 & 64 & 44 \\
 &  & 5 & -- & -- & -- & -- & -- & -- & -- & -- & 90 & 11 & 85 & 53 \\
\addlinespace[0.5pt]
 & $A_{\rm all}$ & 0.5 & 11 & 6 & 10 & 5 & 11 & 5 & 9 & 8 & 10 & 4 & 8 & 7 \\
 &  & 1 & 30 & 9 & 27 & 14 & 18 & 5 & 15 & 13 & 13 & 4 & 10 & 8 \\
 &  & 2 & 94 & 19 & 92 & 37 & 53 & 9 & 47 & 29 & 25 & 7 & 21 & 18 \\
 &  & 3 & 100 & 34 & 100 & 51 & 90 & 11 & 84 & 42 & 47 & 7 & 35 & 32 \\
 &  & 4 & -- & -- & -- & -- & -- & -- & -- & -- & 71 & 8 & 63 & 40 \\
 &  & 5 & -- & -- & -- & -- & -- & -- & -- & -- & 88 & 7 & 80 & 46 \\
\bottomrule
\end{tabular}
\end{adjustbox}
\endgroup
\end{table}

The complete power tables confirm the main message of Table~\ref{tab:power-main}.  In sparse alternatives, M is usually the leading component, especially when $A=A_{\rm sp}$.  In moderately dense and fully dense alternatives, S dominates M because the coordinatewise correlations are individually weaker but accumulate across many active coordinates.  The Cauchy statistic is rarely far from the stronger of S and M and is therefore the most robust default choice.  Across radial laws, the proposed tests consistently exploit the targeted radial--directional signal more effectively than W, whose projection-based construction is not tailored to the dependence pattern in \eqref{eq:alternative-radius}.  The detailed rows further show that the shape structure changes the absolute power only moderately after HR standardization, whereas the active-set size and radial variability have the dominant effects.  This supports the interpretation of S and M as complementary diagnostics rather than competing versions of the same statistic.

\section{Additional Power Figures}
\label{app:power-figures}

Figures in this section give the single-radial-law power summaries.  Each figure averages over the three shape structures and displays the same three alternatives and three dimensions used in the power study.  The horizontal axis is the signal strength $\delta_n$ in \eqref{eq:alternative-radius}.  These plots are intended to complement the tables by showing how quickly each component reacts as radial--directional dependence increases.

\begin{figure}[H]
\centering
\wideportraitfigure{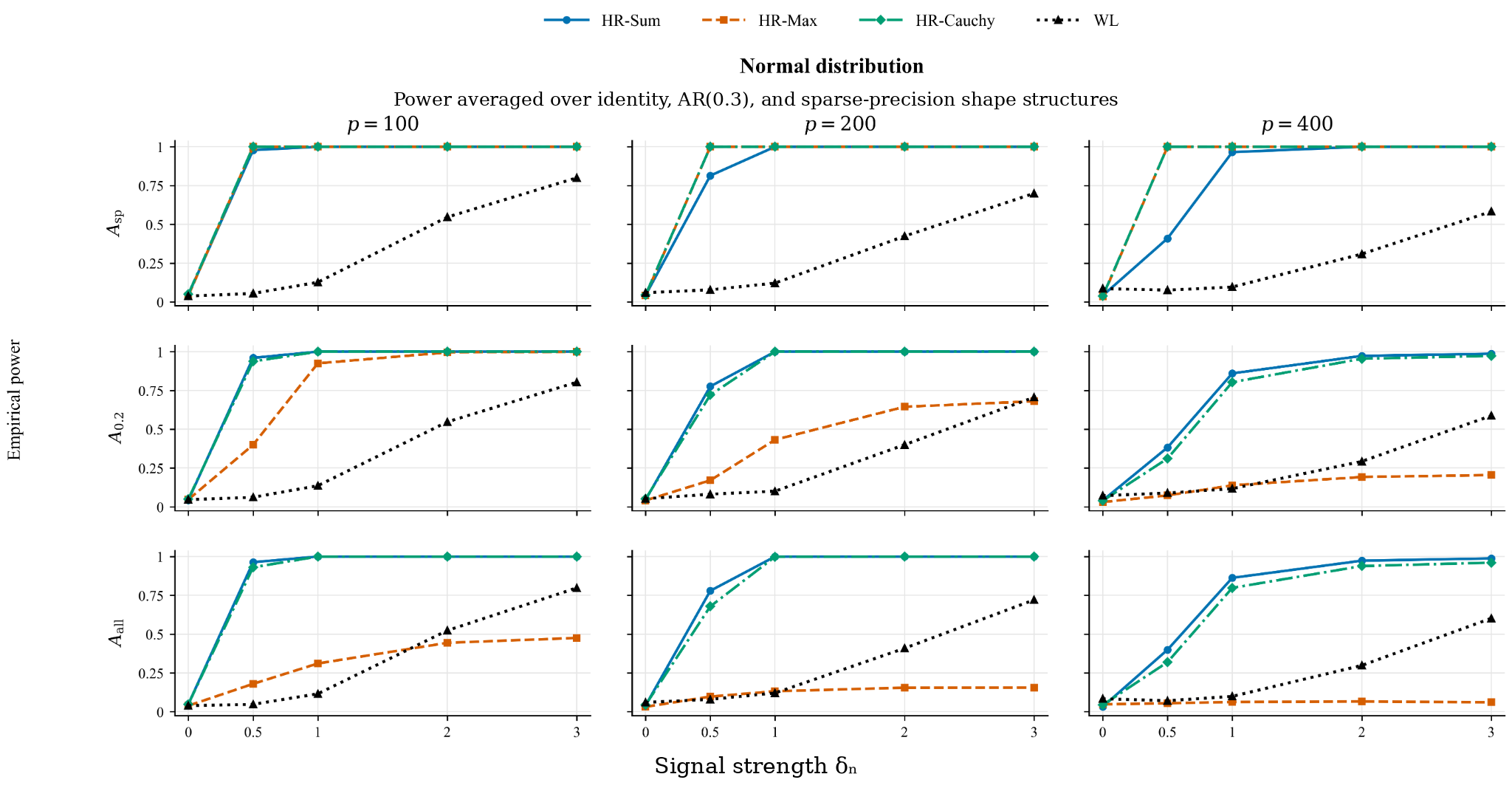}
\caption{Additional empirical power curves for the Gaussian baseline.  The concentrated log-radius makes the radial--directional signal strong, so S and C rise quickly for dense alternatives and M rises quickly for sparse alternatives.}
\label{fig:app-power-normal}
\end{figure}

\begin{figure}[H]
\centering
\wideportraitfigure{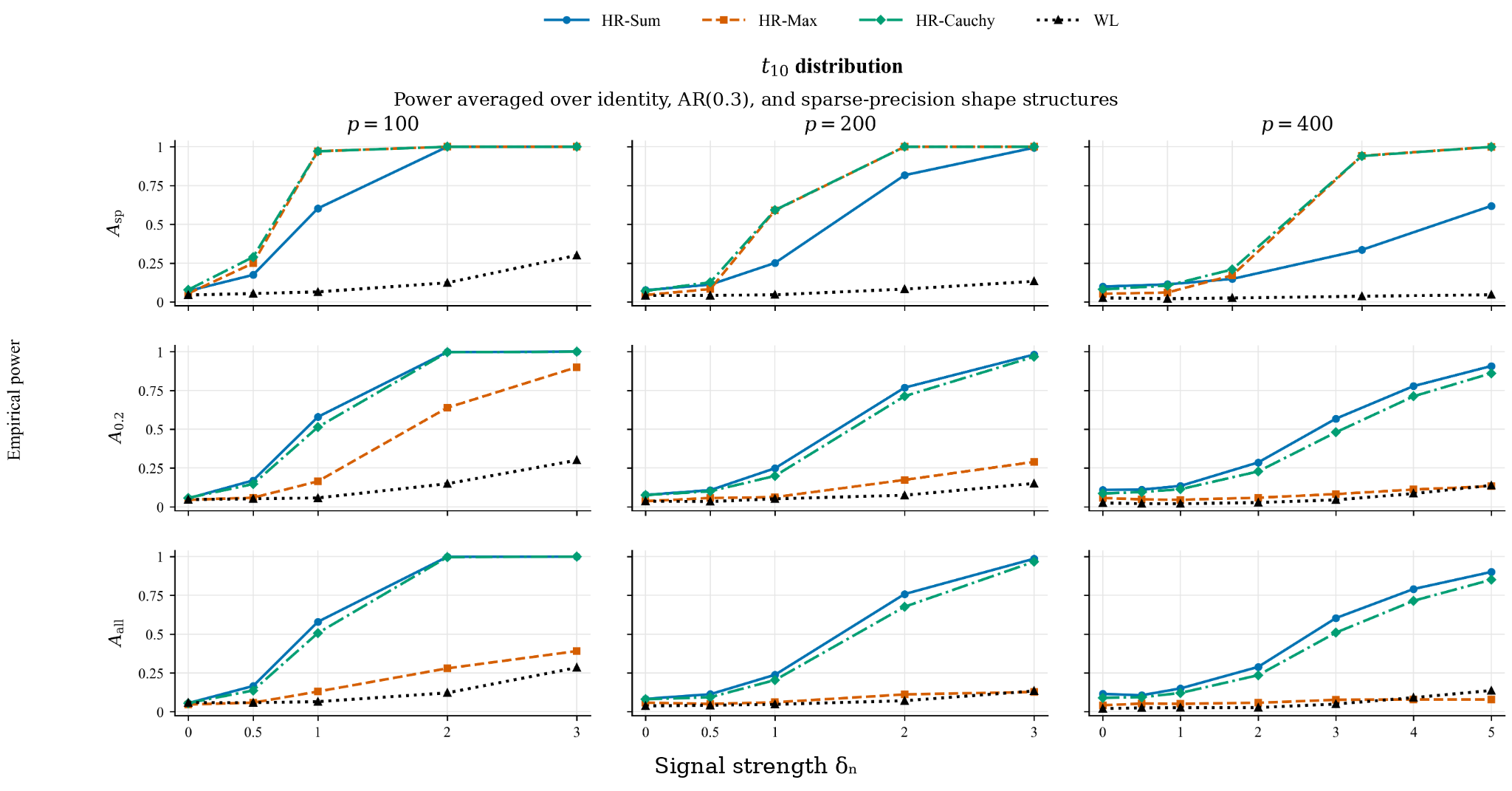}
\caption{Additional empirical power curves for the $t_{10}$ baseline.  Heavy radial variability makes this setting harder, but the same sparse-versus-dense pattern is visible as $\delta_n$ increases.}
\label{fig:app-power-t10}
\end{figure}

\begin{figure}[H]
\centering
\wideportraitfigure{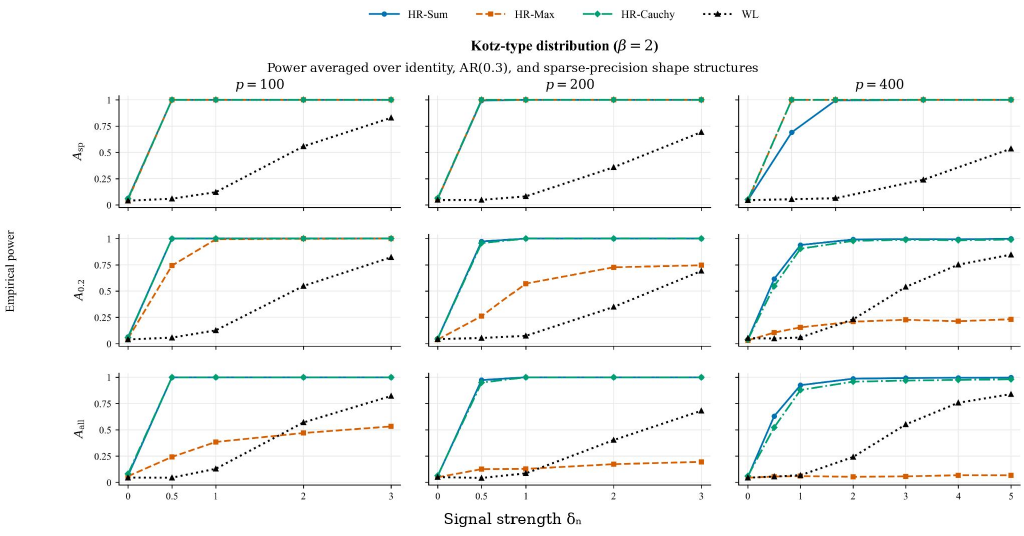}
\caption{Additional empirical power curves for the Kotz-type baseline with $\beta=2$.  This concentrated radial law behaves similarly to the Gaussian baseline after log-radius standardization.}
\label{fig:app-power-kotz}
\end{figure}

\begin{figure}[H]
\centering
\wideportraitfigure{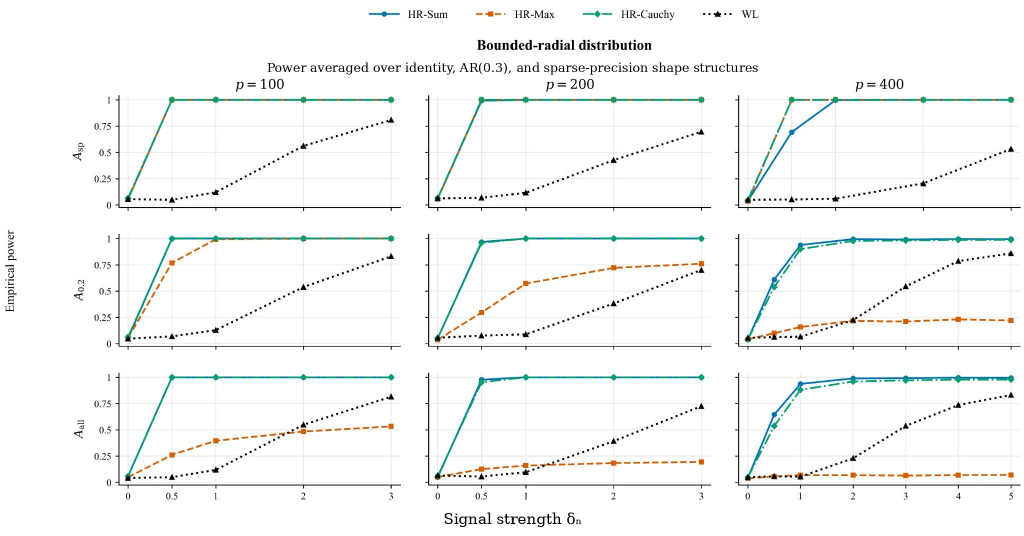}
\caption{Additional empirical power curves for the bounded-radial baseline.  The results show that the tests do not rely on Gaussian radial tails and remain powerful when the radius has bounded support.}
\label{fig:app-power-bounded}
\end{figure}

\begin{figure}[H]
\centering
\wideportraitfigure{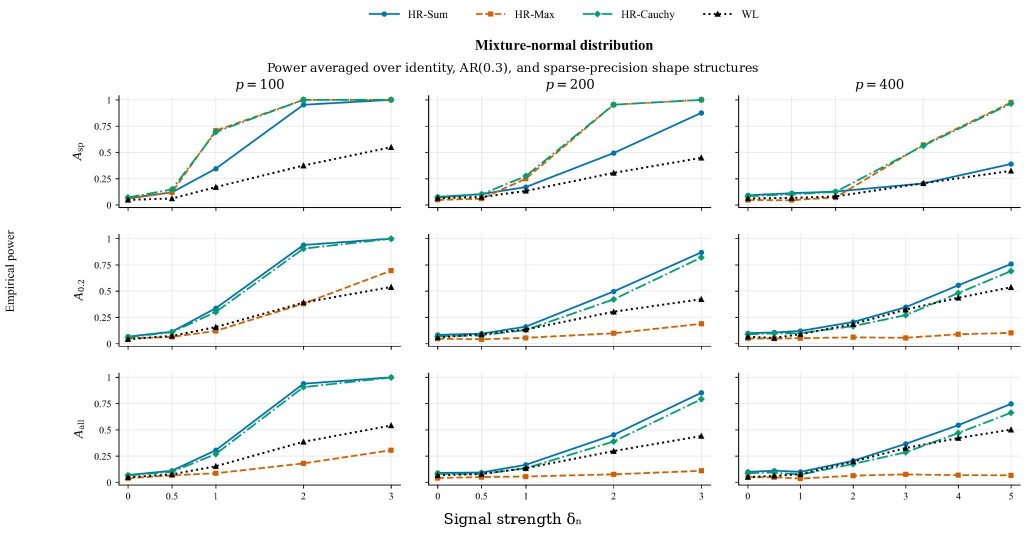}
\caption{Additional empirical power curves for the mixture-normal baseline.  The non-concentrated scale mixture reduces the signal-to-noise ratio, but C continues to adapt between the dense S component and the sparse M component.}
\label{fig:app-power-mixnorm}
\end{figure}

\setcounter{section}{4}
\section{Additional Application Results}
\label{app:application-results}

Table~\ref{tab:appendix-applications} gives selected supplementary application blocks used to support the discussion in Section~\ref{sec:application}.  The supplementary screening uses the same preprocessing principle as the main real-data analysis: variables are extracted within a fixed block, non-finite or constant columns are removed if necessary, and the remaining columns are standardized before the tests are computed.  For Arcene, the blocks are variance-screened feature sets from the pooled training and validation samples of the UCI mass-spectrometry data \citep{GuyonEtAl2004,UCIArcene2008}.  For cookie, the blocks are high-wavelength tail windows of the biscuit-dough NIR spectra \citep{BrownFearnVannucci2001,OsborneFearnMillerDouglas1984}.  For glass, the blocks are local EPXMA spectral channels from the archaeological glass data used in robust sparse PCA studies \citep{LembergeEtAl2000,HubertReynkensSchmittVerdonck2016}.  The table emphasizes blocks where the radial--directional tests provide additional local or sparse-versus-dense interpretation.

\begin{table}[!htbp]
\centering
\caption{Selected supplementary application blocks. The last two columns report the strongest coordinate-level radial--directional association in each block.}
\label{tab:appendix-applications}
\footnotesize
\setlength{\tabcolsep}{3.2pt}
\renewcommand{\arraystretch}{1.05}
\resizebox{\textwidth}{!}{%
\begin{tabular}{llrrrrrrrl}
\toprule
Data & Block & $n$ & $p$ & HR-Sum & HR-Max & HR-Cauchy & WL & Max axis & $n\hat\gamma^2_{\max}$ \\
\midrule
Gasoline & 900--998nm & 60 & 50 & 0.011 & 0.046 & 0.018 & 0.031 & -- & -- \\
Gasoline & 1050--1148nm & 60 & 50 & 0.180 & 0.008 & 0.015 & 0.976 & -- & -- \\
Gasoline & 1100--1198nm & 60 & 50 & 0.078 & 0.006 & 0.011 & $2.0\times10^{-4}$ & -- & -- \\
Arcene & top 300 & 200 & 300 & 0.002 & 0.002 & 0.002 & 0.293 & 5456 & 131.7 \\
Arcene & top 400 & 200 & 400 & 0.002 & 0.002 & 0.002 & 0.051 & 3111 & 73.0 \\
Arcene & top 500 & 200 & 500 & 0.094 & 0.012 & 0.021 & 0.076 & 7159 & 11.6 \\
Cookie & 2350--2498nm & 72 & 75 & 0.002 & 0.002 & 0.002 & 0.061 & 2350 & 30.8 \\
Cookie & 2360--2498nm & 72 & 70 & 0.002 & 0.002 & 0.002 & 0.224 & 2364 & 27.3 \\
Cookie & 2360--2498nm, no outliers & 70 & 70 & 0.002 & 0.002 & 0.002 & 0.130 & 2498 & 24.0 \\
Glass & 501--550 & 180 & 50 & 0.004 & 0.130 & 0.008 & 0.512 & 550 & 8.8 \\
\bottomrule
\end{tabular}}
\end{table}

\bibliographystyle{apa}
\bibliography{refe}

\end{document}